\documentclass[11pt]{article}
\usepackage{amssymb,latexsym,amsmath,amsbsy,amsthm,amsxtra,amsgen,dsfont,authblk}
\usepackage{tabularx}
\usepackage{geometry}
\usepackage[table]{xcolor}
\definecolor{LightGray}{gray}{0.93}

\newfont{\msbm}{msbm10 at 11pt}

\newcommand {\Z} {\mbox{\msbm Z}}

\newcommand {\N} {\mbox{\msbm N}}
\newcommand {\1} {\mathds{1}}

\newfont{\msbmsm}{msbm10 at 8pt}

\newcommand {\Zsm} {\mbox{\msbmsm Z}}

\newtheorem{Theo}{Theorem}
\newtheorem{Lemma}[Theo]{Lemma}
\newtheorem{Cor}[Theo]{Corollary}
\newtheorem{Prop}[Theo]{Proposition}

\newtheorem{Dfn}[Theo]{Definition}
\newtheorem{Rmk}[Theo]{Remark}

\def\eps{\varepsilon}

\usepackage{mathtools,color,bbm}

\DeclareMathOperator{\sgn}{sgn}

\newcommand{\Ib}  {{\mathbb I}}

\newcommand{\Nb}  {{\mathbb N}}

\newcommand{\Zb}  {{\mathbb Z}}

\newcommand{\dZ} {{{Z}_{N}^{(\delta)}}}

\newcommand{\yk} {{y_{k,N}^{}}}

\newcommand{\ind}{\mathbbm{1}}{}

\def\dd{\textnormal{d}}
\newcommand{\dom}{\mathcal{D}}
\newcommand{\idom}{{\bf X}}
\newcommand{\domi}{\mathcal{D}}

\begin{document}

\title{Two waves of adaptation: speciation induced by dormancy in a model with changing environment}
\author{Fernando Cordero\thanks{Institute of Mathematics, BOKU University.  Email: fernando.cordero@boku.ac.at}, Adri\'an Gonz\'alez Casanova\thanks{School of Mathematical and Statistical Sciences, Arizona State University. Email:\\ adrian.gonzalez.casanova.soberon@asu.edu}, and Jason Schweinsberg\thanks{Department of Mathematics, University of California San Diego. Email: jschweinsberg@ucsd.edu}}
\maketitle

\begin{abstract}
We consider a population model in which the season alternates between winter and summer, and individuals can acquire mutations either that are advantageous in the summer and disadvantageous in the winter, or vice versa.  Also, we assume that individuals in the population can either be active or dormant, and that individuals can move between these two states.  Dormant individuals do not reproduce but do not experience selective pressures.  We show that, under certain conditions, over time we see two waves of adaptation.  Some individuals repeatedly acquire mutations that are beneficial in the summer, while others repeatedly acquire mutations that are beneficial in the winter.   Individuals can survive the season during which they are less fit by entering a dormant state.  This result demonstrates that, for populations in fluctuating environments, dormancy has the potential to induce speciation.
\end{abstract}

{\small MSC: Primary 60J99; secondary 60J80, 92D25, 92D15

}

\section{Introduction}

Numerous factors influence the evolutionary trajectory of populations. In this paper, we focus on two pivotal elements and explore their interplay: the population's ability to sustain a dormant subpopulation and the cyclic environmental fluctuations induced by seasonal changes. While these phenomena have been individually investigated, our proposed model integrates them. This approach not only offers a novel perspective but also presents a mechanism that may drive divergent adaptation and, potentially, eventual speciation.

The skills required for survival during the winter vastly differ from those needed to thrive in the summer. Consequently, it is not surprising that in a periodic environment, distinct subpopulations can specialize in each niche, leading to intriguing ecological effects. Recent investigations, exemplified by studies such as those by Gulisija and colleagues \cite{gulisija2016phenotypic, gulisija2017phenotypic}, have shed light on this phenomenon.  We are particularly interested in cycles spanning multiple generations, affording populations the opportunity to adapt over time. Consider, for instance, the adaptation of bacteria in the wild to the contrasting temperatures in the winter and summer.  While these environmental states endure for months, bacterial generations unfold within minutes. 

When a new winter arrives, bacteria need to relearn survival strategies that they learned during the previous winter.  Seed banks can therefore play an important role because these genetic reservoirs contain bacteria in a latent state with adaptations that may have once been advantageous but have since disappeared from the active population.  We will demonstrate, through a simple mathematical model, that seed banks could serve as catalysts for adaptation, fostering the emergence of specialized subpopulations tailored for the summer and the winter.

Seed banks profoundly influence evolution in numerous ways. The capacity to sustain a seed bank has been observed to impact various microevolutionary traits, prompting interest in the macroevolutionary consequences of dormancy. This ongoing discussion (see \cite{Lennon2011, Lennon2021}) highlights notable examples such as the tomato \cite{Tellier2011}, where the presence of seed banks explains a high speciation rate among South American tomatoes. Recent mathematical insights in adaptive dynamics indicate that seed banks can enhance speciation rates \cite{Blath2022}.  This manuscript illustrates how the combination of seed banks and changing environments could lead to subpopulations adapting in divergent directions, increasing their genetic distance. 

In a stable environment, a population undergoing selection and mutation will adapt. This adaptation results from the accumulation of beneficial mutations that survive random genetic drift and alter the genetic composition of the population.  Mathematical modeling of adaptation in controlled environments can be achieved using a Wright-Fisher (or Moran) model that accounts for both selection and mutation. In these models, mutations are recurring events, and an individual's relative fitness is an increasing function of the number of mutations it carries.

Research conducted in \cite{df07} and \cite{Neher2013} established that, in populations for which there is an unlimited supply of beneficial mutations, average fitness increases linearly over time. Furthermore, the empirical distribution of the number of mutations per individual forms a shape centered around this increasing average, which closely approximates a Gaussian distribution.  See \cite{Schweinsberg2017} for a mathematical study of ``adaptive waves," which makes rigorous some of the observations in \cite{df07, Neher2013}.  Here, we aim to generalize this model in two ways:

\begin{enumerate}
\item \textbf{Changing Environment}: We expand the type space from $\mathbb{N}$ to $\mathbb{Z}$ and consider two types of mutations: positive mutations, which increase the type of the individual from type $k$ to $k+1$, and negative mutations, which decrease the type from $k$ to $k-1$. Furthermore, we introduce environmental fluctuations, where the environment can be in the states $\{+1,-1\}$.  Individuals with larger types have a selective advantage in positive environments, while the opposite holds true for negative environments.

\item \textbf{Dormancy}: Individuals can exist in one of two states: active or dormant. Dormant individuals are unable to reproduce but are exempt from selective pressures. Transitions between these states occur at specific rates, with active individuals becoming dormant and vice versa.
\end{enumerate}

Under certain conditions within this model, we can demonstrate the emergence of two waves of adaptation.  Some individuals repeatedly acquire positive mutations, giving them a fitness advantage during the summer.  Other individuals repeatedly acquire negative mutations, giving them a fitness advantage during the winter.  Individuals are able to survive the season in which they are less fit by entering a dormant state.  This result can be interpreted as the emergence of speciation from fundamental principles, as the interplay of selection, mutation, dormancy, and a fluctuating environment causes the genetic distance between the prevalent types in each environment to increase over time.  Without dormancy, this outcome would not be achievable.  

A seed bank is a classic bet-hedging device, and our results partially make the biological payoff explicit. In a two-environment setting, a population without dormancy is driven toward a single compromise strategy which performs tolerably across regimes but is suboptimal in each. By contrast, dormancy allows the simultaneous maintenance of contrasting specialists. The seed bank makes this possible because when the environment flips, the previously favored specialist is quickly restored from the seed bank, accelerating the return to high fitness. Thus, despite the short-term cost of delayed reproduction, seed banks increase the speed of adaptation to alternating conditions. These observations, consistent with recent theoretical and empirical work, may help to explain the repeated emergence of dormancy across diverse lineages.

Although we have emphasized that fluctuating environments can result from seasonal changes, there are also other scenarios in which fluctuating environments could arise.  For example, when a patient takes an antibiotic at regular intervals, the environment experienced by the bacteria will fluctuate in a periodic way.  It has been proposed, for example in \cite{wkk13}, that entering a dormant state may allow bacteria to persist in the presence of antibiotics.

\subsection{The model}

We consider a discrete-time population model with fixed population size in which there are $N$ active individuals and $K_N$ dormant individuals in each generation.  The set of dormant individuals is sometimes called the seed bank.  The environment alternates between two possible states, denoted by $1$ and $-1$, which we think of as representing the summer and winter seasons.  We denote the length of a year by $U_N$ and the length of the summer by $V_N$, which means the environment will be in the state $1$ for $V_N$ generations, then in the state $-1$ for $U_N - V_N$ generations, and so on.  More precisely, for nonnegative integers $m$, we denote the environment in generation $m$ by $R_N(m)$, and we define

\begin{displaymath}
R_N(m) = \left\{
\begin{array}{ll} 1 & \mbox{ if $jU_N \leq m < jU_N + V_N$ for some nonnegative integer $j$} \\
-1 & \mbox{ otherwise.} \end{array} \right.
\end{displaymath}

Each individual has a type in $\Z$.  If an individual in generation $m$ has type $k \in \Z$, then the fitness of this individual is given by $$(1 + s_N)^{k R_N(m)},$$ where $s_N$ is a positive real number.  Note that an individual whose type is a positive number is more fit during the summer than during the winter, while a negative type indicates that the individual is more fit during the winter than during the summer.  In generation zero, all individuals have type $0$.

In each generation, $c_N$ individuals that were dormant in the previous generation are randomly chosen to become active, and $c_N$ individuals that were active in the previous generation are randomly chosen to become dormant.  The remaining $N - c_N$ active individuals each independently choose their parent from the active population in the previous generation, with probability proportional to the individual's fitness, so that individuals with higher fitness are more likely to have offspring.  Each of these active individuals has the same type as the parent with probability $1 - 2 \mu_N$.  With probability $\mu_N$, the type is the parent's type plus one, indicating that the individual acquired a mutation that will increase its fitness in the summer.  With probability $\mu_N$, the type is the parent's type minus one, indicating that the individual acquired a mutation that will increase its fitness in the winter. 

Note that this model has a similar structure to the seed bank model proposed in \cite{bgkw16}, but we are adding mutation, selection, and environmental changes.  A model involving dormancy and fluctuating environments was previously studied in \cite{bhs21}.

\subsection{Behavior of the population}\label{popbeh}

Before we introduce the assumptions that we will require on the parameters and state our results precisely, in this subsection we illustrate, in rough terms, how we expect the population to evolve.  There are two key points to keep in mind:

\begin{enumerate}
\item \textbf{The effect of mutations}:  We will consider strong selection, so that when a beneficial mutation occurs (a positive mutation during the summer or a negative mutation during the winter), there is a probability bounded away from zero that this beneficial mutation will spread to nearly the entire population.  The parameters will be chosen so that the beneficial mutation spreads very quickly, on a time scale much faster than the length of a season.

\item \textbf{The beginning of a season}:  At the beginning of each summer (or winter), the type that was dominant in the population at the end of the previous summer (or winter) quickly emerges from the seed bank and becomes dominant again.
\end{enumerate}

\begin{figure}[h!]
\begin{center}
\includegraphics[width=1\linewidth]{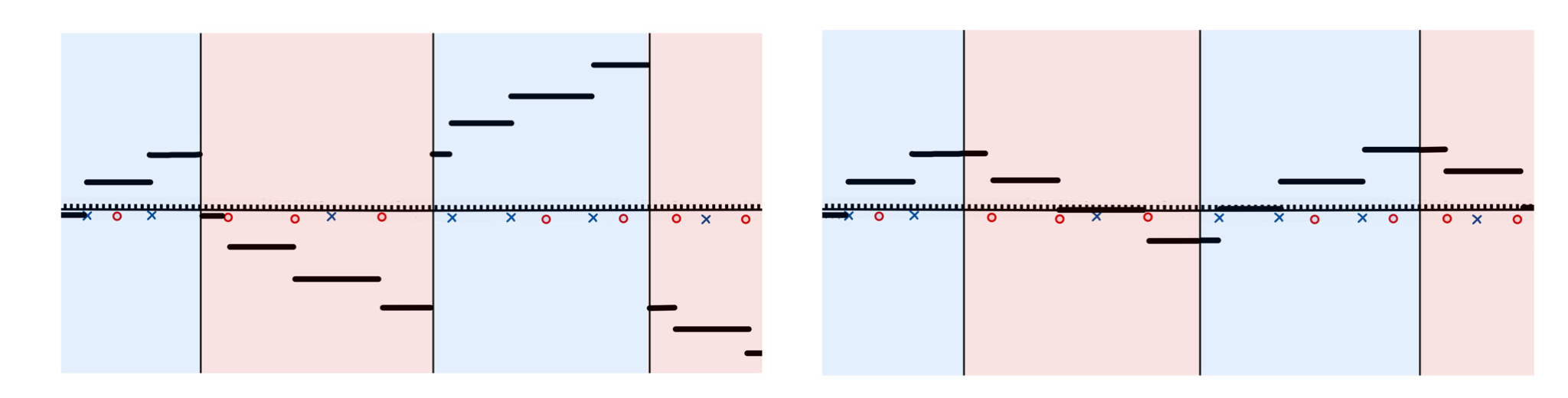}
\end{center}
\vspace{-.4in}
\caption{{\small The adaptation of populations in changing environments is depicted in this figure.  The crosses represent mutations that increase the type and the circles represent mutations that decrease the type. The black lines indicate the predominant type as a function of time.  The panel on the left is for a population with a seed bank, and the panel on the right is for a population without a seed bank.}}
\label{fig:fitness}
\end{figure}

The behavior of the population is illustrated in Figure \ref{fig:fitness}.  In this figure, crosses represent positive mutations, and circles represent negative mutations.  Blue periods indicate summers, and red periods indicate winters.  The left panel shows the evolution of the population in the presence of a seed bank.  During the first summer, there are two positive mutations, indicated by crosses, so the dominant type in the population increases from $0$ to $2$.  The red circle indicates a negative mutation, which is ignored because a negative mutation during the summer quickly dies out.  At the beginning of the following winter, type 0 emerges from the seed bank and becomes dominant.  Then three negative mutations bring the dominant type down to $-3$.  Then type~2 emerges from the seed bank at the beginning of the next summer, and three more positive mutations increase the dominant type to $5$.  Then individuals of type $-3$ return from the seed bank at the beginning of the following winter, and the process continues.

The right panel in Figure \ref{fig:fitness} illustrates how the population would evolve in the absence of a seed bank.  As before, the dominant type increases from $0$ to $2$ during the first summer.  However, because type $0$ is not able to emerge from a seed bank at the beginning of the next winter, the three beneficial mutations during the winter only reduce the dominant type from $2$ to $-1$.  Likewise, the three beneficial mutations during the following summer increase the dominant type from $-1$ back to $2$, and so on.  During each year, a Poisson distributed number of positive and negative mutations spread through the population.  Therefore, as long as the two seasons are of equal length, the dominant type at the end of each year will evolve much like a mean zero random walk.  If $O(1)$ mutations occur each year, the dominant type after $M$ years will be $O(\sqrt{M})$. On the other hand, when there is a seed bank, the dominant type during the $M$th summer will be approximately a constant multiple of $M$, while the dominant type during the $M$th winter will be a constant multiple of $-M$. Our work thus uncovers a novel advantage of maintaining a seed bank, namely an increase in the speed of adaptation in fluctuating environments. Additionally, our research elucidates how the presence of seed banks could facilitate speciation.

In our model, individuals enter and leave the dormant state at random, independently of environmental conditions.  In an alternative model in which individuals enter and exit the dormant state in response to environmental conditions, we would expect it to be even easier for the type that was dominant, say, during one summer to emerge from the seed bank during the following summer, and therefore even more likely that two waves of adaptation would emerge.

\subsection{Assumptions on the parameters}

The model contains six parameters: $K_N$, $c_N$, $U_N$, $V_N$, $s_N$, and $\mu_N$.  Because our main results are limit theorems that hold as the size $N$ of the active population tends to infinity, our assumptions will involve the values of these parameters as $N \rightarrow \infty$.  If $(a_N)_{N=1}^{\infty}$ and $(b_N)_{N=1}^{\infty}$ are sequences of positive numbers, we will sometimes write $a_N \ll b_N$ to mean that $\lim_{N \rightarrow \infty} a_N/b_N = 0$ and $a_N \gg b_N$ to mean $\lim_{N \rightarrow \infty} a_N/b_N = \infty$.  Also, we will sometimes write $a_N \lesssim b_N$ to mean that $\limsup_{N \rightarrow \infty} a_N/b_N < \infty$ and $a_N \sim b_N$ to mean $\lim_{N \rightarrow \infty} a_N/b_N = 1$.

The following five assumptions will be in effect throughout the paper:
\begin{itemize}
\item We have
\begin{equation}\label{slarge}
\lim_{N \rightarrow \infty} s_N = s \in (0, \infty].
\end{equation}

\item We have
\begin{equation}\label{betadef}
\lim_{N \rightarrow \infty} V_N/U_N = \beta \in (0, 1).
\end{equation}

\item We have
\begin{equation}\label{largeU}
U_N \gg \log N.
\end{equation}

\item We have
\begin{equation}\label{NcN}
c_N \ll N.
\end{equation}

\item There exists $b > 0$ such that
\begin{equation}\label{fastmutations}
\lim_{N \rightarrow \infty} N^b \mu_N = \infty.
\end{equation}
\end{itemize}
Because a beneficial mutation increases an individual's fitness by a factor of $1 + s_N$, the first assumption stipulates that we are working with strong selection.  The second assumption ensures that the lengths of the summer and winter seasons are comparable.  Because the time required for a new beneficial mutation to spread to most of the active population (an event known as a selective sweep) is comparable to $\log N$, the third assumption ensures that selective sweeps have time to occur before the environment changes.  The fourth assumption is needed to make sure that individuals do not switch between the active and dormant populations so fast that new beneficial mutations are unable to spread in the active population.  Finally, the fifth assumption is needed to make sure that, with high probability, the next mutation will occur before certain other rare events that would substantially alter the behavior of the population.

Let $\omega(u)$ denote the probability that a Galton-Watson process started with one individual whose offspring distribution is Poisson with mean $1 + u$ will survive forever, and let $\omega(\infty) = 1$.  We will see later that when a new beneficial mutation appears in the population, the probability that it will spread to a large fraction of the active population is approximately $\omega(s_N)$.  Note that there are $N$ individuals in the active population each acquiring beneficial mutations at rate $\mu_N$.  Therefore, we define
\begin{equation}\label{rhoNdef}
\rho_N = N \mu_N \omega(s_N),
\end{equation}
which approximates the probability that, in a given generation, a mutation occurs which eventually causes a selective sweep.  An immediate consequence of \eqref{slarge} is that
\begin{equation}\label{hlower}
\inf_N \omega(s_N) > 0.
\end{equation}

Three different time scales are relevant for our process.  One time scale is $U_N$, the length of a year.  Another is $\rho_N^{-1}$, which is approximately the average number of generations between selective sweeps.  Because a fraction $c_N/K_N$ of the seed bank is replaced in each generation, the number of generations that it takes for the dominant type in the active population to start to take over the seed bank is on the scale of $K_N/c_N$, which is a third important time scale.  The behavior of this population model depends on how these three time scales compare with one another.  We will consider two different parameter regimes.

\bigskip
\noindent {\bf Regime 1}:  In Regime 1, we assume that
\begin{equation}\label{Reg1Asm1}
\lim_{N \rightarrow \infty} \frac{c_N}{K_N \rho_N} = \alpha \in (0, \infty)
\end{equation}
and
\begin{equation}\label{Reg1Asm2}
\lim_{N \rightarrow \infty} U_N \rho_N = \theta \in (0, \infty).
\end{equation}
Under these assumptions, all three of the time scales mentioned above are comparable.

\bigskip
\noindent {\bf Regime 2}: In Regime 2, we assume instead that 
\begin{equation}\label{Reg2Asm}
U_N \ll \frac{K_N}{c_N} \ll \rho_N^{-1}
\end{equation}
and
\begin{equation}\label{KcU}
\lim_{N \rightarrow \infty} \frac{c_N^2 U_N^2}{K_N \log N} = \infty.
\end{equation}
Because $U_N \ll K_N/c_N$, it takes longer for the active population to populate the seed bank than it takes for the environment to change.  Because $K_N/c_N \ll \rho_N^{-1}$, 
the process by which the active population populates the seed bank will typically equilibrate before the next mutation.  The more technical condition (\ref{KcU}) is needed to ensure that individuals can rejoin the active population from the seed bank fast enough for our results to hold.

\bigskip
In Regime 1, we have $c_N U_N/K_N \rightarrow \alpha \theta$ as $N \rightarrow \infty$.  Therefore, in view of (\ref{largeU}) and the fact that $c_N \geq 1$, the result (\ref{KcU}) holds automatically and is not an additional assumption.  However, in Regime 2, the condition (\ref{KcU}) does not follow from the other assumptions.  

Conditions (\ref{Reg1Asm2}) and (\ref{Reg2Asm}) each imply that $U_N \rho_N \lesssim 1$, which by (\ref{rhoNdef}) and (\ref{hlower}) is equivalent to the condition that
\begin{equation}\label{UNmu}
U_N N \mu_N \lesssim 1.
\end{equation}
Therefore, using (\ref{largeU}), we get
\begin{equation}\label{rhoto0}
\rho_N \ll \frac{1}{\log N}
\end{equation}
and
\begin{equation}\label{Nmubound}
N \mu_N \ll \frac{1}{\log N}.
\end{equation}

The assumption \eqref{slarge} is, of course, a strong assumption, as it requires mutations to confer a substantial selective advantage.  We do believe that this assumption could be weakened.  However, if selection were weaker, then selective sweeps would take longer, and therefore seasons would have to be longer to ensure that selective sweeps would be completed within a season and would stay well separated in time.  Our primary aim here is to establish the principle that, under certain conditions, multiple waves of adaptation can be triggered by a combination of dormancy and fluctuation environmental conditions, as explained in subsection \ref{popbeh}.

\subsection{Main Results}

In this section, we will state our main results about the composition of the active and dormant populations in both of the parameter regimes that we are considering.  We will also give some insight into why we should expect these results to be true.  We first introduce some notation.  For all $k \in \Z$ and all nonnegative integers $m$, let $A_{k,N}(m)$ be the number of active individuals in generation $m$ having type $k$, and let $D_{k,N}(m)$ be the number of dormant individuals in generation $m$ having type $k$.
Let $X_{k,N}(m) = A_{k,N}(m)/N$ be the proportion of type $k$ individuals in the active population in generation $m$, and let $Y_{k,N}(m) = D_{k,N}(m)/N$ be the proportion of type $k$ individuals in the dormant population in generation $m$.

\subsubsection{The times when new mutations appear}

Fix $a \in (0, 1)$.  For all $k \in \Z$, define $$S_{k,N}(a) = \min\{m: X_{k,N}(m) > a\},$$ which is the first time that type $k$ individuals comprise a fraction at least $a$ of the population.  Note that $S_{0,N}(a) = 0$ because all individuals initially have type $0$.

The idea behind our first result is that, as noted in the discussion around (\ref{rhoNdef}), in each generation, with probability approximately $\rho_N$, a beneficial mutation will appear that will eventually spread rapidly to most of the active population.  Therefore, the times $S_{k,N}(a)$, after suitable rescaling, will converge to the event times of a Poisson process.  The rescaling of time is subtle, however, because the Poisson clock only runs when the environment has the correct sign.  Indeed, positive mutations only have a chance to spread in the population during the summer, and negative mutations can only spread in the population during the winter.

We will scale time by $\rho_N$, so that one unit of time corresponds to $1/\rho_N$ generations in the original model.  This means that on the new time scale, beneficial mutations will happen at approximately rate one.  We then define
$$I^+(t) = \lim_{N \rightarrow \infty} \rho_N \sum_{m=0}^{\lfloor \rho_N^{-1} t \rfloor} \1_{\{R_N(m) = 1\}}, \qquad I^-(t) = \lim_{N \rightarrow \infty} \rho_N \sum_{m=0}^{\lfloor \rho_N^{-1} t \rfloor} \1_{\{R_N(m) = -1\}}.$$
We will see in what follows that these limits are indeed well-defined. Note that $I^+(t)$ represents the amount of time, up to time $t$ after rescaling, that the process spends in the positive environment, while $I^-(t)$ represents the amount of time, up to time $t$ after rescaling, that the process spends in the negative environment.

In Regime 1, note that $U_N$ generations in the original model corresponds to time $U_N \rho_N$, which converges to $\theta$, while $V_N$ generations in the original model corresponds to time $V_N \rho_N$, which converges to $\beta \theta$.   Define the environment in the limit to be 
\begin{equation}\label{Rtdef}
R(t) = \left\{
\begin{array}{ll} 1 & \mbox{ if $j\theta \leq t < j \theta + \beta \theta$ for some nonnegative integer $j$} \\
-1 & \mbox{ otherwise.} \end{array} \right.
\end{equation}
Then for all $t > 0$, we have
$$I^+(t) = \int_0^t \1_{\{R(u) = 1\}} \: \dd u, \qquad I^-(t) = \int_0^t \1_{\{R(u) = -1\}} \: \dd u.$$
In Regime 2, because the environment changes on a faster time scale and $\beta$ is the long-run fraction of time that the process spends in the positive environment, we have instead $$I^+(t) = \beta t, \qquad I^-(t) = (1 - \beta) t.$$

Let $(\mathcal{N}^+(t), t \geq 0)$ and $(\mathcal{N}^-(t), t \geq 0)$ be independent inhomogeneous Poisson processes whose intensities in Regime 1 are given by $\1_{\{R(t) = 1\}}$ and $\1_{\{R(t) = -1\}}$ respectively, and whose intensities in Regime 2 are given by $\beta$ and $1 - \beta$ respectively.  That is, the distribution of $\mathcal{N}^+(t)$ is Poisson with mean $I^+(t)$, and the distribution of $\mathcal{N}^-(t)$ is Poisson with mean $I^-(t)$.  Let $T_0 = 0$, and for positive integers $k$, define $$T_k = \inf\{t: \mathcal{N}^+(t) = k\}, \qquad T_{-k} = \inf\{t: \mathcal{N}^-(t) = k\}.$$
Note that the sequences $(I^+(T_{k}))_{k=1}^{\infty}$ and $(I^-(T_{-k}))_{k=1}^{\infty}$ are the event times of independent rate one Poisson processes.  We can now state the following limit theorem for the times $S_{k,N}(a)$.

\begin{Theo}\label{Poissontheo}
In both regimes, for all $a \in (0, 1)$ and all positive integers $K$, as $N \rightarrow \infty$ we have the convergence in distribution
$$\big( \rho_N S_{k,N}(a) \big)_{k=-K}^K \Rightarrow (T_k)_{k=-K}^K.$$
\end{Theo}

Note the role that dormancy plays in this result.  Suppose that $k$ positive mutations have taken hold in the active population by the end of a certain summer.  Then, regardless of how many negative mutations arise during the following winter, individuals with $k$ positive mutations quickly emerge from the seed bank and take over the population at the beginning of the following summer.  Consequently, the first successful mutation that happens during the following summer will indeed enable type $k+1$ individuals to spread in the population.  The seed bank ensures that the $k$ previous positive mutations are not lost, and that positive mutations can continue to accumulate at the event times of these Poisson processes.

\subsubsection{The active population}

We next state a result about the composition of the active population, which holds in both regimes.  Except during brief transitional periods, the active population is dominated by one type, which will be nonnegative when the environment is $1$ and nonpositive when the environment is $-1$.  For example, during the summer, after time $S_{k,N}(a)$, type $k$ quickly becomes dominant in the population, and remains dominant until either type $k+1$ emerges, or until the season changes and the type that was dominant during the previous winter returns from the seed bank.  For $a \in (0, 1)$ and nonnegative integers $m$, we can therefore define the dominant type in generation $m$ by 
\begin{displaymath}
{\dom}_{N,a}(m) = \left\{
\begin{array}{ll}  \max\{k: S_{k,N}(a) \leq m\} & \mbox{ if }R_N(m) = 1  \\
\min\{k: S_{k,N}(a) \leq m\} & \mbox{ if }R_N(m) = -1.
\end{array} \right.
\end{displaymath}

\begin{Theo}\label{activethm}
Fix $t_0 > 0$ and $a \in (0, 1)$.  
In both regimes, for all $k \in \Z$, we have
$$\rho_N \sum_{m=0}^{\lfloor \rho_N^{-1} t_0 \rfloor} \big| X_{k,N}(m) - \1_{\{{\dom}_{N, a}(m) = k\}} \big| \rightarrow_p 0.$$
\end{Theo}

Note that we can not state Theorem \ref{activethm} as a supremum over all generations $m$ because there will be brief periods, around the times $S_{k,N}(a)$ and around the times when the environment changes, where $X_{k,N}(m)$ will not be close to zero or one.  However, in Regime 1 it is possible to state a result in terms of convergence of stochastic processes, provided that we replace the usual Skorohod $J_1$ topology on the space of c\`adl\`ag functions with the weaker Meyer-Zheng topology.  The Meyer-Zheng topology is the topology associated with convergence in measure.  If $(E,d)$ is a metric space, then c\`adl\`ag functions $f_n: [0, \infty) \rightarrow E$ converge to $f$ in the Meyer-Zheng topology if for all $T > 0$ and $\eps > 0$, the Lebesgue measure of $\{x \in [0,T]: d(f_n(x) - f(x)) > \eps\}$ tends to zero as $n \rightarrow \infty$.  We refer the reader to \cite{mz84}, section 4 of \cite{k91}, and section 3 of \cite{gonzalez2022symmetric} for more details on this topology.

Let $\Delta$ be the space of all doubly infinite sequences ${\bf x} = (\dots, x_{-2}, x_{-1}, x_0, x_1, x_2, \dots)$ such that $x_k \geq 0$ for all $k \in \Z$ and $\sum_{k \in \Zsm} x_k \leq 1$.  We call $x_k$ the $k$th coordinate of ${\bf x}$.  We equip $\Delta$ with the $\ell_1$ norm.  Let ${\bf X}_N(m)$ be the $\Delta$-valued random variable whose $k$th coordinate is $X_{k,N}(m)$.

\begin{Theo}\label{mztheorem}
For $t \geq 0$, define
$$
\domi(t) =  \mathcal{N}^{+}(t) \1_{\{R(t) = 1\}} - \mathcal{N}^{-}(t)\1_{\{R(t) = -1\}},
$$ 
and let ${\idom}(t)$ be the $\Delta$-valued random variable whose $k$th coordinate is $\1_{\{\domi(t) = k\}}$.
In Regime 1, the processes $({\bf X}_N(\lfloor \rho_N^{-1}t \rfloor), t \geq 0)$ converge as $N \rightarrow \infty$ to $({\idom}(t), t \geq 0)$, in the sense of weak convergence of stochastic processes with respect to the Meyer-Zheng topology.
\end{Theo}

A convergence result of this type is not possible in Regime 2 because the seasons change very fast on the time scale that we are considering, and therefore there is no well-defined limit process.

Theorems \ref{Poissontheo}, \ref{activethm}, and \ref{mztheorem} are proved in Section \ref{proofsactive}.

\subsubsection{The dormant population}

We next consider how the dormant population changes over time, working first in Regime 1.  To motivate our result, suppose in some generation, the fraction of type $k$ individuals in the active population is $x$, and the fraction of type $k$ individuals in the seed bank is $y$.  Then $c_N$ individuals from the active population will enter the seed bank in the next generation, while $c_N$ individuals in the seed bank will leave.  Therefore, in the next generation, the fraction of the seed bank consisting of type $k$ individuals changes from $y$ to approximately $$\frac{xc_N + y(K_N - c_N)}{K_N} = y + (x - y) \frac{c_N}{K_N}.$$  Thus, in Regime 1, after scaling time by $\rho_N$, we obtain the approximate differential equation
\begin{equation}\label{dormdiff}
\frac{\dd y}{\dd t} \approx (x-y) \frac{c_N}{K_N \rho_N} \approx \alpha (x - y).
\end{equation}

In Regime 1, this idea leads to the following convergence theorem.  Here, ${\bf Y}_N(m)$ denotes the $\Delta$-valued random variable whose $k$th coordinate is $Y_{k,N}(m)$.  Note that this theorem implies that we have coexistence of different types in the dormant population, in contrast to the active population which is usually dominated by one type.

\begin{Theo}\label{dormantthm1}
Define a $\Delta$-valued stochastic process $({\bf Y}(t), t \geq 0)$ such that ${\bf Y}(0)$ is the sequence whose $0$th term is $1$ and whose other terms are zero, and $$\frac{\dd}{\dd t} {\bf Y}(t) = \alpha ({\idom}(t) - {\bf Y}(t)).$$
In Regime 1, the processes $({\bf Y}_N(\lfloor \rho_N^{-1} t \rfloor), t \geq 0)$ converge as $N \rightarrow \infty$ to $({\bf Y}(t), t \geq 0)$, in the sense of weak convergence of stochastic processes with respect to Skorohod's $J_1$ topology.
\end{Theo}

The assumptions of Regime 2 imply that the seasons change rapidly, so the active population switches quickly back and forth between being dominated by individuals of positive and negative type.  The assumptions also imply that $c_N/(K_N \rho_N) \rightarrow \infty$, which means that the differential equation \eqref{dormdiff} moves fast on the time scale of interest, and therefore the composition of the dormant population should rapidly converge to an equilibrium.  Because the seed bank is being populated with individuals of the positive type a fraction $\beta$ of the time, in equilibrium a fraction $\beta$ of the individuals in the seed bank will have a positive type, while a fraction $1 - \beta$ will have a negative type.  This leads to the following limit theorem in Regime 2.  

\begin{Theo}\label{dormantthm2}
For all $k \in \Z$, define
\begin{displaymath}
Y_k(t) = \left\{
\begin{array}{ll} \beta \1_{\{T_k \leq t < T_{k+1}\}} & \mbox{ if }k \geq 1  \\
(1 - \beta) \1_{\{T_{k} \leq t < T_{k-1}\}} & \mbox{ if }k \leq -1 \\
\beta \1_{\{t < T_{1}\}} + (1 - \beta) \1_{\{t <T_{-1}\}} & \mbox{ if }k = 0.
\end{array} \right.
\end{displaymath}
Let ${\bf Y}(t)$ be the $\Delta$-valued random variable whose $k$th coordinate is $Y_k(t)$.
In Regime 2, the processes $({\bf Y}_N(\lfloor \rho_N^{-1} t \rfloor), t \geq 0)$ converge as $N \rightarrow \infty$ to $({\bf Y}(t), t \geq 0)$, in the sense of weak convergence of stochastic processes with respect to the Meyer-Zheng topology.
\end{Theo}

Note that in Regime 2, because the limit process does not have continuous paths, we can not expect the convergence to hold in Skorohod's $J_1$ topology, and we must use the weaker Meyer-Zheng topology.

Theorems \ref{dormantthm1} and \ref{dormantthm2} are proved in Section \ref{proofsdormant}.

\subsubsection{Speciation and the genealogy of the population}

There is no consensus within the biological community on a precise definition of bacterial speciation, and addressing this foundational issue is beyond the scope of the present paper. Instead, we introduce two natural notions of distance between individuals.  We then show that, according to both notions, individuals specialized for summer and winter conditions progressively diverge, ultimately becoming distinct species in this sense.

Given an individual $x_1$ in generation $m_1$ and an individual $x_2$ in generation $m_2$, we can define a distance between them in two ways.  The genetic distance, which we will denote by $d(x_1,x_2)$, is the absolute value of the difference between the types of the individuals.  We can also define a genealogical distance $g(x_1, x_2)$.  If $\ell$ is the generation in which the most recent common ancestor between $x_1$ and $x_2$ lived, then we define $g(x_1, x_2) = \rho_N( \frac{1}{2}(m_1 + m_2) - \ell)$.  If $m_1 = m_2$, then $g(x_1,x_2)$ is the time back to the most recent common ancestor of the two individuals, after the time scaling by $\rho_N$.

For $t > 0$, let 
\begin{align*}
\chi_{N,t}^+ &= \min\{jU_N + V_N - 1: j \in \Z \mbox{ and } j U_N + V_N - 1 \geq \rho_N^{-1} t\} \\
\chi_{N,t}^- &= \min\{jU_N - 1: j \in \Z \mbox{ and } jU_N - 1 \geq \rho_N^{-1} t\}.
\end{align*}
Note that $\chi_{N,t}^+$ is the first generation after $\rho_N^{-1} t$ that is at the end of a summer, and $\chi_{N,t}^-$ is the first generation after $\rho_N^{-1} t$ that is at the end of a winter.  For $1 \leq i \leq N$ and nonnegative integers $m$, let $x_i(m)$ denote the $i$th individual in generation $m$ after a random reordering of the individuals, where the ordering is done independently in each generation.  The following result shows that the genetic distance between two individuals sampled at the end of a given season is typically zero, but that the genetic distance between two individuals sampled at the end of opposite seasons gets arbitrarily large over time.

\begin{Prop}\label{mutdistance}
If $t > 0$, then in both regimes,
\begin{equation}\label{mutsame}
\lim_{N \rightarrow \infty} P\big(d(x_1(\chi_{N,t}^+), x_2(\chi_{N,t}^+)) = 0\big) = \lim_{N \rightarrow \infty} P\big(d(x_1(\chi_{N,t}^-), x_2(\chi_{N,t}^-)) = 0\big) = 1.
\end{equation}
Also, for all $\eps > 0$ and all positive integers $n$, there exists a $t^* > 0$ such that if $t > t^*$, then
\begin{equation}\label{mutopp}
\liminf_{N \rightarrow \infty} P\big(d(x_1(\chi_{N,t}^+), x_1(\chi_{N,t}^-)) > n\big) > 1 - \eps.
\end{equation}
\end{Prop}

The next result is a similar result for the genealogical distance.  It shows that with high probability, two individuals sampled from the end of a season have a common ancestor in the recent past, while two individuals sampled at the end of opposite seasons have their most recent common ancestor near time zero.

\begin{Prop}\label{gendistance}
Let $\eps > 0$.  Then in both regimes, there exists a $t^* > 0$ such that for all $t > t^*$, we have
\begin{equation}\label{gensummer}
\liminf_{N \rightarrow \infty} P\big(g(x_1(\chi_{N,t}^+), x_2(\chi_{N,t}^+)) < t^*\big) > 1 - \eps,
\end{equation}
\begin{equation}\label{genwinter}
\liminf_{N \rightarrow \infty} P\big(g(x_1(\chi_{N,t}^-), x_2(\chi_{N,t}^-)) < t^*\big) > 1 - \eps,
\end{equation}
and
\begin{equation}\label{genopp}
\liminf_{N \rightarrow \infty} P\big(g(x_1(\chi_{N,t}^+), x_1(\chi_{N,t}^-)) > t - t^* \big) > 1 - \eps.
\end{equation}
\end{Prop}

Propositions \ref{mutdistance} and \ref{gendistance} both follow easily from ideas in the proof of Theorem \ref{activethm} and are proved in subsection \ref{proofsdistance}.

\subsubsection{A coupling}\label{earlycouple}

To prove these results, we will construct, for each $N$, random times $T_{k,N}$ such that the sequence $(T_{k,N})_{k=-\infty}^{\infty}$ has exactly the same distribution as $(T_k)_{k=-\infty}^{\infty}$.  In particular, the sequences $(I^+(T_{k,N}))_{k=1}^{\infty}$ and $(I^-(T_{k,N}))_{k=1}^{\infty}$ will be the event times of independent rate one Poisson processes.  We will show that for all $a \in (0,1)$ and all $k \in \Z$, as $N \rightarrow \infty$ we have
$$\big| \rho_N S_{k,N}(a) - T_{k,N} \big| \rightarrow_p 0,$$
which immediately implies Theorem \ref{Poissontheo}.

We can also define $n_N^+(t) = \sup\{k \geq 0: T_{k,N} \leq t\}$ and $n_N^-(t) = \sup\{k \geq 0: T_{-k,N} \leq t\}$, so that the processes $(n_N^+(t), t \geq 0)$ and $(n_N^-(t), t \geq 0)$ have the same laws as $(\mathcal{N}^+(t), t \geq 0)$ and $(\mathcal{N}^-(t), t \geq 0)$ respectively.  Then for each $k \in \Z$, we can define $d_{k,N}(t) = \1_{\{n_N^+(t) = k\}}$ when $R(t) = 1$ and $d_{k,N}(t) = \1_{\{n_N^-(t) = -k\}}$ when $R(t) = -1$.

In Regime 1, for $k \in \Z$, we define the function $y_{k,N}$ as the solution to the differential equation
$$\frac{\dd }{\dd t} y_{k,N}(t) = \alpha(d_{k,N}(t) - y_{k,N}(t)), \qquad y_{k,N}(0) = \1_{\{k = 0\}}.$$
Then the process $(y_{k,N}(t), t \geq 0)$ has the same law as the $k$th coordinate of $({\bf Y}(t), t \geq 0)$.
We will show that for each fixed $t_0 > 0$, as $N \rightarrow \infty$ we have
$$\sup_{t \in [0, t_0]} \big| Y_{k,N}(\lfloor \rho_N^{-1} t \rfloor) - y_{k,N}(t) \big| \rightarrow_p 0,$$
which implies Theorem \ref{dormantthm1}.

In Regime 2, we define $y_{k,N}(t)$ in the same way that $Y_k(t)$ is defined in Theorem \ref{dormantthm2}, except with $T_{k,N}$ in place of $T_k$.  We will show that for all $\delta > 0$ there is a set $I_{k,N}(\delta, t_0) \subseteq [0, t_0]$ such that the Lebesgue measure of $[0, t_0] \cap I_{k,N}(\delta, t_0)^c$ is bounded above by $4 \delta$, and
$$\sup_{t \in I_{k,N}(\delta, t_0)} \big| Y_{k,N}(\lfloor \rho_N^{-1} t \rfloor) - y_{k,N}(t) \big| \rightarrow_p 0.$$
This will imply Theorem \ref{dormantthm2}.

\subsection{Simulation results}

We ran some simulations of our process using MATLAB.  For these simulations, we used the parameters $N = 10,000$, $K_N = 10,000$, $c_N = 10$, $s_N = 0.1$, and $\mu_N = 0.00002$.  Figure \ref{fitness} shows the evolution of the mean fitness of the active population over time.  For this simulation, the summers and winters were both 500 generations long, and we plotted the mean type of the $10,000$ active individuals in each generation, for the first five years.

\begin{figure}[h!]
\begin{center}
\includegraphics[scale=0.5, trim={1.8cm 8.1cm 1.9cm 8.2cm}, clip]{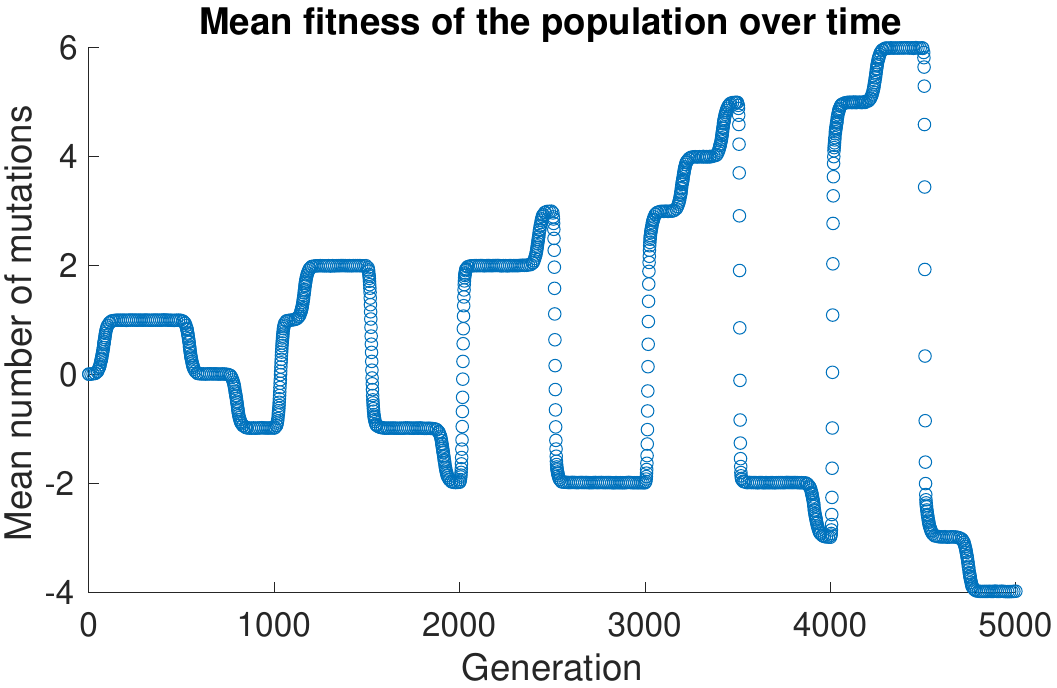}
\end{center}
\caption{{\small The mean type of the active individuals is plotted as a function of the generation number.  Positive mutations became established in the population early in each of the first two summers and late in the third summer.  Two more positive mutations became established during the fourth summer, and a sixth positive mutation appeared during the fifth summer.  Negative mutations became established in the population during the first, second, fourth, and fifth winters.}}
\label{fitness}
\end{figure}

The simulations shown in Figure \ref{actdorm} used the same parameters, except that each season lasted for 250 generations instead of 500.  We have plotted histograms of the type distributions in the active and dormant populations at different times.

\begin{figure}[h!]
\begin{center}
\includegraphics[scale=0.212, trim={1cm 6.6cm 2cm 6.8cm}, clip]{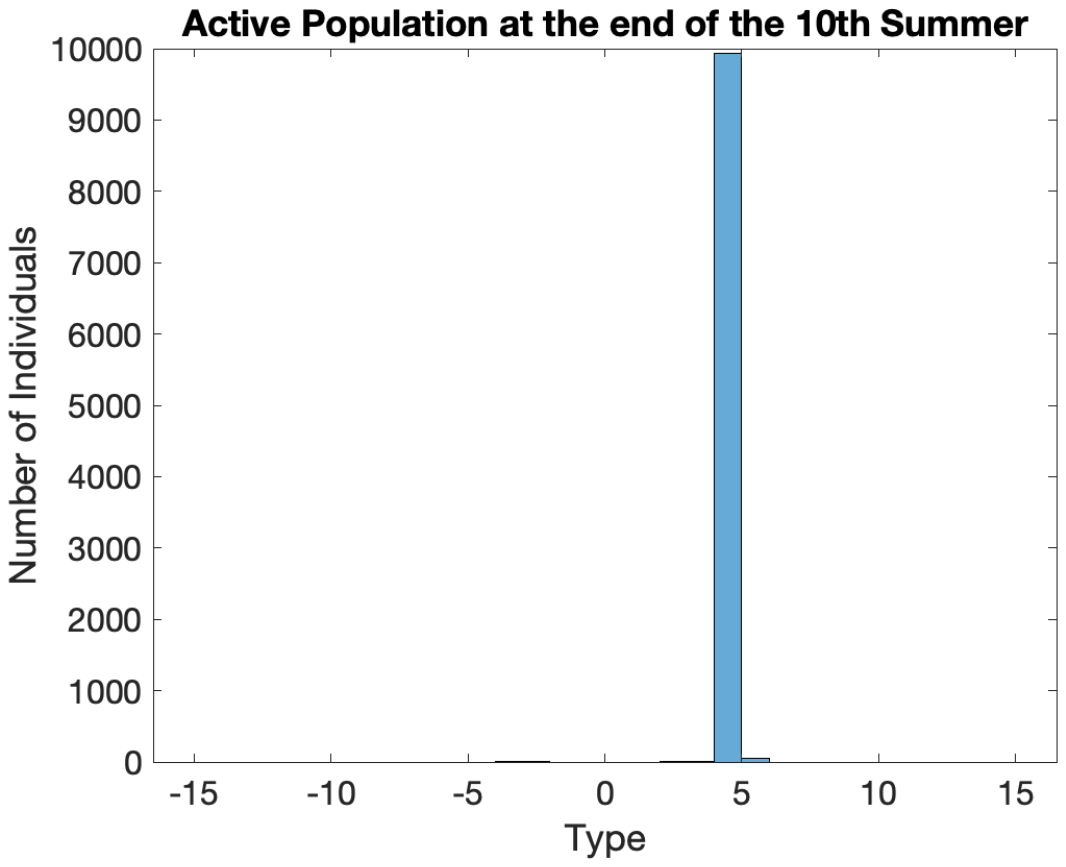}
\includegraphics[scale=0.212, trim={1cm 6.6cm 2cm 6.8cm}, clip]{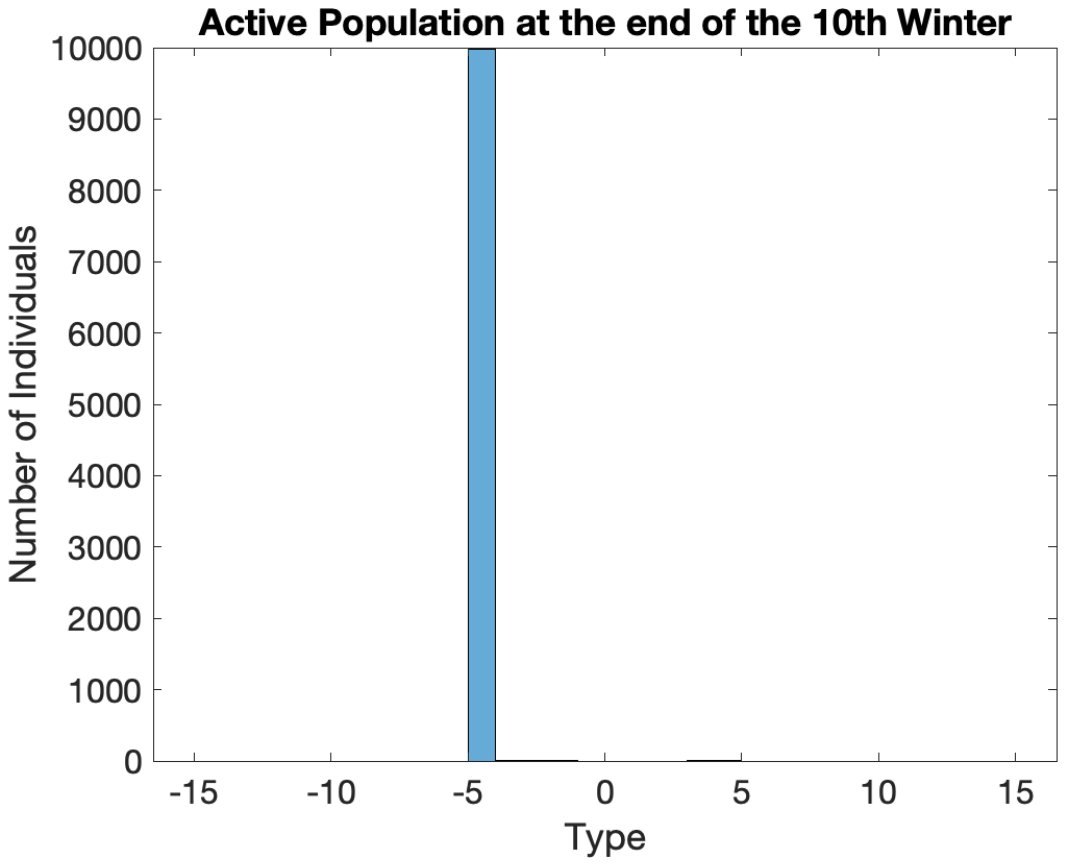} 
\includegraphics[scale=0.212, trim={1cm 6.6cm 2cm 6.8cm}, clip]{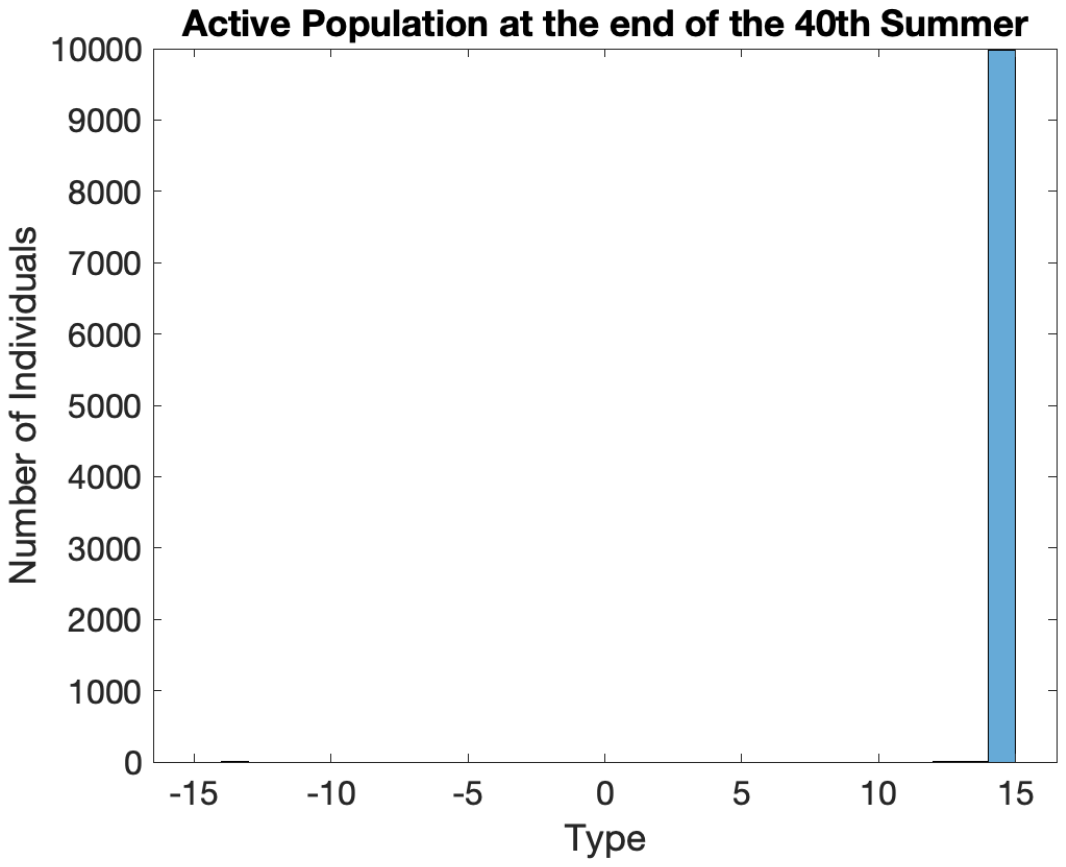}
\includegraphics[scale=0.212, trim={1cm 6.6cm 2cm 6.8cm}, clip]{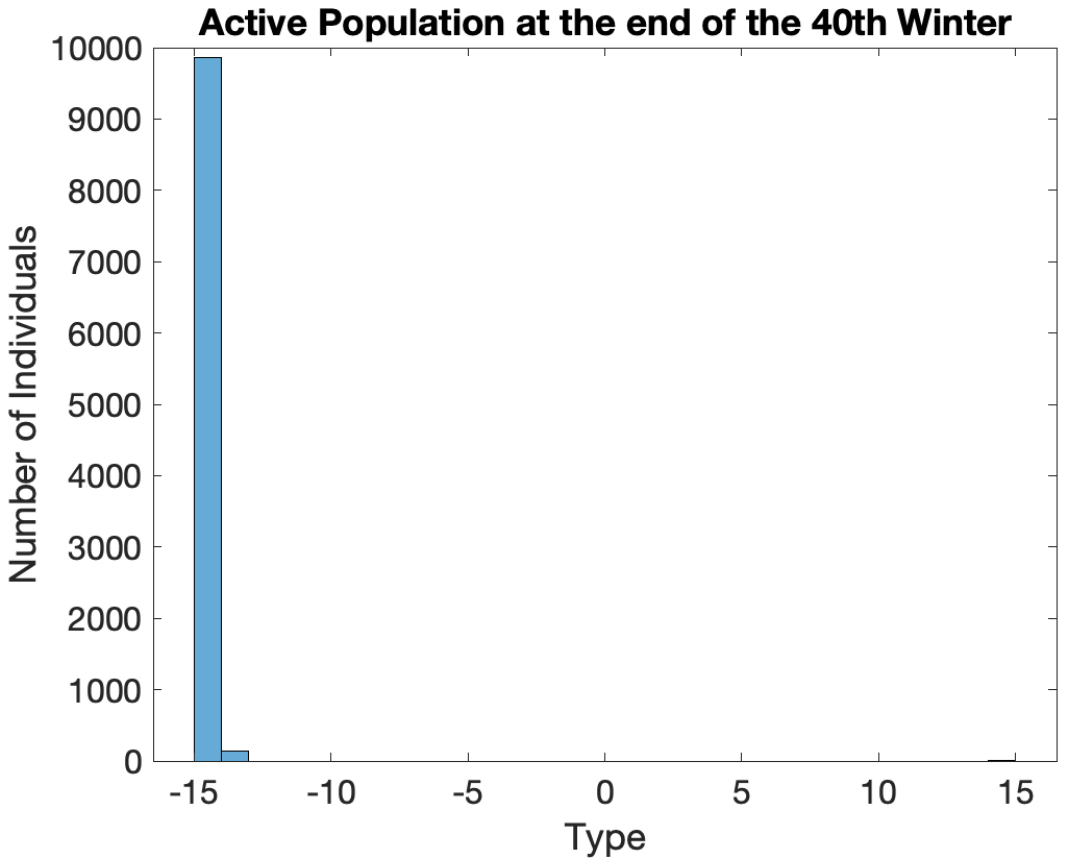}
\includegraphics[scale=0.212, trim={1cm 6.6cm 2cm 6.8cm}, clip]{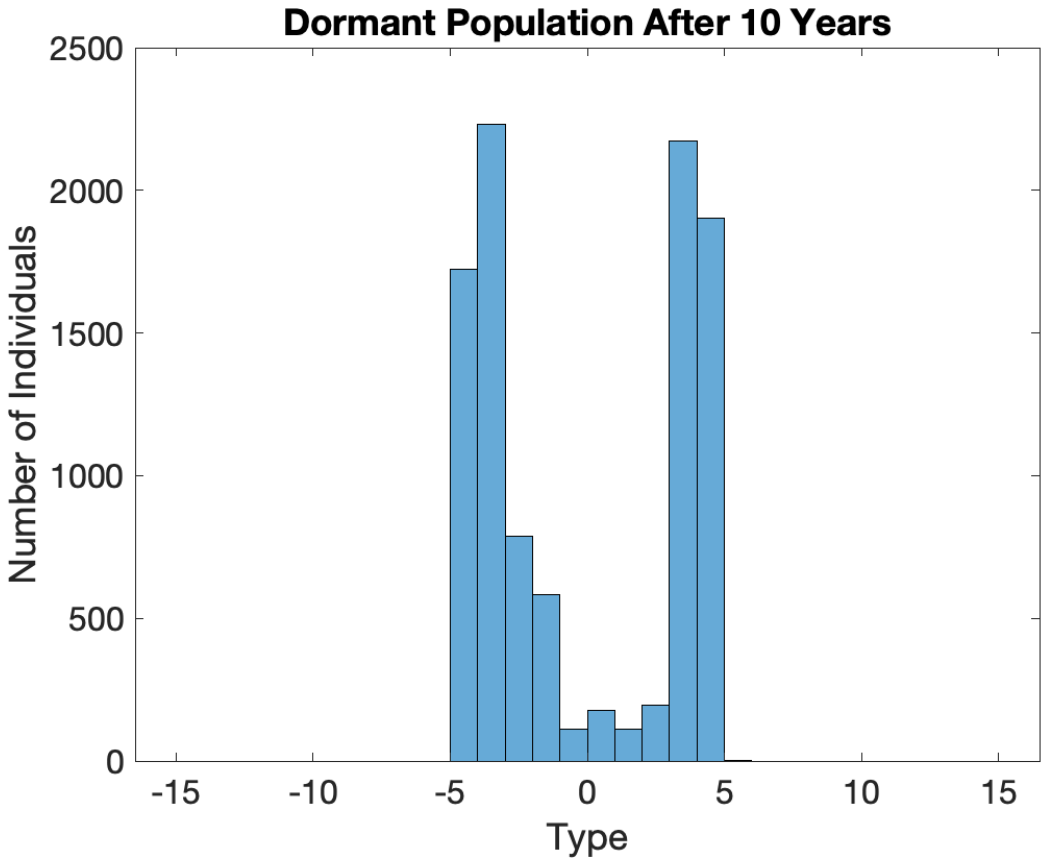}
\includegraphics[scale=0.212, trim={1cm 6.6cm 2cm 6.8cm}, clip]{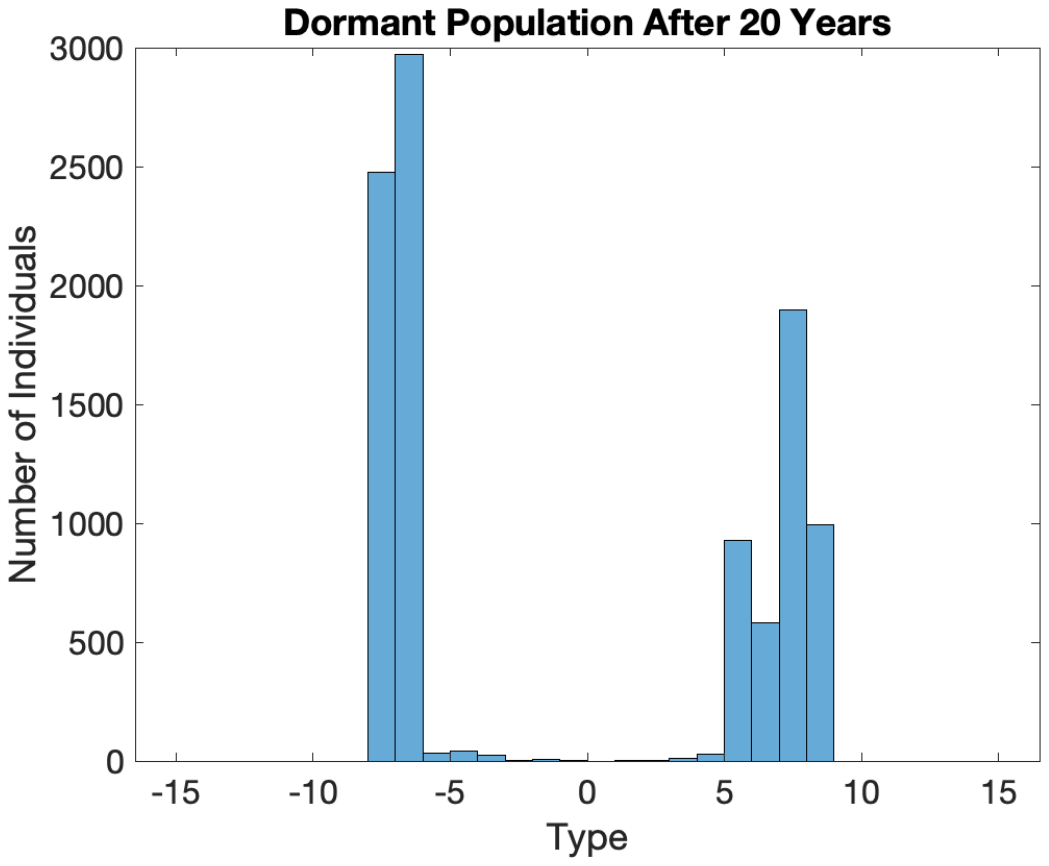} 
\includegraphics[scale=0.212, trim={1cm 6.6cm 2cm 6.8cm}, clip]{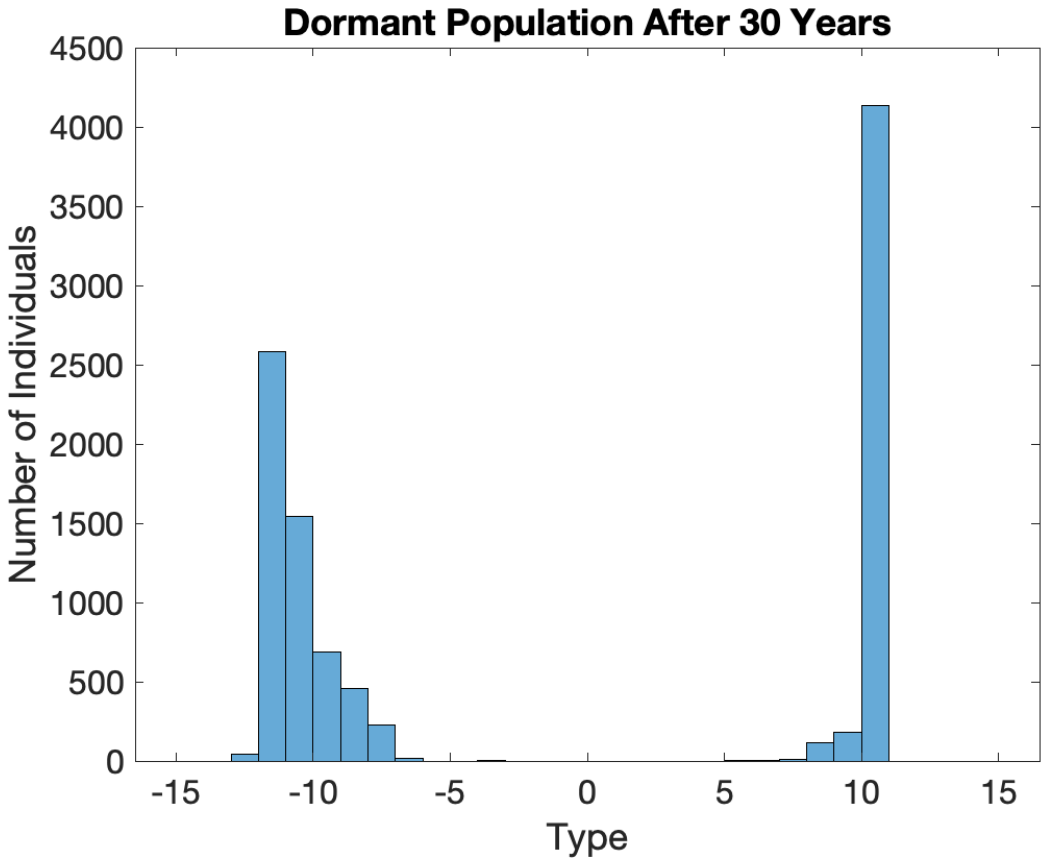}
\includegraphics[scale=0.212, trim={1cm 6.6cm 2cm 6.8cm}, clip]{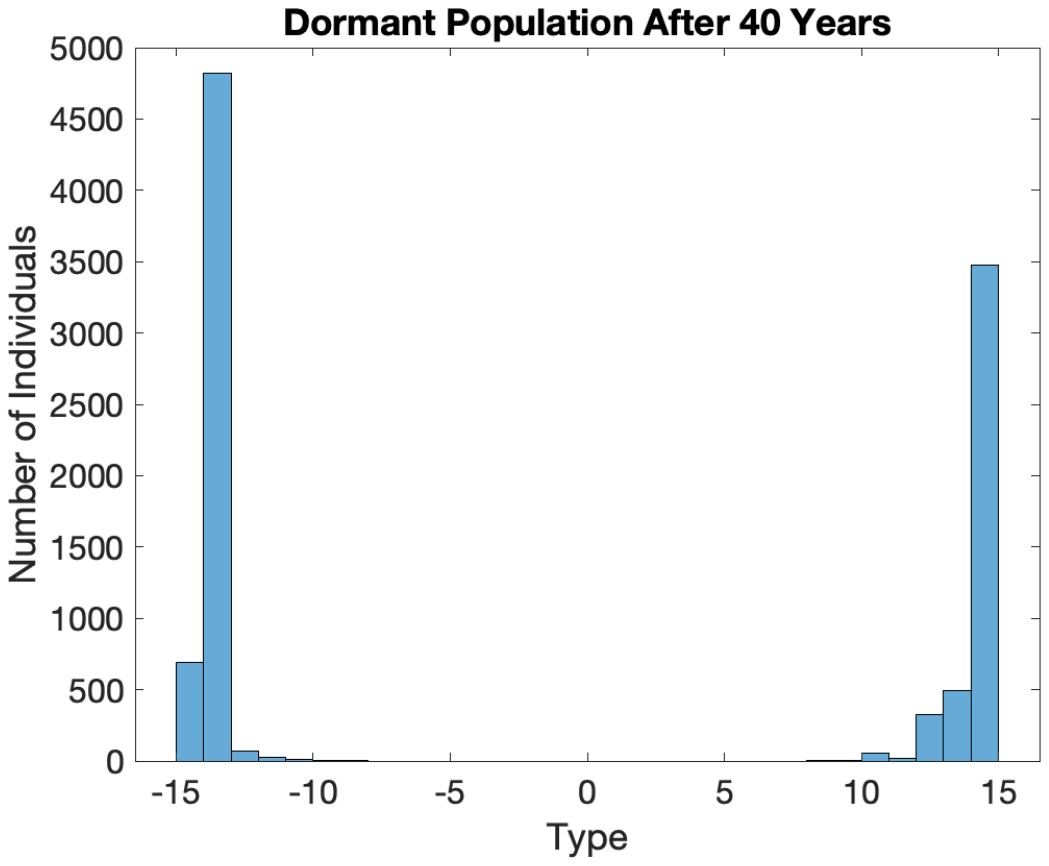}
\end{center}
\caption{{\small The top row shows the type distribution in the active population at the end of the 10th summer, the 10th winter, the 40th summer, and the 40th winter.  The bottom row shows the dormant population after 10 years, 20 years, 30 years, and 40 years.  We see the emergence of two subpopulations, whose genetic distance gets further apart over time.}}
\label{actdorm}
\end{figure}

\subsection{Organization of the paper}

The rest of this paper is organized as follows.  In Section \ref{const}, we give a construction of the population process which will be used throughout the proofs.  In Section \ref{proofsactive}, we prove Theorems \ref{Poissontheo}, \ref{activethm}, and \ref{mztheorem}, all of which pertain to the active population.  In Section \ref{proofsdormant}, we prove 
Theorems \ref{dormantthm1} and \ref{dormantthm2}, which describe to the composition of the dormant population.  A table of notation appears at the end of the paper.

\section{Construction of the process}\label{const}

In this section, we give a careful construction of the population process.  The construction that we give is not the simplest possible but instead is designed to facilitate a coupling with a branching process.  In every generation, the individuals in the active population will be labeled by the integers $1, \dots, N$, and the individuals in the dormant population will be labeled by the integers $1, \dots, K_N$.  The following random variables will be needed for the construction:
\begin{enumerate}
\item For every $h \in \{1, \dots, N\}$, every $i \in \{1, \dots, N - c_N\}$, and every $m \in \N_0$ we have a random variable $U_{h,i,m,N}$.  The random variables $U_{h,i,m,N}$ are independent and have a uniform distribution on $(0,1)$.  These random variables will be used to identify the parents of the active individuals in generation $m+1$.

\item For every $i \in \{1, \dots, N - c_N\}$ and every $m \in \N_0$, we have a random variable $\xi_{i,m,N}$ such that $P(\xi_{i,m,N} = 1) = P(\xi_{i,m,N} = -1) = \mu_N$, and $P(\xi_{i,m,N} = 0) = 1 - 2 \mu_N$.  The random variables $\xi_{i,m,N}$ are independent of one another, and independent of the random variables $U_{h,i,m,N}$.  The random variables $\xi_{i,m,N}$ will be used to determine whether the active individual with label $i$ in generation $m+1$ acquires a mutation.

\item For every $m \in \N_0$, we have a random ordered $c_N$-tuple $(\mathcal{S}^a_{1,m,N}, \dots, \mathcal{S}^a_{c_N, m,N})$ consisting of $c_N$ distinct integers chosen uniformly at random from the set $\{1, \dots, N\}$, and placed in random order.  We also have another random ordered $c_N$-tuple $(\mathcal{S}^d_{1,m,N}, \dots \mathcal{S}^d_{c_N,m,N})$ consisting of $c_N$ distinct integers chosen uniformly at random from the set $\{1, \dots, K_N\}$, and placed in random order.  These random $c_N$-tuples are independent of each other, and independent of the random variables $U_{h,i,m,N}$ and $\xi_{i,m,N}$.  They will be used to identify which individuals move between the active and dormant populations.
\end{enumerate}

We denote by $Z^a_{i,m,N}$ the type of the individual labelled $i$ in the active population in generation $m$, and we denote by $Z^d_{i,m,N}$ the type of the individual labelled $i$ in the dormant population in generation $m$.  Therefore,
$$A_{k,N}(m) = \sum_{i=1}^N \1_{\{Z^a_{i,m,N} = k\}}, \qquad D_{k,N}(m) = \sum_{i=1}^{K_N} \1_{\{Z^d_{i,m,N} = k\}}.$$
In generation zero, all individuals have type zero.  That is, we have $Z^a_{i,0,N} = 0$ for all $i \in \{1, \dots, N\}$ and $Z^d_{i,0,N} = 0$ for all $i \in \{1, \dots, K_N\}$.  It remains to describe how, for each nonnegative integer $m$, we obtain generation $m+1$ from generation $m$.

To construct the dormant population in generation $m+1$, for $i \in \{1, \dots, K_N\}$, we let
\begin{displaymath}
Z^d_{i,m+1,N} = \left\{
\begin{array}{ll} Z^d_{i,m,N} & \mbox{ if }i \notin \{\mathcal{S}^d_{1,m,N}, \dots, \mathcal{S}^d_{c_N,m,N}\}  \\
Z^a_{\mathcal{S}^a_{j,m,N}, m,N} & \mbox{ if }i = \mathcal{S}^d_{j,m,N}.
\end{array} \right.
\end{displaymath}
That is, for $i \notin \{\mathcal{S}^d_{1,m,N}, \dots, \mathcal{S}^d_{c_N,m,N}\}$, the individual labeled $i$ remains dormant.  The dormant individuals labeled $\mathcal{S}^d_{1,m,N}, \dots, \mathcal{S}^d_{c_N,m,N}$ move to the active population, while the active individuals labeled $\mathcal{S}^a_{1,m,N}, \dots, \mathcal{S}^a_{c_N,m,N}$ move to the dormant population.

Constructing the active population is more involved and proceeds in four steps:

\begin{enumerate}
\item If $i = N - c_N + j$, where $j \in \{1, \dots, c_N\}$, then let
$$Z^a_{i, m+1,N} = Z^d_{\mathcal{S}^d_{j,m,N},m,N}.$$
That is, the active individuals with the last $c_N$ labels in generation $m+1$ will have been dormant in generation $m$.

\item We assign marks to the $N$ active individuals in generation $m$.  These marks will become important when we couple the population process with a branching process.  The marks are a permutation of the labels which orders the individuals by fitness, so that the individuals with the highest fitness come first.  More precisely, the individual with label $i$ in generation $m$ has mark $\sigma_{m,N}(i)$, where $\sigma_{m,N}$ is the unique permutation of $\{1, \dots, N\}$ such that $\sigma_{m,N}(i) < \sigma_{m,N}(j)$ if and only if either $R_N(m) Z^a_{i,m,N} > R_N(m) Z^a_{j,m,N}$ or $Z^a_{i,m,N} = Z^a_{j,m,N}$ and $i < j$.  In particular, note that if $j$ individuals in generation $m$ all belong to the fittest type in the population, then their marks will be $1, \dots, j$.  We denote by $\pi_{m,N}$ the inverse of $\sigma_{m,N}$, so that $\pi_{m,N}(h)$ is the label of the individual with mark $h$.

\item We next identify the parents of the individuals labeled $1, \dots, N - c_N$ in generation $m+1$.  For $h \in \{1, \dots, N\}$, let $F_{h,N}(m)$ be the fitness of the individual with the mark $h$ in generation $m$, so $$F_{h,N}(m) = (1 + s_N)^{R_N(m) Z^a_{\pi_{m,N}(h),m,N}}.$$  Then let
\begin{equation}\label{Thetadef}
\Theta_{h,N}(m) = \frac{F_{h,N}(m)}{\sum_{h'=h}^N F_{h',N}(m)}.
\end{equation}
For $i \in \{1, \dots, N - c_N\}$, let $$Q_{i,m,N} = \min\{h: U_{h,i,m,N} \leq \Theta_{h,N}(m)\},$$
which will be the mark of the individual in generation $m$ that is the parent of the individual labeled $i$ in generation $m+1$.
To understand why this gives a valid construction, note that we are examining the individuals in generation $m$ one at a time, in the order given by their marks, to decide whether or not that individual should be the parent of the individual labeled $i$ in generation $m+1$.  Because the parent of each individual is chosen with probability proportional to fitness, the probability that the individual with mark $h$ is the parent, conditional on the event that the previous $h-1$ individuals were not chosen, should be the fitness of the individual with mark $h$ divided by the sum of the fitnesses of the individuals with marks $h$ and higher, which is exactly $\Theta_{h,N}(m)$.  We therefore go through the individuals one at a time until we find a value of $h$ for which $U_{h,i,m,N} \leq \Theta_{h,N}(m)$.  Because $\Theta_{N,N}(m) = 1$, this process must terminate.

\item Finally, we assign the types of the individuals labeled $1, \dots, N - c_N$ in generation $m+1$.  This requires taking into account the mutations, which are determined by the random variables $\xi_{i,m,N}$.  The type of the individual labeled $i$ in generation $m+1$ will be the type of the parent plus $\xi_{i,m,N}$.  More precisely, for $i \in \{1, \dots, N - c_N\}$, we define
\begin{equation}\label{Zadef}
Z^a_{i,m+1,N} = Z^a_{\pi_{m,N}(Q_{i,m,N}), m,N} + \xi_{i,m,N}.
\end{equation}
\end{enumerate}

Finally, for $m \in \N_0$, let ${\cal G}_{m,N}$ denote the $\sigma$-field generated by the random variables $U_{h,i,m,N}$, $\xi_{i,m,N}$, $\mathcal{S}^a_{i,m,N}$, and $\mathcal{S}^d_{i,m,N}$, and let ${\cal F}_{m,N}$ be the $\sigma$-field generated by ${\cal G}_{0,N}, {\cal G}_{1,N}, \dots, {\cal G}_{m-1,N}$.
Note that the random variables $Z^a_{i,m,N}$ and $Z^d_{i,m,N}$ are all ${\cal F}_{m,N}$-measurable.  

\section{Proofs of results for the active population}\label{proofsactive}

In this section, we prove Theorems \ref{Poissontheo} and \ref{activethm}, both of which pertain to the composition of the active population.  We will fix $t_0 > 0$ and focus on understanding the evolution of the active population from generation 0 until generation $\lfloor \rho_N^{-1} t_0 \rfloor$.  Let $${\cal T}_N = \{0, 1, 2, \dots, \lfloor \rho_N^{-1} t_0 \rfloor\}$$ be the set of generations being considered.  A consequence of (\ref{fastmutations}) and (\ref{hlower}) is that
\begin{equation}\label{totalgen}
\lfloor \rho_N^{-1} t_0 \rfloor \ll N^{b-1},
\end{equation}
which means the number of generations we are considering is bounded above by a power of $N$.  To distinguish between the summer and winter seasons, we will also define
$${\cal T}_N^+ = \{m \in {\cal T}_N: R_N(m) = 1\}, \qquad \quad {\cal T}_N^- = \{m \in {\cal T}_N: R_N(m) = -1\}.$$

To understand how the active population evolves over time, it will be helpful to divide the generations into intervals of three different classes.  Let $(r_N)_{N=1}^{\infty}$ be a sequence such that
\begin{equation}\label{rNdef}
\log N \ll r_N \ll U_N.
\end{equation}
and
\begin{equation}\label{Kcr}
\lim_{N \rightarrow \infty} \frac{c_N^2 r_N^2}{K_N \log N} = \infty.
\end{equation}
which is possible because \eqref{largeU} and \eqref{KcU} hold.  Let
\begin{align*}
{\cal M}_N^+ &= \{m: \mbox{for some $i$, either $R_N(m) = \xi_{i,m,N} = 1$, or $R_N(m) = \xi_{i,m,N} = -1$}\}, \\
{\cal U}_N &= \{m: \mbox{for some nonnegative integer $j$, either $m = jU_N$ or $m = jU_N + V_N$}\}.
\end{align*}
Note that ${\cal M}_N^+$ consists of all generations in which some active individual acquires a beneficial mutation, and ${\cal U}_N$ consists of all generations when the season changes.  We next define the three classes of intervals.
\begin{enumerate}
\item Let $J_{1,N}$ consist of intervals lasting for $r_N$ generations after each mutation that increases fitness.  That is, we have $n \in J_{1,N}$ if $n = m + \ell$ for some $m \in {\cal M}_N^+$ and $\ell \in \{1, \dots, r_N\}$.

\item Let $J_{2,N}$ consist of the intervals lasting for $r_N$ generations at the beginning of every season.  That is, we have $n \in J_1$ if $n = m + \ell$ for some $m \in {\cal U}_N$ and $\ell \in \{0, \dots, r_N - 1\}$.

\item Let $J_{3,N} = \{0, 1, 2, \dots\} \setminus (J_{1,N} \cup J_{2,N})$ consist of all generations not in $J_{1,N}$ or $J_{2,N}$.  These generations form intervals that end with a mutation or a seasonal change.
\end{enumerate}
For $i \in \{1, 2, 3\}$, we call an interval of consecutive positive integers in $J_{i,N}$ a class $i$ interval.  Note that when $m \in \mathcal{M}_N^+$, the class 1 interval begins in generation $m+1$ rather than generation $m$ because, from \eqref{Zadef}, when $\xi_{i,m,N} \neq 0$, the mutation affects the type of the $i$th individual in generation $m+1$.

Define a sequence $(L_N)_{N=1}^{\infty}$ of positive integers such that
\begin{equation}\label{LNdef}
\lim_{N \rightarrow \infty} L_N = \infty, \qquad L_N! \ll \sqrt{N/c_N}, \qquad L_N! \ll \log N,
\end{equation}
which is possible by (\ref{NcN}).  Then define a sequence $(\eps_N)_{N=1}^{\infty}$ such that
\begin{equation}\label{epsNdef}
\lim_{N \rightarrow \infty} \eps_N = 0, \quad \lim_{N \rightarrow \infty} \eps_N^3 L_N = \infty, \quad \eps_N \gg \frac{c_N}{N}, \quad \eps_N \gg \frac{\log N}{r_N},
\end{equation}
which is possible by (\ref{NcN}) and (\ref{rNdef}).  Note that $L_N$ tends to infinity very slowly and $\eps_N$ tends to zero very slowly as $N \rightarrow \infty$.

\begin{Dfn}
For $k \in \Z$ and $m \in \N_0$, we say that type $k$ {\em dominates} the population in generation $m$ if the following hold:
\begin{enumerate}
\item We have $A_{k,N}(m) \geq (1 - \eps_N)N$.

\item Either $R_N(m) = 1$ and $A_{\ell,N}(m) = D_{\ell,N}(m) = 0$ for all $\ell > k$, or $R_N(m) = -1$ and $A_{\ell,N}(m) = D_{\ell,N}(m) = 0$ for all $\ell < k$.  That is, no type with higher fitness than type $k$ is present in the active or dormant population.
\end{enumerate}
We denote by $\Lambda_{k,m,N}$ the event that type $k$ dominates the population in generation $m$.
\end{Dfn}

Very roughly, the key to understanding the evolution of the active population will be to establish that in most generations, some type dominates the population.  In generation zero, the population is dominated by type $0$.  The dominant type in the population changes rapidly following a mutation or seasonal change.  More specifically:
\begin{itemize}
\item During class 1 intervals, a new beneficial mutation appears.  This beneficial mutation may die out, or it may become dominant in the population within $r_N$ generations.

\item During class 2 intervals, the season changes.  Within the first $r_N$ generations after the beginning of summer, the type that dominated the population at the end of the previous summer emerges from the seed bank and once again becomes dominant.  An analogous result holds during the winter.

\item During class 3 intervals, whichever type dominated the population at the beginning of the interval continues to dominate until the next beneficial mutation or seasonal change.
\end{itemize}

In subsection \ref{LDsub}, we collect some large deviations bounds that will be useful later.  We establish some results about the times of mutations in subsection \ref{mutsub}.  In subsections~\ref{3sub}, \ref{1sub}, and \ref{2sub}, we study the behavior of the population during class 3, class 1, and class 2 intervals respectively.  Finally, in subsection \ref{thmsec}, we prove Theorems \ref{Poissontheo}, \ref{activethm}, and \ref{mztheorem}.

\subsection{Some large deviations bounds}\label{LDsub}

We first record the following standard binomial tail bound.  The result (\ref{chernoff1}) was proved in \cite{ok58} and follows directly from ideas in \cite{ch52}.  See also Remark 2.5 of \cite{jlr00}.  The result (\ref{chernoff2}) appears, for example, as Corollary 2.3 in \cite{jlr00}.

\begin{Lemma}
Let $Z$ have a binomial distribution with parameters $n$ and $p$.  If $c > 0$, then
\begin{equation}\label{chernoff1}
P(|Z - np| \geq cn) \leq 2e^{-2nc^2}.
\end{equation}
If $0 < c < 3/2$, then
\begin{equation}\label{chernoff2}
P(|Z - np| \geq c np) \leq 2e^{-c^2 np/3}.
\end{equation}
\end{Lemma}

Next, we record the following result about negatively correlated Bernoulli random variables.  Recall that events $A_1, \dots, A_m$ are said to be negatively correlated if for all $I \subseteq \{1, \dots, m\}$, we have
\begin{equation}\label{negcordef}
P \bigg( \bigcap_{i \in I} A_i \bigg) \leq \prod_{i \in I} P(A_i).
\end{equation}

\begin{Lemma}\label{negcorlem}
Suppose $X_1, \dots, X_n$ are $\{0,1\}$-valued random variables, and let $X = X_1 + \dots + X_n$.  Then the following hold:
\begin{enumerate}
\item If the events $\{X_1 = 1\}, \dots \{X_n = 1\}$ are negatively correlated and $0 < \eps \leq 1$, then $$P(X > (1 + \eps)E[X]) \leq e^{-\eps^2 E[X]/3}.$$

\item If the events $\{X_1 = 0\}, \dots, \{X_n = 0\}$ are negatively correlated and $0 < \eps \leq 1$, then $$P(X < (1 - \eps)E[X]) \leq e^{-\eps^2 E[X]/2}.$$

\item If the events $\{X_1 = 1\}, \dots \{X_n = 1\}$ are negatively correlated and $\eps > 1$, then $$P(X > (1 + \eps)E[X]) \leq e^{-\eps E[X]/3}.$$
\end{enumerate}
\end{Lemma}

\begin{proof}
Parts 1 and 2 of the lemma are equivalent to Lemma 3 in \cite{molloy}.  Part 3 of the lemma follows directly from Theorem 3.4 and equation (1) of \cite{ps97}.  Indeed, when $\eps > 1$, from Theorem 3.4 and equation (1) of \cite{ps97}, we get $$P(X > (1 + \eps)E[X]) \leq \bigg( \frac{e^{\eps}}{(1 + \eps)^{1 + \eps}} \bigg)^{E[X]}.$$
To establish part 3 of the lemma, it therefore suffices to show that $e^{\eps}/(1 + \eps)^{1+\eps} \leq e^{-\eps/3}$, or equivalently that $e^{4 \eps/3} \leq (1 + \eps)^{1 + \eps}$.  Taking logarithms of both sides, we see it suffices to show that $4/3 \leq (1 + 1/\eps) \log(1 + \eps)$ for $\eps \geq 1$.  This result holds numerically when $\eps = 1$ and therefore holds for all $\eps \geq 1$ because the right-hand side is increasing over $\eps \in [1, \infty)$.
\end{proof}

The following lemma will be useful for studying the times when individuals switch between the active and dormant populations.

\begin{Lemma}\label{geomlem}
Let $c$, $J$, and $K$ be positive integers such that $c \leq K$ and $J \leq K$.  Let $S_1, S_2, \dots$ be independent $c$-element subsets of $\{1, \dots, K\}$, chosen uniformly at random.  For $1 \leq k \leq K$, let $M_k = \min\{j: k \in S_j\}.$  Let $T \subseteq \N$, and let $Y_T = \sum_{k=1}^J \1_{\{M_k \in T\}}$.  Then $$P(|Y_T - E[Y_T]| > \eps E[Y_T]) \leq 2e^{-(\eps^2 \wedge \eps)E[Y_T]/3}.$$
\end{Lemma}

\begin{proof}
For $k \in \{1, \dots, J\}$, let $A_k = \{M_k \in T\}$.  We will show that the events $A_1, \dots, A_J$ and the events $A_1^c, \dots A_J^c$ are negatively correlated.  The result then follows from Lemma \ref{negcorlem}.

Because the labeling of the integers is arbitrary, it suffices to show that for $1 \leq k \leq J$, we have $P(A_1 \cap \dots \cap A_k) \leq P(A_1) \dots P(A_k)$ and $P(A_1^c \cap \dots \cap A_k^c) \leq P(A_1^c) \dots P(A_k^c)$.  By induction, it then suffices to show that the events $A_1 \cap \dots \cap A_{k-1}$ and $A_k$ are negatively correlated, and that the events $A_1^c \cap \dots \cap A_{k-1}^c$ and $A_k^c$ are negatively correlated.

To prove this result, we fix $k \in \{2, \dots, J\}$.  We construct the sets $S_j$ using i.i.d. random variables $U_{i,j}$ having the uniform distribution on $[0,1)$, where $1 \leq i \leq c$ and $j \geq 1$.  The random variable $U_{i,j}$ will tell us how to pick the $i$th element of the set $S_j$.    There will be $K - i + 1$ elements from which to choose, so we will put these elements in a specific order and choose the $\ell$th element if $(\ell - 1)/(K-i+1) \leq U_{i,j} < \ell/(K-i+1)$.  Note that regardless of how these elements are ordered, this scheme ensures that each one is chosen with probability $1/(K - i + 1)$, which makes the sets $S_j$ uniform $c$-element subsets of $\{1, \dots, K\}$.

To order the elements, first suppose $j \in T$.  We begin by listing in increasing order any elements of $\{1, \dots, k-1\}$ that were not in any of the sets $S_1, \dots, S_{j-1}$ and were not among the first $i-1$ elements of $S_j$.  We will list $k$ last if $k$ was not in one of the sets $S_1, \dots, S_{j-1}$ and was not among the first $i-1$ elements of $S_j$.  In between, we list all other integers in increasing order.  When $j \notin T$, we use exactly the reverse ordering, so that $k$ will come first if $k$ was not in one of the sets $S_1, \dots, S_{j-1}$, and the elements of $\{1, \dots, k-1\}$ that were not in any of the sets $S_1, \dots, S_{j-1}$ and were not among the first $i-1$ elements of $S_j$ will come last.

The events $A_1, \dots, A_k$ are now deterministic functions of the random variables $U_{i,j}$.  Furthermore, the ordering has been constructed so that the indicators of the events $A_1 \cap \dots \cap A_{k-1}$ and $A_k^c$ are decreasing functions of the $U_{i,j}$, while the indicators of the events $A_1^c \cap \dots \cap A_{k-1}^c$ and $A_k$ are increasing functions of the $U_{i,j}$.  The result when $T$ is finite then follows from the fact that increasing (or decreasing) functions of independent random variables are positively correlated.  This is a version of Harris' Inequality and appears, for example, as Theorem 2.1 in \cite{epw67}.

The result when $T$ is infinite follows easily by taking limits because if we define the event $A_{k,n} = \{M_k \in T \cap \{1, \dots, n\}\}$, then $\lim_{n \rightarrow \infty} P(A_k \cap A_{k,n}^c) = 0$.  Consequently, because because $A_{1,n} \cap \dots \cap A_{k-1,n}$ and $A_{k,n}$ are negatively correlated for every $n$, the events $A_1, \dots, A_{k-1}$ and $A_k$ are negatively correlated.  The same reasoning implies that $A_1^c \cap \dots \cap A_{k-1}^c$ and $A_k^c$ are negatively correlated.
\end{proof}

\begin{Rmk}\label{geomrem}
Note that in Lemma \ref{geomlem}, the distribution of $M_k$ is geometric with parameter $c/K$, which means $$E[Y_T] = J P(M_1 \in T) = J \sum_{n \in T} \frac{c}{K} \bigg(1 - \frac{c}{K} \bigg)^{n-1} = \frac{cJ}{K} \sum_{n \in T} \bigg(1 - \frac{c}{K} \bigg)^{n-1}.$$
\end{Rmk}

\subsection{The times of mutations}\label{mutsub}

In this section, we establish some results about the times of mutations.  We will let 
$${\cal M}_N = \{m: \xi_{i,m,N} \neq 0 \mbox{ for some }i\}$$ be the set of all generations in which a mutation appears.

Lemma \ref{mutlem1} shows that with high probability, there will be no more than one mutation in any generation, and Lemma \ref{mutlem2} shows that mutations will be separated in time by more than $2r_N$ generations.  Lemma \ref{mutlem3} shows that with high probability, no mutation will occur within $2r_N$ generations of a seasonal change.  Finally, Lemma \ref{mutlem4} shows that if one type dominates the population, then it is unlikely that any other type will have an offspring that acquires a mutation.

\begin{Lemma}\label{mutlem1}
Let ${\cal M}_N^* = \{m: \mbox{ there exist $i \neq j$ such that $\xi_{i,m,N} \neq 0$ and $\xi_{j,m,N} \neq 0$}\}$
be the set of all generations in which two or more individuals acquire a mutation.  Then
$$\lim_{N \rightarrow \infty} P({\cal M}^*_N \cap {\cal T}_N = \emptyset) = 1.$$
\end{Lemma}

\begin{proof}
In each generation, there are $N - c_N$ individuals in the active population that each acquire a mutation with probability $2 \mu_N$.  The probability that two or more of these individuals acquire a mutation in generation $m$ is bounded above by $\binom{N - c_N}{2} (2 \mu_N)^2$, which is bounded above by $2 N^2 \mu_N^2$.  Therefore, using Boole's Inequality, the probability that there are two or more mutations in some generation is bounded above by $(1 + \rho_N^{-1} t_0) \cdot 2 N^2 \mu_N^2$.  This expression tends to zero as $N \rightarrow \infty$ by (\ref{rhoNdef}), (\ref{hlower}), and (\ref{Nmubound}).
\end{proof}

\begin{Lemma}\label{mutlem2}
We have $$\lim_{N \rightarrow \infty} P(\mbox{for all }m_1, m_2 \in {\cal M}_N \cap {\cal T}_N \mbox{ such that }m_1 \neq m_2, \mbox{ we have }|m_1 - m_2| > 2r_N) = 1.$$
\end{Lemma}

\begin{proof}
The number of pairs of integers $(m_1, m_2)$ such that $m_1 \in {\cal T}_N$ and $m_1 < m_2 \leq m_1 + 2r_N$ is bounded above by $(1 + \rho_N^{-1} t_0)(2 r_N)$.  The probability that there are mutations in both generations $m_1$ and $m_2$ is bounded above by $(2N \mu_N)^2$, so the probability that this happens for at least one of the pairs $(m_1, m_2)$ is at most $(1 + \rho_N^{-1} t_0) (2 r_N) (2N \mu_N)^2$.  This expression tends to zero as $N \rightarrow \infty$ because $r_N N \mu_N \ll U_N N \mu_N \lesssim 1$ by (\ref{rNdef}) and (\ref{UNmu}), and $\omega(s_N)$ is bounded away from zero by (\ref{hlower}).
\end{proof}

\begin{Lemma}\label{mutlem3}
We have
$$\lim_{N \rightarrow \infty} P(\mbox{for all }m_1 \in {\cal M}_N \cap {\cal T}_N \mbox{ and }m_2 \in {\cal U}_N \cap {\cal T}_N, \mbox{ we have }|m_1 - m_2| > 2r_N) = 1.$$
\end{Lemma}

\begin{proof}
The probability of a mutation in a given generation is bounded above by $2 N \mu_N$.  Because there are two season changes each year, the number of season changes before generation $\lfloor \rho_N^{-1} t_0 \rfloor$ is at most $1 + 2 U_N^{-1} \rho_N^{-1} t_0$, so the number of generations $m$ that are within $2r_N$ generations of a season change is at most $(4 r_N + 1) (1 + 2U_N^{-1}\rho_N^{-1} t_0)$.  Therefore, the probability of a mutation in one of these generations is at most $$(2 N \mu_N) (4r_N + 1) \bigg(1 + \frac{2t_0}{U_N \rho_N} \bigg).$$
Because (\ref{Reg1Asm2}) and (\ref{Reg2Asm}) imply $U_N \rho_N \lesssim 1$, this expression is bounded above by a constant multiple of $N \mu_N r_N/(U_N \rho_N)$,
which tends to zero as $N \rightarrow \infty$ because $r_N \ll U_N$ by (\ref{rNdef}), and $N \mu_N/\rho_N$ is bounded above by (\ref{rhoNdef}) and (\ref{hlower}).
\end{proof}

\begin{Lemma}\label{mutlem4}
Let $G_{0,N}^*$ be the event that for all $m \in {\cal T}_N$ and all $k \in \Z$ such that $\Lambda_{k,m,N}$ occurs, we have $\xi_{i,m,N} = 0$ for all $i \in \{1, \dots, N - c_N\}$ for which $Z^a_{\pi_{m,N}(Q_{i,m,N}), m,N} \neq k$.  Then $$\lim_{N \rightarrow \infty} P(G_{0,N}^*) = 1.$$
\end{Lemma}

\begin{proof}
Suppose $\Lambda_{k,m,N}$ occurs.  Then $X_{k,N}(m) \geq 1 - \eps_N$, and no individual in generation $m-1$ has higher fitness than the individuals of type $k$.  Therefore, the $N - c_N$ individuals in the active population in generation $m+1$ whose parent comes from the active population in generation $m$ each choose a parent of type $k$ with probability at least $1 - \eps_N$, and therefore choose a parent of a different type with probability at most $\eps_N$.  Each of these individuals independently acquires a mutation with probability $2 \mu_N$.  Therefore, the probability that an individual whose parent does not have type $k$ acquires a mutation is bounded above by $2 \eps_N N \mu_N$.  Using Boole's Inequality, we can sum this probability over all generations to get
$$P((G_{0,N}^*)^c) \leq (1 + \rho_N^{-1} t_0) \cdot 2 \eps_N N \mu_N.$$
The result now follows from (\ref{hlower}), (\ref{Nmubound}), and (\ref{epsNdef}).
\end{proof}

We now combine these results in the following corollary.  Note that on the event $G_{0,N}$ defined below, the class 1 and class 2 intervals are well separated in time.  For $N$ large enough that $V_N > 2 r_N$, the generations in ${\cal T}_N$ can be divided into class 1 and class 2 intervals having length exactly $r_N$, separated by class 3 intervals of length at least $r_N$.

\begin{Cor}\label{mutcor}
Let $G_{0,N}$ be the event that $G_{0,N}^*$ occurs, and the events described in the statements of Lemmas \ref{mutlem1}, \ref{mutlem2}, and \ref{mutlem3} all occur.  Then $$\lim_{N \rightarrow \infty} P(G_{0,N}) = 1.$$
\end{Cor}

\subsection{Class 3 intervals: Persistence of a dominant type}\label{3sub}

In this subsection, we show that once a type becomes dominant in the population, with high probability it stays dominant until the next beneficial mutation or seasonal change.  This will allow us to establish that, during class 3 intervals, whichever type dominates the population at the beginning of the interval will continue to dominate until the end of the interval.

The next lemma establishes that if at least a fraction $1 - \eps_N$ of the active population during the summer is comprised of individuals of type $k$ or higher, then with high probability this will continue to be the case in the next generation.  Of course, an analogous result holds for the winter.  

\begin{Lemma}\label{kstayhigh}
There exists a constant $\gamma > 0$ such that for all positive integers $k$ and all positive integers $m$ for which $R_N(m) = 1$, we have
$$P\bigg(\sum_{\ell = k}^{\infty} A_{\ell,N}(m+1) < (1 - \eps_N)N \Big|{\cal F}_{m,N} \bigg) \leq 2e^{-\gamma N \eps_N^2} \quad\mbox{on }\bigg\{\sum_{\ell = k}^{\infty} A_{\ell,N}(m) \geq (1 - \eps_N)N \bigg\}$$ for sufficiently large $N$.
\end{Lemma}

\begin{proof}
According to the model, $N - c_N$ individuals in the active population in generation $m+1$ will choose their parent at random from the active population in generation $m$, with probability proportional to fitness, and inherit the same type as their parent with probability $1 - 2 \mu_N$.  In generation $m$, on the event that $\sum_{\ell = k}^{\infty} A_{\ell, N}(m) \geq (1 - \eps_N)N$, for some $j \geq (1 - \eps_N)N$ there are $j$ individuals with fitness at least $(1 + s_N)^k$ and $N - j$ individuals with fitness at most $(1 + s_N)^{k-1}$.  Therefore, an individual in generation $m+1$ whose parent was active in the previous generation will be type $k$ or higher with probability at least
$$\frac{j(1 + s_N)^k}{j(1 + s_N)^k + (N-j)(1 + s_N)^{k-1}} \cdot (1 - 2 \mu_N) = \frac{j(1 + s_N)(1 - 2\mu_N)}{N + js_N}.$$
Because the expression on the right-hand side is an increasing function of $j$, this probability is greater than or equal to
\begin{equation}\label{fitprob}
\frac{N(1 - \eps_N)(1 + s_N)(1 - 2\mu_N)}{N + N (1 - \eps_N)s_N} = \frac{(1 - \eps_N)(1 + s_N)(1 - 2\mu_N)}{1 + (1 - \eps_N)s_N}.
\end{equation}
Let $$s_N^* = \frac{s_N}{1 + s_N}.$$
Because
\begin{equation}\label{sstareq}
\frac{1 + s_N}{1 + (1 - \eps_N) s_N} = 1 + \frac{\eps_N s_N}{1 + (1 - \eps_N)s_N} \geq 1 + \frac{\eps_N s_N}{1 + s_N} = 1 + \eps_N s_N^*,
\end{equation}
the probability in (\ref{fitprob}) is greater than or equal to
$$q_N := (1 - \eps_N)(1 + \eps_N s_N^*)(1 - 2\mu_N).$$
Thus, on the event that $\sum_{\ell = k}^{\infty} A_{\ell, N}(m) \geq (1 - \eps_N)N$, the random variable $\sum_{\ell = k}^{\infty} A_{\ell,N}(m+1)$ stochastically dominates a random variable $Z_N$ having a binomial$(N - c_N, q_N)$ distribution.
It follows that on the event $\sum_{\ell = k}^{\infty} A_{\ell, N}(m) \geq (1 - \eps_N)N$, we have
$$P\bigg(\sum_{\ell = k}^{\infty} A_{\ell,N}(m+1) < (1 - \eps_N)N \Big|{\cal F}_{m,N} \bigg) \leq P(Z_N \leq (1 - \eps_N)N).$$
Note that $$E[Z_N] - (1 - \eps_N)N = N(1 - \eps_N) \Big( \Big(1 - \frac{c_N}{N} \Big)(1 + \eps_N s_N^*)(1 - 2 \mu_N) - 1 \Big).$$
Because $\eps_N s_N \gg \mu_N$ and $\eps_N s_N \gg c_N/N$ by \eqref{slarge}, \eqref{Nmubound}, and \eqref{epsNdef}, and $(s_N^*)_{N=1}^{\infty}$ is bounded away from zero, there is a constant $\kappa > 0$ such that for sufficiently large $N$, we have $$(1 - \eps_N) \Big( \Big(1 - \frac{c_N}{N} \Big)(1 + \eps_N s_N^*)(1 - 2 \mu_N) - 1 \Big) \geq \kappa \eps_N.$$
In particular, we have $E[Z_N] - (1 - \eps_N) \geq N \kappa \eps_N$.  
Therefore, it now follows from (\ref{chernoff1}) that
for sufficiently large $N$, on the event $\sum_{\ell = k}^{\infty} A_{\ell, N}(m) \geq (1 - \eps_N)N$, we have
\begin{align*}
P\bigg(\sum_{\ell = k}^{\infty} A_{\ell,N}(m+1) < (1 - \eps)N \Big|{\cal F}_{m,N} \bigg) &\leq P \big( |Z_N - E[Z_N]| \geq E[Z_N] - (1 - \eps_N)N \big) \\
&\leq P \big( |Z_N - E[Z_N]| \geq N \kappa \eps_N \big) \\
&\leq 2e^{-2N \kappa^2 \eps_N^2}.
\end{align*}
The result follows if we choose $\gamma = 2 \kappa^2$.
\end{proof}

We will need a more precise result in one special case.  The following lemma shows that if the population consists entirely of type 0 individuals and then a deleterious mutation occurs, with high probability the population will once again have all type 0 individuals after $r_N$ generations.

\begin{Lemma}\label{delmut}
Let $m$ be a positive integer such that $R_N(\ell) = 1$ for all $\ell \in \{m, m+1, \dots, m+r_N\}$.   Define the event
$$\Phi_{m,N} = \{A_{0,N}(m) = N - 1\} \cap \{A_{-1,N}(m) = 1\} \cap \{D_{0,N}(m) = K_N\}.$$
Let $\delta > 0$.  Then for sufficiently large $N$, on the event $\Phi_{m,N}$, we have
$$P\big(\{A_{0,N}(m + r_N) = N\} \cap \{D_{0,N}(m + r_N) = K_N\} \big| {\cal F}_{m,N} \big) > 1 - \delta.$$
\end{Lemma}

\begin{proof}
First note that
\begin{equation}\label{del1}
P \big( {\cal M}_N \cap \{m, m+1, \dots, m+r_N-1\} = \emptyset \big| {\cal F}_{m,N} \big) \geq 1 - 2 N \mu_N r_N,
\end{equation}
which tends to one as $N \rightarrow \infty$ by (\ref{rNdef}) and (\ref{UNmu}).  Therefore, on the event $\Phi_{m,N}$, with high probability type $-1$ remains the least fit type in the population until time $m + r_N$.

Let $$\zeta_{L,m,N} = \min\{n \geq m: A_{-1,N}(n) = 0 \mbox{ or }A_{-1,N}(n) \geq L\}$$ be the first time that the number of individuals of type $-1$ either drops to zero or exceeds $L$.  Let $\kappa_{m,N} = \min\{n \geq m: n \in {\cal M}_N\}$ be the first time after time $m$ that a mutation occurs.  Let $M_N(n) = A_{-1,N}(n)$ on $\Phi_{m,N}$, and let $M_N(n) = 0$ on $\Phi_{m,N}^c$.  Because individuals of type $-1$ have lower fitness than other individuals in the population, at least until time $\kappa_{m,N}$, the process $(M_N(n))_{n=m}^{(m + r_N) \wedge \kappa_{m,N}}$
is a supermartingale.  Therefore, by Doob's Maximal Inequality,
$$P \big( M_N(n) \geq L \mbox{ for some }n \in \{m, m+1, \dots, (m + r_N) \wedge \kappa_{m,N} \} \big| {\cal F}_{m,N} \big) \leq \frac{1}{L}.$$
Choose a positive integer $L$ large enough that $1/L < \delta/4$.  Then on $\Phi_{m,N}$ we have
\begin{equation}\label{del2}
P \big( \{A_{-1,N}(\zeta_{L,m,n}) \geq L\} \cap \{\zeta_{L,m,n} \leq (m + r_N) \wedge \kappa_{m,N} \} \big| {\cal F}_{m,N} \big) \leq \frac{\delta}{4}.
\end{equation}

In each generation, $N - c_N$ individuals in the active population choose their parent at random from the active population in the previous generation with probability proportional to fitness.  As long as there are fewer than $L$ individuals of type $-1$ in the previous generation and no individuals of lower types, the probability that a given individual chooses a type $-1$ individual as its parent is at most $(L-1)/N$.  Therefore, the probability that there are no type $-1$ individuals in the next generation is at least $(1 - (L-1)/N)^{N - c_N}$, which converges to $e^{-(L-1)}$ as $N \rightarrow \infty$ by (\ref{NcN}).  Therefore, if we choose $q$ such that $0 < q < e^{-(L-1)}$ and then choose a positive integer $B_L$ such that $(1 - q)^{B_L} < \delta/4$, on $\Phi_{m,N}$ we have for sufficiently large $N$,
\begin{equation}\label{del3}
P \big( \zeta_{L,m,N} > m + B_L \big| {\cal F}_{m,N} \big) < \frac{\delta}{4}.
\end{equation}

Also, the probability that any particular individual joins the seed bank in the next generation is $c_N/N$.  Therefore, before time $\zeta_{L,m,N}$, the probability that a type $-1$ individual joins the seed bank in any given generation is at most $(L-1)c_N/N$.  It follows that
\begin{equation}\label{del4}
P\big(D_{-1,N}(n) > 0 \mbox{ for all }n \in \{m, m+1, \dots, \zeta_{L,m,N} \wedge (m + B_L)\}\big|{\cal F}_{m,N}\big) \leq \frac{B_L (L-1) c_N}{N}.
\end{equation}

As long as $N$ is large enough that $r_N > B_L$, $2 N \mu_N r_N < \delta/4$, and $B_L(L-1) c_N/N < \delta/4$, it follows from (\ref{del1}), (\ref{del2}), (\ref{del3}), and (\ref{del4}) that with probability at least $1 - \delta$, the type $-1$ individuals all disappear from the population before there are any further mutations, and before any type $-1$ individual migrates to the dormant population.  The result of the lemma follows.
\end{proof}

The following corollary describes how the active population evolves during class 3 intervals.  In particular, it shows that whichever type dominates the population at the beginning of a class~3 interval will continue to dominate the population until the end of the class 3 interval.

\begin{Cor}\label{class3cor}
Let $G_{3,N}$ be the event that whenever $\{m, m+1, \dots, m^*\}$ is a class 3 interval (meaning that $m \in {\cal T}_N$ and $\{m, m+1, \dots, m^*\} \subseteq J_{3,N}$ with $m-1 \notin J_{3,N}$ and $m^* + 1 \notin J_{3,N}$), the following hold:
\begin{enumerate}
\item If $\Lambda_{k,m,N}$ occurs for some $k$, then $\Lambda_{k,n,N}$ occurs for all $n$ such that $m \leq n \leq m^*$.

\item If $A_{0,N}(m) = N$ and $D_{0,N}(m) = K_N$, then $A_{0,N}(m^*) = N$ and $D_{0,N}(m^*) = K_N$.
If, in addition, we have $m^* + 1 \in \mathcal{U}_N$, then $A_{0,N}(m^* + 1) = N$ and $D_{0,N}(m^*+1) = K_N$.
\end{enumerate}
Then $$\lim_{N \rightarrow \infty} P(G_{3,N}) = 1.$$
\end{Cor}

\begin{proof}
Suppose $R_N(m) = 1$.  On the event $\Lambda_{k,m,N}$, we have $A_{k,N}(m) \geq (1 - \eps_N)N$.  Therefore, by Lemma \ref{kstayhigh}, on the event $\Lambda_{k,m,N}$, the probability conditional on ${\cal F}_{m,N}$ that the number of individuals of type $k$ or higher remains at least $(1 - \eps_N)N$ until the first generation of the next winter is at least $1 - 2 V_N e^{-\gamma N \eps_N^2}$.  This probability tends to one as $N \rightarrow \infty$ because $N \eps_N^2 \gg \log N$ by \eqref{LNdef} and \eqref{epsNdef}
and because $V_N \leq U_N \lesssim \rho_N^{-1}$ by (\ref{Reg1Asm2}) and (\ref{Reg2Asm}) and therefore $V_N \ll N^{b-1}$ by (\ref{totalgen}).  Furthermore, because a seasonal change or a mutation that creates a type $k+1$ individual would bring the class 3 interval to an end, the probability that type $k$ dominates the population until the end of the class 3 interval tends to one as $N \rightarrow \infty$.  A similar argument holds when $R_N(m) = -1$.  This proves the first part of the corollary.

To prove the second part of the corollary, suppose $A_{0,N}(m) = N$ and $D_{0,N}(m) = K_N$.  A beneficial mutation would bring the class~3 interval to an end.  The expected number of deleterious mutations in a given season is bounded above by $2 N \mu_N U_N$, which by (\ref{UNmu}) is bounded above by a positive constant.  Furthermore, by Lemma \ref{delmut}, each deleterious mutation disappears from the population within $r_N$ generations with probability at least $1 - \delta$, where $\delta$ is arbitrary, and by Lemmas \ref{mutlem1}, \ref{mutlem2}, and \ref{mutlem3}, with probability tending to one, no two mutations will occur within $r_N$ generations of one another or within $r_N$ generations of the end of the class 3 interval.  It follows that with probability tending to one, all individuals will be type $0$ at the end of the class 3 interval, and in the following generation if $m^*+1 \in \mathcal{U}_N$.  The result follows.
\end{proof}

\subsection{Class 1 intervals: the population following a beneficial mutation}\label{1sub}

In this subsection, we describe what happens after the fittest type in the population acquires a mutation.  We describe what happens during the summer, with the understanding that analogous results hold for the winter season.  

\subsubsection{The initial phase}

Suppose that during the summer, type $k-1$ dominates the population, and then a beneficial mutation produces a type $k$ individual in generation $m$.  
More formally, we will work on the event
\begin{align}\label{PhikmNdef}
\Phi_{k,m,N} &= \{A_{k,N}(m) = 1\} \cap \{A_{\ell, N}(m) = 0 \mbox{ for all }\ell > k\} \nonumber \\
&\qquad \qquad \cap\{A_{k-1,N}(m) \geq (1 - \eps_N)N - 1\} \cap \{D_{\ell,N}(m) = 0 \mbox{ for all }\ell \geq k\}.
\end{align}
Recall the definition of $L_N$ from (\ref{LNdef}).  For positive integers $m$ and $k$, define
\begin{equation}\label{taudef}
\tau_{k,m,N} = \min\{n \geq m: A_{k,N}(n) = 0 \mbox{ or }A_{k,N}(n) \geq L_N\}.
\end{equation}
The following lemma shows that, with high probability, the number of type $k$ individuals either dies out or exceeds $L_N$ after a short time.  Also, during this short time period, no type $k$ individuals enter the seed bank.

\begin{Lemma}\label{1toL}
Let $m$ be a positive integer such that $R_N(n) = 1$ for $n \in \{m, m+1, \dots, m+\lfloor r_N/2 \rfloor\}$.  Let $k$ be a positive integer.  Let $\delta > 0$.  For sufficiently large $N$, we have
\begin{equation}\label{taukL}
P\big(\tau_{k,m,N} \leq m + r_N/2 \big|{\cal F}_{m,N} \big) > 1 - \delta \quad\mbox{on }\Phi_{k,m,N}
\end{equation}
and
\begin{equation}\label{noseed}
P\big(D_{k,N}(n) = 0 \mbox{ for all }n \in \{m, m+1, \dots, \tau_{k,m,N}\}\big|{\cal F}_{m,N}\big) > 1 - \delta \quad\mbox{on }\Phi_{k,m,N}.
\end{equation}
\end{Lemma}

\begin{proof}
The probability that some individual acquires a mutation in one of the generations $m+1, \dots, m + \lfloor r_N/2 \rfloor$ is bounded above by $N \mu_N r_N$, which tends to zero as $N \rightarrow \infty$ by (\ref{rNdef}) and (\ref{UNmu}).  Therefore, we may make our calculations for a modified process in which there are no mutations, and assume that type $k$ remains the fittest type in the population during this time period.

In each generation, $N - c_N$ individuals in the active population choose their parent at random from the active population in the previous generation with probability proportional to fitness.  As long as there is at least one type $k$ individual in the previous generation and no individuals of higher types, any type $k$ individual will be chosen as the parent with probability at least $1/N$.  Therefore, disregarding mutations as noted above, the distribution of the number of type $k$ individuals in the next generation stochastically dominates a binomial($N - c_N, 1/N$) distribution, so the probability that there are $L_N$ or more type $k$ individuals in the next generation is at least
\begin{equation}\label{binLN}
\binom{N - c_N}{L_N} \bigg( \frac{1}{N} \bigg)^{L_N} \bigg(1 - \frac{1}{N} \bigg)^{N - c_N - L_N}.
\end{equation}
Because $c_N \ll N$ by (\ref{NcN}) and $L_N(L_N + c_N)/N \rightarrow 0$ by (\ref{LNdef}), we have $(1 - 1/N)^{N - c_N - L_N} \rightarrow e^{-1}$ and $(N - c_N)(N - c_N - 1) \dots (N - c_N - L_N + 1)/N^{L_N} \rightarrow 1$ as $N \rightarrow \infty$.
Therefore, using that $e^{-1} > 1/2$, the probability in (\ref{binLN}) is bounded below by $1/(2L_N!)$ for sufficiently large $N$.  Let $$d_N = \bigg\lfloor \min \bigg\{ \sqrt{\frac{N}{c_N}}, \log N \bigg\} \bigg\rfloor.$$
Then on the event $\Phi_{k,m,N}$,
\begin{equation}\label{BL}
P\big(\tau_{k,m,N} \leq m + d_N \big|{\cal F}_{m,N} \big) > 1 - \bigg(1 - \frac{1}{2 L_N!} \bigg)^{d_N}.
\end{equation}
Because $d_N/L_N! \rightarrow \infty$ by (\ref{LNdef}) and $d_N \ll r_N$ by (\ref{rNdef}), it follows that (\ref{taukL}) holds.

Furthermore, the probability that any particular individual joins the seed bank in the next generation is $c_N/N$.  Therefore, before time $\tau_{k,m,N}$, the probability that a type $k$ individual joins the seed bank in any given generation is at most $(L_N-1)c_N/N$.  Thus,
\begin{align*}
&P(D_{k,N}(n) \neq 0 \mbox{ for some }n \in \{m, m+1, \dots, \tau_{k,m,N}\}|{\cal F}_{m,N}) \\
&\qquad \qquad \qquad \qquad \leq P(\tau_{k,m,N} > m + d_N|{\cal F}_{m,N}) + \frac{d_N (L_N-1) c_N}{N},
\end{align*}
which tends to zero on $\Phi_{k,m,N}$ by (\ref{BL}) and (\ref{LNdef}).  Equation (\ref{noseed}) follows.
\end{proof}

\subsubsection{A branching process approximation}\label{bpsec}

We still need to estimate the probability that the number of type $k$ individuals in the active population reaches $L_N$ before $0$, when a type $k$ individual appears in generation $m$ as a result of a beneficial mutation.  To do this, we will construct two branching processes which will be coupled with the population process.  One branching process $(W_{m,N}^*(n))_{n = m}^{\infty}$ will bound the number of type $k$ individuals from above with high probability, while the other branching process $(W'_{m,N}(n))_{n = m}^{\infty}$ will bound the number of type $k$ individuals from below with high probability.  Both branching processes will be stopped if the number of individuals reaches $L_N$ or higher.  Recall that throughout this section, we are assuming that the mutation to type $k$ occurs during the summer.

To motivate this construction, we first obtain upper and lower bounds for the probabilities $\Theta_{h,N}(n)$ defined in (\ref{Thetadef}).  We will work on the event
\begin{align*}
B_{k,n,N} &= \{1 \leq A_{k,N}(n) \leq L_N - 1\} \cap \{A_{\ell, N}(n) = 0 \mbox{ for all }\ell > k\} \\
&\qquad \qquad \cap \{A_{k-1,N}(n) + A_{k,N}(n) \geq (1 - \eps_N)N\} \cap \{D_{\ell,N}(m) = 0 \mbox{ for all }\ell \geq k\}.
\end{align*}  
Note that this is the same as $\Phi_{k,n,N}$, except that we allow the number of type $k$ individuals to be anything between $1$ and $L_N-1$.
Suppose the individual with mark $h$ in generation $n$ has type $k$. 
Then on $B_{k,n,N}$, we must have $h \leq L_N - 1$, so there must be at least $(1 - \eps_N)N - (L_N - 2)$ individuals in generation $n$ whose mark is greater than or equal to $h$ and who have type $k-1$ or larger.  It follows that
\begin{equation}\label{Thetaupper}
\Theta_{h,N}(n) \leq \frac{(1 + s_N)^k}{((1 - \eps_N)N - (L_N - 2))(1 + s_N)^{k-1}} \leq \frac{1 + s_N}{N(1 - \eps_N) - L_N}.
\end{equation}
Also, we can upper bound the total fitness of the population by considering the case in which $L_N-1$ individuals have type $k$ and the rest have type $k-1$.  This leads to the inequality
\begin{align}\label{Thetalower}
\Theta_{h,N}(n) &\geq \frac{(1 + s_N)^k}{(L_N - 1) (1 + s_N)^k + (N - L_N + 1) (1 + s_N)^{k-1}} \nonumber \\
&\geq \frac{1 + s_N}{L_N(1 + s_N) + (N-L_N)} \nonumber \\
&= \frac{1 + s_N}{N + L_N s_N}.
\end{align}
We therefore define
\begin{equation}\label{qdef}
q_{N}^* = \frac{1 + s_N}{N(1 - \eps_N) - L_N}, \qquad q_{N}' = \frac{1 + s_N}{N + L_N s_N}.
\end{equation}
Then on the event $B_{k,n,N}$, we have $$q_N' \leq \Theta_{h,N}(n) \leq q_N^*.$$

We now construct the two branching processes.  First, let $W_{m,N}^*(m) = 1$.  For $n \geq m$, if $W_{m,N}^*(n) = 0$ or if $W_{m,N}^*(n) \geq L_N$, then let $W_{m,N}^*(n + 1) = W_{m,N}^*(n)$.  Suppose instead $W_{m,N}^*(n) = j$, where $j \in \{1, \dots, L_N-1\}$.  Then let
\begin{equation}\label{Wkstardef}
W_{m,N}^*(n + 1) = \sum_{h=1}^j \sum_{i =1}^{N - c_N} \1_{\{U_{h,i, n,N} \leq q_{N}^*\}}.
\end{equation}
We interpret the sum $\sum_{i = 1}^{N - c_N} \1_{\{U_{h,i, n,N} \leq q_{N}^*\}}$ as the number of offspring in generation $n + 1$ of the $h$th individual in generation $n$.  Because the random variables $U_{h,i,n,N}$ are independent, $(W_{m,N}^*(n))_{n = m}^{\infty}$ is a branching process whose offspring distribution is binomial$(N - c_N, q_{N}^*)$, stopped when it reaches or exceeds $L_N$.  

The construction of $(W_{m,N}'(n))_{n = m}^{\infty}$ is more involved and requires external randomness.
We define additional independent uniformly distributed random variables $U_{h,i,n,N}$ for $h \in \{1, \dots, N\}$, $i \in \{N - c_N + 1, N - c_N + 2, \dots\}$, and positive integers $n$.
Let $W_{m,N}'(m) = 1$.  For $n \geq m$, if $W_{m,N}'(n) = 0$ or if $W_{m,N}'(n) \geq L_N$, then let $W_{m,N}'(n + 1) = W_{m,N}'(n)$.  Suppose instead that $W_{m,N}'(n) = j$, where $j \in \{1, \dots, L_N-1\}$. 
For $h \in \{1, \dots, j\}$, let
\begin{equation}\label{Sprimedef}
{\cal S}_{h, n,N}' = \{i \in \{1, \dots, N - c_N\}: Q_{i,n,N} \geq h\}
\end{equation}
be the set of individuals in generation $n+1$ who did not choose an individual marked $1, \dots, h-1$ as their parent.  Let ${\cal S}_{h, n,N}$ be the first $N - c_N - L_N$ elements of ${\cal S}_{h, n,N}' \cup \{N - c_N + 1, N - c_N + 2, \dots\}$.  Now let
\begin{equation}\label{Wkprimedef}
W_{m,N}'(n + 1) = \sum_{h=1}^j \sum_{i \in {\cal S}_{h, n,N}} \1_{\{U_{h,i, n,N} \leq q_{N}'\}}.
\end{equation}
Because the set ${\cal S}_{h,n,N}$ always has cardinality $N - c_N - L_N$ and the random variables $U_{h,i,n,N}$ are independent, the process $(W_{m,N}'(n))_{n = m}^{\infty}$ is a branching process whose offspring distribution is binomial$(N - c_N - L_N, q_{N}'$).
  
Define
\begin{align*}
\Psi^*_{m,N} &= \big\{W_{m,N}^*(n) \geq L_N \mbox{ for some }n \in \{m, m+1, \dots, m+\lfloor r_N/2 \rfloor\} \big\} \\
\Psi'_{m,N} &= \big\{W_{m,N}'(n) \geq L_N \mbox{ for some }n \in \{m, m+1, \dots, m+ \lfloor r_N/2 \rfloor\} \big\},
\end{align*}
which are the events that the population sizes in the branching processes reach $L_N$ within $r_N/2$ generations.  We require the branching processes to reach $L_N$ within $r_N/2$ generations so that $\Psi^*_{m,N}$ and $\Psi^*_{n,N}$ will be independent if $|m - n| > r_N$.

We will need the following elementary lemma about Galton-Watson processes.

\begin{Lemma}\label{GWlemma}
Let $(Z_n)_{n=0}^{\infty}$ be a Galton-Watson process with offspring distribution $(p_k)_{k=0}^{\infty}$ such that $Z_0 = 1$.  Suppose $p_1 \neq 1$.  Let $L$ be a positive integer.  Then
\begin{equation}\label{bpeq}
0 \leq P(Z_n \geq L \textup{ for some }n) - P(Z_n > 0 \textup{ for all }n) \leq 1/L.
\end{equation}
\end{Lemma}

\begin{proof}
Because $p_1 \neq 1$, we have $\lim_{n \rightarrow \infty} Z_n = \infty$ almost surely on the event that the process survives forever, which gives the first inequality in (\ref{bpeq}).  To prove the second inequality, it suffices to prove that $P(\{Z_n \geq L \mbox{ for some }n \} \cap \{Z_n = 0 \mbox{ for some }n\}) \leq 1/L$.  If the extinction probability is zero, then this bound is trivial.  Otherwise, let $(Z_n^*)_{n=0}^{\infty}$ have the same law as $(Z_n)_{n=0}^{\infty}$ conditioned on extinction.  It suffices to show that
\begin{equation}\label{Zstar}
P(Z_n^* \geq L \mbox{ for some }n) \leq 1/L.
\end{equation}
It is well-known (see, for example, section I.12 of \cite{athney}) that $(Z_n^*)_{n=0}^{\infty}$ is a subcritical branching process.  Therefore, $(Z_n^*)_{n=0}^{\infty}$ is a nonnegative supermartingale, so the result (\ref{Zstar}) follows from Doob's Maximal Inequality.
\end{proof}

\begin{Lemma}\label{compareGW}
For all nonnegative integers $m$, we have
$$\lim_{N \rightarrow \infty} P(\Psi^*_{m,N}) =\lim_{N \rightarrow \infty} P(\Psi'_{m,N}) = \omega(s).$$
\end{Lemma}

\begin{proof}
 Let $p_N^*$ be the probability that a branching process started with one individual whose offspring distribution is binomial($N - c_N, q_{N}^*)$ will survive forever.  Let $p_N'$ be the survival probability if the offspring distribution is binomial$(N - c_N - L_N, q_{N}')$.   Note that $$\lim_{N \rightarrow \infty} (N - c_N) q_{N}^* = \lim_{N \rightarrow \infty} (N - c_N - L_N)q_{N}' = 1 + s,$$ so as $N \rightarrow \infty$, the binomial$(N - c_N , q_{N}^*)$ distribution and the binomial$(N - c_N - L_N, q_{N}')$ distribution both converge to the Poisson$(1 + s)$ distribution.  It follows that
\begin{equation}\label{PA1}
\lim_{N \rightarrow \infty} p_{N}^* = \lim_{N \rightarrow \infty} p_{N}' = \omega(s) .
\end{equation}

Note that $p_N^*$ and $p_N'$ differ from $P(\Psi^*_{m,N})$ and $P(\Psi_{m,N}')$ respectively for two reasons.  First, it is possible for the branching process to reach $L_N$ or higher but still eventually go extinct.  By Lemma \ref{GWlemma}, the probability that this occurs is at most $1/L_N$, which tends to zero.  Second, it is possible for the branching process to reach $L_N$ or higher for the first time after more than $r_N/2$ generations.  However, using that $q_N' \geq 1/N$ for sufficiently large $N$, for all $n \geq m$ we have
$$P\big( W_{m,N}'(n+1) \geq L_N \big| 1 \leq W_{m,N}'(n) \leq L_N - 1 \big) \geq \binom{N - c_N - L_N}{L_N} \bigg( \frac{1}{N} \bigg)^{L_N} \bigg(1 - \frac{1}{N} \bigg)^{N - c_N - 2L_N}.$$  The reasoning following (\ref{binLN}) implies this probability is at least $1/(2L_N!)$ for sufficiently large $N$, and the same bound holds if $W_{m,N}'$ is replaced by $W_{m,N}^*$.  Because $L_N! \ll r_N$ by (\ref{rNdef}) and (\ref{LNdef}), the probability that the branching process reaches $L_N$ or higher for the first time after $r_N/2$ generations tends to zero as $N \rightarrow \infty$.  The result follows.
\end{proof}

\begin{Lemma}\label{bplemma}
Let $m$ be a positive integer such that $R_N(n) = 1$ for all $n \in \{m, m+1, \dots, m+r_N\}$.
Let $k$ be a positive integer, and recall the definitions of $\Phi_{k,m,N}$ and $\tau_{k,m,N}$ from \eqref{PhikmNdef} and \eqref{taudef}.  As $N \rightarrow \infty$,
\begin{equation}\label{bptriangle}
P\big(\{A_{k,N}(\tau_{k,m,N}) \geq L_N\} \triangle \Psi^*_{m,N} \big| {\cal F}_{m,N} \big) \1_{\Phi_{k,m,N}} \rightarrow 0,
\end{equation}
where we use $\triangle$ to denote the symmetric difference between two events.
\end{Lemma}

\begin{proof}
We need to explain how the two branching processes constructed above are coupled with the population process.  Let
$${\bar B}_{k,m,N} = \bigcap_{n=m}^{\tau_{k,m,N} - 1} \Big(B_{k,n,N} \cap \{\xi_{i,n,N} = 0 \mbox{ for all }i\} \Big).$$
Let $\gamma_{k,m,N} = \min\{n \geq m: W_{m,N}^*(n) \geq L_N\}$. 
We claim that if the event ${\bar B}_{k,m,N} \cap \Phi_{k,m,N}$ occurs, then the following hold:
\begin{enumerate}
\item We have $A_{k,N}(n) \leq W_{m,N}^*(n)$ for all $n \in \{m, m+1, \dots, \gamma_{k,m,N} \wedge (m + \lfloor r_N/2 \rfloor)\}$.

\item We have $W_{m,N}'(n) \wedge L_N \leq A_{k,N}(n)$ for all $n \in \{m, m+1, \dots, \tau_{k,m,N} \wedge (m + \lfloor r_N/2 \rfloor) \}$.
\end{enumerate}

We proceed by induction.  On $\Phi_{k,m,N}$, we have $W_{m,N}'(m) = A_{k,N}(m) = W_{m,N}^*(m) = 1$, so the two claims hold when $n = m$.  Suppose that for some $n \geq m$, we have $A_{k,N}(n) = j$ and $W_{m,N}^*(n) = j^*$, where $j \leq j^* \leq L_N - 1$.  For $i \in \{1, \dots, N - c_N\}$, we have $Q_{i,n,N} \leq j$ if and only if $U_{h,i,n,N} \leq \Theta_{h,N}(n)$ for some $h \in \{1, \dots, j\}$.  If $1 \leq h \leq j$ and $U_{h,i,n,N} \leq \Theta_{h,N}(n)$, then $U_{h,i,n,N} \leq q_N^*$ on $B_{k,n,N}$ by (\ref{Thetaupper}).  Therefore, corresponding to every individual in generation $n+1$ whose parent has a mark in $\{1, \dots, j\}$, there is an indicator random variable on the right-hand side of (\ref{Wkstardef}) that equals $1$.  It follows that as long as ${\bar B}_{k,m,N}$ occurs, we have $W_{m,N}^*(n+1) \geq A_{k,N}(n+1)$.  Note that when ${\bar B}_{k,m,N}$ occurs, there are no dormant type $k$ individuals, so all type $k$ individuals in generation $n+1$ have type $k$ parents in generation $n$.

Now suppose for some $n \geq m$, we have $W_{m,N}'(n) = j'$ and $A_{k,N}(n) = j$, where $j' \leq j \leq L_N - 1$. 
From \eqref{Wkprimedef}, we have $$W_{m,N}'(n+1) = \sum_{h=1}^{j'} \sum_{i \in \mathcal{S}_{h,n,N}} \1_{\{U_{h,i,n,N} \leq q_N'\}}.$$
Consider an indicator random variable $\1_{\{U_{h,i,n,N} \leq q_{N}'\}}$ on the right-hand side that equals $1$.  By (\ref{Thetalower}), if $B_{k,n,N}$ occurs and $i \leq N - c_N$, then $U_{h,i,n,N} \leq \Theta_{h,N}(n)$ and $Q_{i,n,N} = h$.  That is, the individual labelled $i$ in generation $n+1$ has a parent whose mark is $h \in \{1, \dots, j'\}$.  Note that if there exists $i \in {\cal S}_{h,n,N}$ such that $i > N - c_N$, then there at least $L_N$ individuals in generation $n+1$ who have parents with a mark in $\{1, \dots, j'\}$, and we must have $A_{k,N}(n+1) \geq L_N$.  It follows that as long as ${\bar B}_{k,m,N}$ occurs, we must have $W_{m,N}'(n+1) \wedge L_N \leq A_{k,N}(n+1)$.  The two claims now follow by induction.

The two claims imply that on $\Phi_{k,m,N}$, we have
\begin{align}\label{Atriangle}
\{A_{k,N}(\tau_{k,m,N}) \geq L_N\} \triangle \Psi^*_{m,N} \subseteq \big( \Psi'_{m,N} \triangle \Psi^*_{m,N} \big) \cup {\bar B}_{k,m,N}^c \cup \{\tau_{k,m,N} > m + r_N/2\}.
\end{align}
The probability of the third event on the right-hand side of (\ref{Atriangle}) tends to zero by (\ref{taukL}).
The event ${\bar B}_{k,n,N}$ can only fail to occur if, before time $\tau_{k,m,N}$, either some individual acquires a mutation, some type $k$ individual moves to the dormant population, or the number of individuals of type $k-1$ or higher drops below $(1 - \eps_N)N$.  The probability that some individual acquires a mutation before time $m + r_N/2$ is bounded above by $N \mu_N r_N$.  The probability that a type $k$ individual moves to the dormant population before time $\tau_{k,m,N}$ is bounded by (\ref{noseed}).  Lemma \ref{kstayhigh} shows that the probability that the number of individuals of type $k-1$ or higher drops below $(1 - \eps_N)N$ before time $m + r_N/2$ is bounded above by $r_N e^{\gamma N \eps_N^2}$.  Note that
\begin{equation}\label{rNgamzero}
\lim_{N \rightarrow \infty} r_N e^{\gamma N \eps_N^2} = 0
\end{equation}
because $r_N \ll U_N \lesssim \rho_N^{-1} \ll N^{b-1}$ by \eqref{Reg1Asm2}, \eqref{Reg2Asm}, \eqref{totalgen}, and \eqref{rNdef}, and $N \eps_N^2 \gg \log N$ by \eqref{LNdef} and \eqref{epsNdef}.
Combining these observations, we get that as $N \rightarrow \infty$,
\begin{equation}\label{At1}
P \big( {\bar B}_{k,m,N}^c \cup \{\tau_{k,m,N} > m + r_N/2\} \big| {\cal F}_{m,N} \big) \1_{\Phi_{k,m,N}} \rightarrow 0.
\end{equation}
If the event ${\bar B}_{k,m,N}$ occurs, then the two claims introduced above imply that $\Psi^*_{m,N}$ occurs whenever $\Psi'_{m,N}$ occurs.
Also, because the branching processes are independent of ${\cal F}_{m,N}$, it follows from Lemma \ref{compareGW} that $|P(\Psi^*_{m,N}|{\cal F}_{m,N}) - P(\Psi'_{m,N}|{\cal F}_{m,N})| \rightarrow 0$.  Therefore, as $N \rightarrow \infty$,
\begin{equation}\label{At2}
P \big( \Psi'_{m,N} \triangle \Psi^*_{m,N} \big| {\cal F}_{m,N} \big) \1_{\Phi_{k,m,N}} \rightarrow 0.
\end{equation}
The result follows from (\ref{Atriangle}), (\ref{At1}), and (\ref{At2}).
\end{proof}

\subsubsection{Establishment of the new type}

We have seen that when a new beneficial mutation appears, the number of individuals with the new type quickly either drops to zero or reaches $L_N$, and that the probability that the number reaches $L_N$ is approximately $\omega(s_N)$.  It remains to show that once there are $L_N$ individuals of the new type, the new type quickly becomes dominant.  Again, recall that throughout this section, we are assuming the beneficial mutation occurs during the summer.

\begin{Lemma}\label{Aincrease}
There exist constants $\kappa > 0$ and $\beta > 0$ such that for all positive integers $k$, all positive integers $m$ for which $R_N(m) = 1$, and all positive integers $j$ for which $j \leq (1 - \eps_N)N$, we have
$$P\big(A_{k,N}(m+1) \leq (1 + \kappa \eps_N) j \,\big|{\cal F}_{m,N} \big) \leq 2e^{-\beta \eps_N^2 j}$$
for sufficiently large $N$ on the event
\begin{equation}\label{Aincevent}
\{A_{k,N}(m) = j\} \cap \{A_{\ell,N}(m) = 0 \mbox{ for all }\ell > k\}.
\end{equation}
\end{Lemma}

\begin{proof}
The proof is similar to the proof of Lemma \ref{kstayhigh}.  On the event in \eqref{Aincevent}, the active population in generation $m$ consists of $j$ individuals with fitness $(1 + s_N)^k$ and $N - j$ individuals with fitness at most $(1 + s_N)^{k-1}$.  Therefore, these $N - c_N$ individuals in generation $m+1$ will each be type $j$ with probability at least
$$\frac{j(1 + s_N)^k}{j(1 + s_N)^k + (N-j)(1 + s_N)^{k-1}} \cdot (1 - 2 \mu_N) = \frac{j(1+ s_N)(1 - 2 \mu_N)}{N + js_N}.$$
Therefore, on the event in \eqref{Aincevent}, $A_{k,N}(m+1)$ stochastically dominates a random variable $Z_N$ having a binomial$(N - c_N, q_N)$ distribution, where $q_N = j(1 + s_N)(1 - 2 \mu_N)/(N + js_N)$.
As long as $j \leq (1 - \eps_N)N$, we have
$$E[Z_N] = \frac{(N - c_N) j (1 + s_N)(1 - 2 \mu_N)}{N + j s_N} \geq \frac{(N - c_N) j (1 + s_N)(1 - 2 \mu_N)}{N + (1 - \eps_N) N s_N}.$$
Let $s_N^* = s_N/(1 + s_N)$.  It then follows from (\ref{sstareq}) that
$$E[Z_N] \geq \Big(1 - \frac{c_N}{N} \Big) (1 - 2 \mu_N) (1 + \eps_N s_N^*) j.$$
As noted in the proof of Lemma \ref{kstayhigh}, we have $\eps_N s_N \gg \mu_N$ and $\eps_N s_N \gg c_N/N$.
Therefore, recalling (\ref{slarge}), we can choose $\kappa > 0$ such that
$E[Z_N] \geq (1 + 2 \kappa \eps_N) j$ for sufficiently large $N$.
Using (\ref{chernoff2}), on the event in \eqref{Aincevent}, we have for sufficiently large $N$,
\begin{align*}
P\big(A_{k,N}(m+1) \leq (1 + \kappa \eps_N) j \,\big|{\cal F}_{m,N} \big) &\leq P(Z_N \leq (1 + \kappa \eps_N) j) \\
&\leq P \bigg(Z_N \leq E[Z_N] \cdot \frac{1 + \kappa \eps_N}{1 + 2 \kappa \eps_N} \bigg) \\
&\leq P \bigg(\big|Z_N - E[Z_N]\big| \geq \Big(1 - \frac{1 + \kappa \eps_N}{1 + 2 \kappa \eps_N} \Big) E[Z_N] \bigg) \\
&= P \bigg(\big|Z_N - E[Z_N]\big| \geq \Big(\frac{\kappa \eps_N}{1 + 2 \kappa \eps_N} \Big) E[Z_N] \bigg) \\
&\leq 2 \exp \bigg( - \frac{1}{3} \Big(\frac{\kappa \eps_N}{1 + 2 \kappa \eps_N} \Big)^2 E[Z_N] \bigg) \\
&\leq 2 \exp \bigg( - \frac{\kappa^2 \eps_N^2 j}{3(1 +2 \kappa \eps_N)} \bigg).
\end{align*}
Because $1 + 2 \kappa \eps_N \rightarrow 1$ as $N \rightarrow \infty$, the result follows if we choose $0 < \beta < \kappa^2/3$.
\end{proof}

\begin{Lemma}\label{finishsweep}
Let $m$ be a positive integer such that $R_N(\ell) = 1$ for $\ell \in \{m, m+1, \dots, m+ \lfloor r_N/2 \rfloor\}$.   Let $$\tau^*_{k, m, N} = \min\{n \geq m: A_{k,N}(n) = 0 \mbox{ or }A_{k,N}(n) \geq (1 - \eps_N) N\}.$$  Let $\delta > 0$.  Then for sufficiently large $N$, on the event
\begin{equation}\label{2event}
\{A_{k,N}(m) \geq L_N\} \cap \{A_{\ell,N}(m) + D_{\ell,N}(m) = 0 \mbox{ for all }\ell > k\},
\end{equation}
we have $$P\big(\{A_{k,N}(\tau^*_{k, m, N}) \geq (1 - \eps_N)N \} \cap \{ \tau^*_{k, m, N} \leq m + r_N/2 \} \big|{\cal F}_{m,N} \big) > 1 - \delta.$$
\end{Lemma}

\begin{proof}
Define the constants $\kappa$ and $\beta$ as in Lemma \ref{Aincrease}.
For nonnegative integers $i$, define the events
\begin{align*}
B_i &= \{A_{k,N}(m + i) \geq L_N (1 + \kappa \eps_N)^i\} \cap \{A_{\ell,N}(m + i) + D_{\ell,N}(m + i) = 0 \mbox{ for all }\ell > k\}, \\
B_i^* &= B_i \cap \{A_{k,N}(m + i) < (1 - \eps_N) N\}.
\end{align*}
The probability of a mutation in any generation is bounded above by $2 N \mu_N$.  Also, by Lemma~\ref{Aincrease}, if there are $j$ individuals of type $k$ in the previous generation and $j \leq (1 - \eps_N)N$, the probability that the number of type $k$ individuals does not increase by at least a factor of $1 + \kappa \eps_N$ in the next generation is at most $2 e^{-\beta \eps_N^2 j}$ for sufficiently large $N$.  Therefore, $$P\big(B_{i+1}^c \big| B_0^* \cap \dots \cap B_i ^*\big) \leq 2N \mu_N + 2e^{-\beta \eps_N^2 L_N (1 + \kappa \eps_N)^i}.$$
Because $\eps_N \gg (\log N)/r_N$ by (\ref{epsNdef}), for sufficiently large $N$ we have $(1 + \kappa \eps_N)^{\lfloor r_N/2 \rfloor} \geq N$, and therefore $L_N (1 + \kappa \eps_N)^{\lfloor r_N/2 \rfloor} > (1 - \eps_N) N$.
Therefore, if the event in (\ref{2event}) occurs, then we will have $A_{k,N}(\tau^*_{k, m, N}) \geq (1 - \eps_N)N$ and $\tau^*_{k, m, N} \leq m + r_N/2$ unless $B_{i+1}^c \cap B_0^* \cap \dots \cap B_i^*$ occurs for some $i \in \{0, 1, \dots, \lfloor r_N/2 \rfloor - 1\}$.  The probability that this occurs is bounded above by
\begin{equation}\label{TminusTL}
\sum_{i=0}^{\lfloor r_N/2 \rfloor-1} \big(2N \mu_N + 2e^{-\beta \eps_N^2 L_N (1 + \kappa \eps_N)^i} \big) \leq N \mu_N r_N + 2 \sum_{i=0}^{\infty} e^{-\beta \eps_N^2 L_N (1 + \kappa \eps_N)^i}.
\end{equation}
The first term tends to zero as $N \rightarrow \infty$ by (\ref{rNdef}) and (\ref{UNmu}). To bound the second term, we consider the $i = 0$ term separately, and use the bound $(1 + \kappa \eps_N)^i \geq \kappa \eps_N i$ for $i \geq 1$ so that the sum can be bounded above by a geometric series.  Because $\eps_N^3 L_N \rightarrow \infty$ by (\ref{epsNdef}), the second term also tends to zero as $N \rightarrow \infty$.  The result of the lemma follows.
\end{proof}

We can now combine the results in this subsection into the following corollary, which describes how the active population evolves during class 1 intervals.  For a class 1 interval that occurs during the summer, the corollary says that if type $k$ dominates the population at the beginning of the interval, then either type $k+1$ dominates the population at the end of the interval, indicating that the new beneficial mutation spread in the population, or else type $k$ still dominates the population at the end of the interval, indicating that the new beneficial mutation died out.  In the latter case, the beneficial mutation dies out before spreading to $L_N$ members of the active population.

\begin{Cor}\label{class1cor}
Let $G_{1,N}^+$ be the event that for all $m \in {\cal M}_N^+ \cap {\cal T}_N^+$ such that $\Lambda_{k,m,N}$ occurs for some $k$, the following hold:
\begin{enumerate}
\item If $\Psi^*_{m+1,N}$ occurs, then $\Lambda_{k+1, m+r_N+1, N}$ occurs.

\item If $\Psi^*_{m+1,N}$ does not occur, then $\Lambda_{k, m + r_N+1}$ occurs, and for all $n \in \{m, m+1, \dots, m+r_N\}$, we have $A_{k+1,N}(n) < L_N$.
\end{enumerate}
Let $G_{1,N}^-$ be the event that for all $m \in {\cal M}_N^+ \cap {\cal T}_N^-$ such that $\Lambda_{k,m,N}$ occurs for some $k$, the following hold:
\begin{enumerate}
\item If $\Psi^*_{m+1,N}$ occurs, then $\Lambda_{k-1, m+r_N+1, N}$ occurs.

\item If $\Psi^*_{m+1,N}$ does not occur, then $\Lambda_{k, m + r_N+1}$ occurs, and for all $n \in \{m, m+1, \dots, m+r_N\}$, we have $A_{k-1,N}(n) < L_N$.
\end{enumerate}
Let $G_{1,N} = G_{1,N}^+ \cap G_{1,N}^-$.  Then $$\lim_{N \rightarrow \infty} P(G_{1,N}) = 1.$$
\end{Cor}

\begin{proof}
The expected number of elements of ${\cal M}_N^+$ is at most $N \mu_N (1 + \rho_N^{-1} t_0)$.  By (\ref{rhoNdef}) and (\ref{hlower}), there exists a positive constant $C$ such that the expected number of elements of ${\cal M}_N^+$ is bounded above by $C$ for all $N$.  Also, by Lemma \ref{mutlem3}, with probability tending to one as $N \rightarrow \infty$, no generation $m \in {\cal M}_N^+$ is within $r_N$ generations of a seasonal change, and by Lemma \ref{mutlem2}, with probability tending to one as $N \rightarrow \infty$, no further mutations will occur until after generation $m+r_N+1$.

Suppose $m \in {\cal M}_N^+ \cap {\cal T}_N^+$ and $\Lambda_{k,m,N}$ occurs.  By Lemma \ref{mutlem1}, the probability that there is more than one mutation in a generation tends to zero, and by Lemma \ref{kstayhigh}, with probability tending to one, the number of individuals of type $k$ or higher will remain above $(1 - \eps_N)N$ in generation $m+1$.  Therefore, with probability tending to one as $N \rightarrow \infty$, for all $m \in {\cal M}_N^+ \cap {\cal T}_N^+$ such that $\Lambda_{k,m,N}$ occurs, the event $\Phi_{k+1,m+1,N}$ occurs.  By the same argument, with probability tending to one as $N \rightarrow \infty$, for all $m \in {\cal M}_N^+ \cap {\cal T}_N^-$ such that $\Lambda_{k,m,N}$ occurs, the event $\Phi_{k-1,m+1,N}$ occurs.

Now suppose $R_N(n) = 1$ for all $n \in \{m,m+1, \dots, m+r_N\}$.  It follows from Lemma \ref{1toL} that on the event $\Phi_{k+1, m+1, N}$, with probability tending to one,
the number of type $k+1$ individuals in the active population will reach $0$ or $L_N$ before generation $m+1+r_N/2$, and no type $k+1$ individuals will enter the dormant population before this time.  By Lemma \ref{bplemma}, outside an event whose probability tends to zero, we have $A_{k+1,N}(\tau_{k+1,m+1,N}) \geq L_N$ on $\Psi_{m+1,N}^*$ and $A_{k+1,N}(\tau_{k+1,m+1,N}) = 0$ on $\Psi_{m+1,N}^c$.  We consider these two cases separately.

In the former case, Lemma \ref{finishsweep} implies that the number of type $k+1$ individuals in the active population will reach $(1 - \eps_N)N$ by generation $m+1+r_N$ with probability tending to one as $N \rightarrow \infty$.  Then Lemma \ref{kstayhigh} implies that the number of type $k+1$ individuals stays above $(1 - \eps_N)N$ until generation $m + 1 + r_N$ with probability tending to one because $r_N e^{-\gamma N \eps_N^2} \rightarrow 0$ as $N \rightarrow \infty$ by \eqref{rNgamzero}.  Also, Lemma~\ref{mutlem2} implies that with probability tending to one as $N \rightarrow \infty$, no additional mutation will produce an individual of type $k+2$.  These events are enough to ensure that $\Lambda_{k+1, m+r_N+1, N}$ occurs.

In the latter case, Lemma \ref{kstayhigh} implies that with probability tending to one as $N \rightarrow \infty$, the number of individuals with type $k$ and higher will stay above $(1 - \eps_N)N$ until generation $m + r_N + 1$, and Lemma \ref{mutlem2} ensures that with probability tending to one as $N \rightarrow \infty$, no more type $k+1$ individuals will appear in the population.  These two events are enough to guarantee that $\Lambda_{k,m+r_N+1,N}$ occurs, and that $A_{k+1,N}(n) < L_N$ for all $n \in \{m+1, \dots, m+r_N\}$.

These observations imply that $$\lim_{N \rightarrow \infty} P(G_{1,N}^+) = 1$$  The same result holds for $G_{1,N}^-$, which implies the result of the corollary.
\end{proof}

\subsection{Class 2 intervals: the beginning of a season}\label{2sub}

Suppose type $k$ individuals become dominant in the active population during the summer.  We now consider the process by which type $k$ individuals reestablish themselves as dominant in the active population, after the season switches from summer to winter and back to summer.  Type $k$ individuals need to rejoin the active population from the seedbank, and then spread rapidly again through the active population. 

Our goal in this section is to prove two lemmas.  Lemma \ref{ktoseed} shows that if type $k$ individuals are dominant in the active population $r_N$ time units before the end of the summer, then there will be a sufficient supply of type $k$ individuals in the dormant population at the beginning of the following summer.  Lemma \ref{kfromseed} shows that if there is a sufficient supply of type $k$ individuals in the dormant population at the beginning of the following summer, then type $k$ individuals quickly become dominant again in the active population.  Because, in Regime 2, there may be many changes of season before time $\rho_N^{-1} t_0$, we need rather tight bounds on the rare events that the type $k$ population does not behave as expected during a seasonal change.  In particular, we need to show that the probability of these rare events tends to zero faster than $N^{-a}$ for any $a > 0$.  

\begin{Lemma}\label{ktoseed}
Let $m = jU_N + V_N - r_N$ for some nonnegative integer $j$.  Let $F_{m,N}$ be the event that no individual acquires a mutation in generations $m+1, \dots, m+r_N$.
There exists a positive number $\eta > 0$ and a sequence $(\delta_N)_{N=1}^{\infty}$ with $N^a \delta_N \rightarrow 0$ for all $a > 0$ such that for all nonnegative integers $k$ and all $N$, we have
\begin{equation}\label{Dbig}
P \big( F_{m,N} \cap \{D_{k, N}((j+1)U_N) < \eta r_N c_N\} \big| {\cal F}_{m,N} \big) \leq \delta_N \quad \mbox{on } \Lambda_{k,m,N}.
\end{equation}
\end{Lemma}

\begin{proof}
For $i \in \{1, \dots, r_N\}$, let $W_i$ be the number of individuals of type $k$ that enter the seed bank in generation $m + i$, and let $B_i$ be the event that $A_{k,N}(m+i-1) \geq (1 - \eps_N)N$.  On the event $B_i$, we have $E[W_i|{\cal F}_{m+i-1,N}] \geq (1 - \eps_N)c_N$.  It follows that on $B_i$, we have $E[c_N - W_i|{\cal F}_{m+i-1,N}] \leq \eps_N c_N$, and therefore by the conditional Markov's Inequality, $P(c_N - W_i > 2 \eps_N c_N|{\cal F}_{m+i-1,N}) \leq 1/2$.  That is, we have
\begin{equation}\label{Wimain}
P(W_i \geq (1 - 2 \eps_N) c_N|{\cal F}_{m+i-1,N}) \geq 1/2 \quad \mbox{on } B_i.
\end{equation}
Let $W = \sum_{i=1}^{r_N} \1_{\{W_i \geq (1 - 2 \eps_N)c_N\}}$ be the number of generations in which at least $(1 - 2 \eps)c_N$ individuals of type $k$ enter the seed bank.  By (\ref{Wimain}), if $Z$ has a binomial$(r_N, 1/2)$ distribution, then $$P\bigg(F_{m,N} \cap \Big\{ W \leq \frac{r_N}{4} \Big\} \Big| {\cal F}_{m,N} \bigg) \leq P\bigg(Z \leq \frac{r_N}{4} \bigg) + P \bigg( F_{m,N} \cap \bigcup_{i=1}^{r_N} B_i^c  \,\Big| {\cal F}_{m,N} \bigg).$$
It follows from Lemma \ref{kstayhigh} that
$$P \bigg( F_{m,N} \cap \bigcup_{i=1}^{r_N} B_i^c \, \Big| {\cal F}_{m,N} \bigg) \leq 2 (r_N - 1) e^{-\gamma N \eps_N^2} \quad \mbox{on } \Lambda_{k,m,N}.$$
Furthermore, it follows from (\ref{chernoff1}) that $P(Z \leq r_N/4) \leq 2e^{-r_N/8}.$
Also, as long as $W > r_N/4$, we have $\sum_{i=1}^{r_N} W_i > (1 - 2 \eps_N)r_Nc_N/4$.  Therefore,
\begin{equation}\label{enterseed}
P \bigg( F_{m,N} \cap \bigg\{ \sum_{i=1}^{r_N} W_i \leq \frac{1 - 2 \eps_N}{4} r_N c_N \bigg\} \Big| {\cal F}_{m,N} \bigg) \leq 2(r_N - 1)e^{-\gamma N \eps_N^2} + 2e^{-r_N/8} \quad \mbox{on } \Lambda_{k,m,N}.
\end{equation}
Note that (\ref{enterseed}) shows that with high probability, many individuals of type $k$ or higher enter the seed bank before the end of the summer.  It remains to show that enough of these individuals stay in the seed bank through the winter.

Let $h_N = \lceil (1 - 2 \eps_N) r_N c_N/4 \rceil$.  Consider the first $h_N$ individuals of type $k$ to enter the seed bank between generations $m+1$ and $m+r_N$, breaking ties arbitrarily if necessary when several such individuals enter the seed bank in one generation.  If fewer than $h_N$ individuals of type $k$ enter the seed bank during this period, then arbitrarily label additional individuals that enter the seed bank during this period to ensure that $h_N$ individuals are labeled.  Let $H_N$ be the number of these $h_N$ individuals that stay in the seed bank until generation $(j+1)U_N$.  Note that $H_N$ will be stochastically smallest if all $h_N$ individuals enter the seed bank in generation $m+1$.  Therefore, we can define a random variable $H_N^*$ to be the number of individuals, out of $h_N$ individuals in the seed bank in generation $m+1$, that stay in the seed bank until generation $(j+1)U_N$, and then $H_N$ stochastically dominates $H_N^*$.  Note that these individuals must stay in the seedbank for $(j+1)U_N - (m+1) = U_N - V_N + r_N - 1$ generations.  Because the probability that an individual exits the seed bank in a given generation is $c_N/K_N$, the probability that an individual stays in the seed bank for $U_N - V_N + r_N - 1$ generations is $p_N = (1 - c_N/K_N)^{U_N - V_N + r_N - 1}$.  It follows that $E[H_N^*] = h_N p_N$.  Because $c_N (U_N - V_N + r_N - 1)/K_N \lesssim 1$ by (\ref{betadef}), (\ref{Reg1Asm1}), (\ref{Reg1Asm2}), (\ref{Reg2Asm}), and (\ref{rNdef}), the sequence $(p_N)_{N=1}^{\infty}$ is bounded away from zero.  

Now label $h_N$ individuals in the seed bank in generation $m+1$ by the integers $1, \dots, h_N$.
Let $\Lambda_i$ be the event that the individual labelled $i$ stays in the seed bank until generation $(j+1)U_N$.  We can write $H_N^* = \sum_{i=1}^{h_N} \1_{\Lambda_i}$, which has the same distribution as the random variable $Y_T$ in Lemma \ref{geomlem}, if we take $c = c_N$, $J = h_N$, $K = K_N$, and $T = \{U_N - V_N + r_N, U_N - V_N + r_N + 1, \dots\}$.  Therefore,
$$P \bigg( H_N \leq \frac{h_N p_N}{2} \Big| {\cal F}_{m,N} \bigg) \leq P \bigg( H_N^* \leq \frac{h_N p_N}{2} \bigg) \leq P \bigg(\bigg|\sum_{i=1}^{h_N} \1_{\Lambda_i} - E \bigg[ \sum_{i=1}^{h_N} \1_{\Lambda_i} \bigg] \bigg| \geq \frac{h_N p_N}{2} \bigg).$$
It follows from Lemma \ref{geomlem} with $\eps = 1/2$ that
$$P \bigg( H_N \leq \frac{h_N p_N}{2} \Big| {\cal F}_{m,N} \bigg) \leq 2 e^{-h_N p_N/12}.$$
Let $$\eta = \frac{1}{9} \min_N p_N > 0,$$
so that $\eta r_N c_N \leq h_N p_N/2$ for sufficiently large $N$.  Note that there will be at least $\eta r_N c_N$ individuals of type $k$ in the seed bank in generation $(j+1)U_N$ unless either fewer than $h_N$ individuals of type $k$ enter the seed bank between generations $m+1$ and $m+r_N$, or among the first $h_N$ such individuals to enter the seed bank beginning in generation $m+1$, fewer than $h_Np_N/2$ of these individuals stay in the seed bank until generation $(j+1)U_N$.  It follows that on the event $\Lambda_{k,m,N}$, for sufficiently large $N$ we have, using \eqref{enterseed},
\begin{align*}
&P \big( F_{m,N} \cap \{D_{k, N}((j+1)U_N) < \eta r_N c_N\} \big| {\cal F}_{m,N} \big) \\
&\qquad \qquad \leq P \bigg( F_{m,N} \cap \bigg\{ \sum_{i=1}^{r_N} W_i \leq \frac{1 - 2 \eps_N}{4} r_N c_N \bigg\} \Big| {\cal F}_{m,N} \bigg) + P \bigg( H_N \leq \frac{h_N p_N}{2} \Big| {\cal F}_{m,N} \bigg) \\
&\qquad \qquad \leq 2(r_N - 1)e^{-\gamma N \eps_N^2} + 2e^{-r_N/8} + 2 e^{-h_N p_N/12}.
\end{align*}
The result follows because, recalling that $r_N \gg \log N$ by (\ref{rNdef}) and using \eqref{rNgamzero}, we see that this expression tends to zero faster than any power of $N$.
\end{proof}

\begin{Lemma}\label{kfromseed}
Let $m = (j+1)U_N$ for some nonnegative integer $j$.
Let $k$ be a positive integer.  Let $F_{m,N}$ be the event that no individual acquires a mutation in generations $m+1, \dots, m+r_N$.  Suppose $\eta > 0$.  Then there exists a sequence $(\delta_N^*)_{N=1}^{\infty}$ such that $N^a \delta_N^* \rightarrow 0$ for all $a > 0$ and such that on the event
\begin{equation}\label{largeD}
\{D_{k,N}(m) \geq \eta r_N c_N\} \cap \{A_{\ell,N}(m) + D_{\ell,N}(m) = 0 \mbox{ for all }\ell > k\},
\end{equation}
we have for all $N$,
\begin{equation}\label{largeA}
P(F_{m,N} \cap \{A_{k,N}(m+r_N) < (1 - \eps_N)N\}|{\cal F}_{m,N}) \leq \delta_N^*.
\end{equation}
\end{Lemma}

\begin{proof}
We work on the event in (\ref{largeD}).  Let $\ell_N = \lceil \eta r_N c_N \rceil$.  On the event in (\ref{largeD}), there are at least $\ell_N$ individuals of type $k$ in the seed bank in generation $m$.  Label $\ell_N$ such individuals $1, \dots, \ell_N$.  For $i \in \{m+1, m+2, \dots, m+ \lceil r_N/2 \rceil\}$, let $H_i^*$ be the number of these individuals that exit the seed bank in generation $i$, and let $H_i$ be the total number of type $k$ individuals that exit the seed bank in generation $i$.  Let $$H^* = H^*_{m+1} + H^*_{m+2} + \dots + H^*_{m + \lceil r_N/2 \rceil}$$ be the number of the $\ell_N$ labeled individuals that exit the seed bank before generation $m + \lceil r_N/2 \rceil$.
Let $B_i$ be the event that the $i$th such individual exits the seed bank before generation $m + \lceil r_N/2 \rceil$, so that $H^* = \sum_{i=1}^{\ell_N} \1_{B_i}$.  Note that $H^*$ has the same distribution as the random variable $Y_T$ in Lemma \ref{geomlem} if we take $c = c_N$, $K = K_N$, $J = \ell_N$, and $T = \{1, \dots, \lceil r_N/2 \rceil\}$.
Because $c_N$ out of $K_N$ individuals in the seed bank are selected for removal in each generation, we have
$$P(B_i) = 1 - \bigg( 1 - \frac{c_N}{K_N} \bigg)^{\lceil r_N/2 \rceil}$$ for all $i \in \{1, \dots, \ell_N\}$.  Therefore, using the elementary bound $(1-a)^n \geq an - \frac{1}{2}(an)^2$ for $0 < a < 1$ and positive integers $n$, we have
$$E[H^*] = \ell_N \bigg(1 - \bigg(1 - \frac{c_N}{K_N} \bigg)^{\lceil r_N/2 \rceil} \bigg) \geq \lceil \eta r_N c_N \rceil \bigg( \frac{c_N \lceil r_N/2 \rceil}{K_N} - \frac{1}{2} \bigg( \frac{c_N \lceil r_N/2 \rceil}{K_N} \bigg)^2 \bigg).$$
We have $c_N U_N/K_N \lesssim 1$ by (\ref{Reg1Asm1}), (\ref{Reg1Asm2}), and (\ref{Reg2Asm}), and therefore $c_N r_N/K_N \rightarrow 0$ by (\ref{rNdef}).  It follows that for sufficiently large $N$, we have
$$E[H^*] \geq \frac{\eta c_N^2 r_N^2}{3 K_N}.$$
It now follows from Lemma \ref{geomlem} with $\eps = 1/2$ that
\begin{equation}\label{Dbound}
P \bigg(H^* \leq \frac{\eta c_N^2 r_N^2}{6K_N} \bigg) \leq P\bigg(H^* \leq \frac{1}{2} E[H^*] \bigg) \leq e^{-E[H^*]/12} \leq e^{-(\eta/36)(c_N^2 r_N^2/K_N)}.
\end{equation}
This implies that with high probability, at least $\eta c_N^2 r_N^2/(6K_N)$ type $k$ individuals will exit the seed bank before generation $m + \lceil r_N/2 \rceil$.  It remains to show that the descendants of these individuals will spread rapidly in the active population.

To do this, we couple the number of type $k$ individuals in the active population with a branching process with immigration, which we will denote by $(W_{k,m,N}(n))_{n=m}^{m + \lceil 3r_N/4 \rceil}$.  The construction of the branching process with immigration will be similar to the construction of the process $(W'_{m,N}(n))_{n=m}^{\infty}$ in section \ref{bpsec}.  We first set $W_{k,m,N}(m) = A_{k,N}(m)$.  For $n$ such that $m \leq n < m + \lceil 3r_N/4 \rceil$, if $W_{k,m,N}(n) \geq \lceil (\log N)^2 \rceil$, then let $W_{k,m,N}(n+1) = W_{k,m,N}(n)$.  Suppose instead that $W_{k,m,N}(n) = j$, where $0 \leq j < \lceil (\log N)^2 \rceil$.  Following (\ref{Sprimedef}), for $h \in \{1, \dots, j\}$, let ${\cal S}_{h, n,N}' = \{i \in \{1, \dots, N - c_N\}: Q_{i,n,N} \geq h\}$ be the set of individuals in generation $n+1$ who did not choose an individual marked $1, \dots, h-1$ as their parent.  Let ${\cal S}_{h, n,N}$ be the first $N - c_N - \lceil (\log N)^2 \rceil$ members of the set ${\cal S}_{h, n,N}' \cup \{N - c_N + 1, \dots, N - c_N + \lceil (\log N)^2 \rceil\}$.  Now let $${\tilde q}_N = \frac{1 + s_N}{N + \lceil (\log N)^2 \rceil s_N},$$ which is similar to the definition of $q_N'$ in (\ref{qdef}) but with $\lceil (\log N)^2 \rceil$ in place of $L_N$, and let
\begin{equation}\label{Wdef}
W_{k,m,N}(n + 1) = H_{n+1} + \sum_{h=1}^j \sum_{i \in {\cal S}_{h, n,N}} \1_{\{U_{h,i, n,N} \leq {\tilde q}_N\}}.
\end{equation}
Because the set ${\cal S}_{h,n,N}$ always has cardinality $N - c_N - \lceil (\log N)^2 \rceil$ and the random variables $U_{h,i,n,N}$ are independent, the process $(W_{k,m,N}(n))_{n = m}^{\infty}$ is a branching process whose offspring distribution is binomial$(N - c_N - \lceil (\log N)^2 \rceil, {\tilde q}_N$) in which $H_n$ immigrants arrive in generation $n$.

Let $$\tau = (m +\lceil 3r_N/4 \rceil) \wedge \min\{n \geq m: A_{k,N}(n) > (\log N)^2\} \wedge \min\{n \geq m: \xi_{i,n,N} \neq 0 \mbox{ for some }i\}.$$
We now show by induction that on the event in (\ref{largeD}), we have
\begin{equation}\label{immcouple}
A_{k,N}(n) \geq W_{k,m,N}(n) \wedge \lceil (\log N)^2 \rceil \mbox{ for all }n \in \{m, m+1, \dots, \tau\}.
\end{equation}
We have defined $W_{k,m,N}(m) = A_{k,N}(m)$, so the result holds for $n = m$.  Suppose $A_{k,N}(n) \geq W_{k,m,N}(n)$ with $m \leq n < \tau$.  Because we are working on the event in (\ref{largeD}), there are no individuals of type greater than $k$ in the population at time $m$, and because no mutations occur before time $\tau$, this means the type $k$ individuals are the fittest in the population at time $n$.  We know there will be $H_{n+1}$ type $k$ individuals in generation $n+1$ emerging from the seed bank.
Now consider the $N - c_N$ individuals in generation $n+1$ whose parent comes from the active population in generation $n$.  Write $j = W_{k,m,N}(n)$.  Suppose $i \leq N - c_N$ and the indicator random variable $\1_{\{U_{h,i,n,N} \leq {\tilde q}_N\}}$ on the right-hand side of (\ref{Wdef}) equals $1$.  Because $n < \tau$, we have $A_{k,n}(n) < \lceil(\log N)^2 \rceil$, so by the calculation in \eqref{Thetalower} with $\lceil (\log N)^2 \rceil$ in place of $L_N$, we have $U_{h,i,n,N} \leq \Theta_{h,N}(n)$, which ensures that $Q_{i,n,N} = h$.  Thus, the individual labelled $i$ in generation $n+1$ has a parent with a mark in $\{1, \dots, j\}$, and therefore has type $k$.  If there is an $i \in {\cal S}_{h,n,N}$ with $i > N - c_N$, then we have $A_{k,N}(n+1) \geq \lceil (\log N)^2 \rceil$.  It follows that $A_{k,N}(n+1) \geq W_{k,m,N}(n+1) \wedge \lceil (\log N)^2 \rceil$.  The result (\ref{immcouple}) follows by induction.

It remains to deduce (\ref{largeA}) from (\ref{immcouple}).  Let $B_{1,N}$ be the event that $W_{k,m,N}(n) < (\log N)^2$ for all $n$ such that $m \leq n \leq m + \lceil 3r_N/4 \rceil$.  Let $B_{2,N}$ be the event that, for some $n$ such that $m \leq n \leq m + \lceil 3r_N/4 \rceil$, we have $A_{k,N}(n) \geq (\log N)^2$, but that $A_{k,N}(m + r_N) < (1 - \eps_N)N$.  A consequence of (\ref{immcouple}) is that if $F_{m,N}$ occurs and the event in (\ref{largeD}) occurs, then
$A_{k,N}(m + r_N) \geq (1 - \eps_N)N$ unless $B_{1,N} \cup B_{2,N}$ occurs.  Therefore, on the event in (\ref{largeD}), $$P(F_{m,N} \cap \{A_{k,N}(m + r_N) < (1 - \eps_N)N\}|{\cal F}_{m,N}) \leq P(F_{m,N} \cap B_{1,N}|{\cal F}_{m,N}) + P(F_{m,N} \cap B_{2,N}|{\cal F}_{m,N}).$$

We first bound $P(B_{1,N}|{\cal F}_{m,N})$.  Equation (\ref{Dbound}) upper bounds the probability that fewer than $\eta c_N^2 r_N^2/(6 K_N)$ immigrants appear in the branching process.  Suppose instead there are at least $\eta c_N^2 r_N^2/(6 K_N)$ immigrants.  Recall that the offspring distribution of the branching process is binomial$(N - c_N - \lceil (\log N)^2 \rceil, {\tilde q}_N$).  Arguing as in the proof of Lemma \ref{compareGW}, the offspring distribution converges as $N \rightarrow \infty$ to the Poisson distribution with mean $1 + s$, and in particular, there is a distribution $\nu$ with mean greater than one such that the offspring stochastically dominates the distribution $\nu$ for sufficiently large $N$.  By the Kesten-Stigum Theorem, for a Galton-Watson process with offspring distribution $\nu$, there exists $\delta > 0$ and $q \in (0,1)$ such that for sufficiently large $m$, conditional on the event that the family survives forever, the probability that more than $e^{\delta m}$ descendants of this immigrant are alive in generation $m$ is at least $q$.  Because $e^{\delta (r_N/4)} \gg (\log N)^2$, it follows that for sufficiently large $N$, the probability that a particular immigrant has more than $(\log N)^2$ descendants after $\lceil r_N/4 \rceil$ generations is at least $q$.  Therefore, for sufficiently large $N$, the probability that no immigrant family has had $(\log N)^2$ descendants in some generation before generation $\lceil 3r_N/4 \rceil$ is at most
$$(1-q)^{\eta c_N^2 r_N^2/(6 K_N)} = e^{(\log (1-q)) (\eta/6) (c_N^2 r_N^2/K_N)}.$$
Combining this result with (\ref{Dbound}), we get
\begin{equation}\label{B1N}
P(B_{1,N}|{\cal F}_{m,N}) \leq e^{-(\eta/36)(c_N^2 r_N^2/K_N)} + e^{(\log (1-q))(\eta/6)(c_N^2 r_N^2/K_N)}.
\end{equation}
It follows from (\ref{Kcr}) that this expression tends to zero faster than any power of $N$.

It remains to bound $P(F_{m,N} \cap B_{2,N}|{\cal F}_{m,N})$.  Let
$$\alpha = \min\{n: A_{k,N}(m + n) \geq (\log N)^2\}.$$
Let $\beta$ and $\kappa$ be the constants from Lemma \ref{Aincrease}.  Note that $(1 + \kappa \eps_N)^{\lfloor r_N/4 \rfloor} (\log N)^2 \geq (1 - \eps_N) N$ because $\eps_N \gg (\log N)/r_N$ by (\ref{epsNdef}).  The event $B_{2,N}$ can only occur if  $\alpha \leq \lceil 3r_N/4 \rceil$ and
if, for some $i \in \{0, 1, \dots \lfloor r_N/4 \rfloor\}$, we have $(1 + \kappa \eps_N)^i (\log N)^2 \leq A_{k,N}(m + \alpha + i) < (1 + \eps_N)N$ and $A_{k,N}(m + \alpha + i + 1) < (1 + \kappa \eps_N) A_{k,N}(m + \alpha + j)$, or else for some $i$ such that $\kappa + i \leq r_N$, we have $A_{k,N}(m + \alpha + i) \geq (1 - \eps_N)N$ and $A_{k,N}(m + \alpha + i + 1) < (1 - \eps_N)N$.  Lemma~\ref{Aincrease} implies that the probability of the first of these events is at most $\sum_{i=0}^{\infty} 2 e^{-\beta \eps_N^2 (\log N)^2 (1 + \kappa \eps_N)^i},$
while Lemma~\ref{kstayhigh} implies that the probability of the second of these events is at most $2 r_N e^{-\gamma N \eps_N^2}$.
It follows that $$P(F_{m,N} \cap B_{2,N}|{\cal F}_{m,N}) \leq 2 \sum_{i=0}^{\infty} e^{-\beta \eps_N^2 (\log N)^2 (1 + \kappa \eps_N)^i} + 2 r_N e^{-\gamma N \eps_N^2}.$$
Because $(1 + \kappa \eps_N)^i \geq \kappa \eps_N i$ and $\eps_N^3 (\log N) \rightarrow \infty$ by (\ref{epsNdef}), we can bound the infinite sum by a geometric series to see
that this expression, like the expression on the right-hand side of (\ref{B1N}), tends to zero faster than any power of $N$.  The conclusion of the lemma follows.
\end{proof}

We can now state the following corollary, which summarizes how the active population evolves during class 2 intervals.  The event $G_{2,N}^o$ means that $r_N$ generations after the beginning of the first summer and $r_N$ generations after the beginning of the first winter, the population is dominated by type $0$.  The event $G_{2,N}^+$ means that $r_N$ generations after the beginning of subsequent summers, the type that dominated the population at the end of the previous summer will be dominant again.  The event $G_{2,N}^-$ is the analogous event for the winters.

\begin{Cor}\label{class2cor}
We define the following three events:
\begin{enumerate}
\item Let $G_{2,N}^+$ be the event that for all $m \in {\cal U}_N \cap {\cal T}_N^+$ such that for some $k$, the event $$\Lambda_{k, m - (U_N - V_N) - r_N, N} \cap \{A_{\ell,N}(m) = D_{\ell,N}(m) = 0 \mbox{ for all }\ell > k\}$$ occurs, the event $\Lambda_{k, m+r_N, N}$ also occurs.

\item Let $G_{2,N}^-$ be the event that for all $m \in {\cal U}_N \cap {\cal T}_N^-$ such that for some $k$, the event $$\Lambda_{k, m - V_N - r_N, N} \cap \{A_{\ell,N}(m) = D_{\ell,N}(m) = 0 \mbox{ for all }\ell < k\}$$ occurs, the event $\Lambda_{k, m+r_N, N}$ also occurs.

\item Let $G_{2,N}^o$ be the event that $\Lambda_{0,r_N,N}$ occurs and that, if $A_{\ell,N}(V_N) = D_{\ell,N}(V_N) = 0$ for all $\ell < 0$, the event $\Lambda_{0, V_N + r_N, N}$ also occurs.
\end{enumerate} 
Let $G_{2,N} = G_{2,N}^+ \cap G_{2,N}^- \cap G_{2,N}^o$.  Then $$\lim_{N \rightarrow \infty} P(G_{2,N}) = 1.$$
\end{Cor}

\begin{proof}
Lemma \ref{mutlem3} implies that with probability tending to one as $N \rightarrow \infty$, no mutation will occur within $r_N$ generations of a seasonal change.
Suppose $m \in {\cal U}_N \cap {\cal T}_N^+$ and $m > 0$.  Then for some nonnegative integer $j$, we have $m = (j+1)U_N$ and $m - (U_N - V_N) - r_N = jU_N + V_N - r_N$.  Lemma~\ref{ktoseed} implies that on the event $\Lambda_{k, jU_N + V_N - r_N, N}$, we will have $D_{k,N}(m) \geq \eta r_N c_N$ with probability at least $1 - \delta_N$ as long as there are no mutations in generations $jU_N + V_N - r_N + 1, \dots, jU_N + V_N$.  Lemma \ref{kfromseed} then implies that if this occurs and if $A_{\ell,N}(m) = D_{\ell,N}(m) = 0 \mbox{ for all }\ell > k$, then as long as there are no mutations in generations $m+1, \dots, m+r_N$, we have $A_{k,N}(m+r_N) \geq (1 - \eps_N)N$ with probability at least $1 - \delta_N^*$.  Because individuals of type higher than $k$ could only appear due to a mutation, it follows that in this case $A_{\ell,N}(m+r_N) = D_{\ell,N}(m+r_N) = 0 \mbox{ for all }\ell > k$, and therefore the event $\Lambda_{k, m+r_N,N}$ occurs.  Because the number of generations in ${\cal T}_N$ grows no faster than a power of $N$ by (\ref{totalgen}), it follows that $$\lim_{N \rightarrow \infty} P(G_{2,N}^+) = 1.$$  The same argument applied during the winter season yields $$\lim_{N \rightarrow \infty} P(G_{2,N}^-) = 1.$$

Let $F_{0,N}$ be the event that there are no mutations in generations $1, \dots, r_N$.  Note that $P(F_{0,N}) \rightarrow 1$ as $N \rightarrow \infty$ by Lemma \ref{mutlem3}.  As long as $F_{0,N}$ occurs, the event $\Lambda_{0,r_N,N}$ occurs because all individuals at time $r_N$ still have type zero.  Also, we claim that the proof of Lemma~\ref{ktoseed} implies that $P(F_{0,N} \cap \{D_{0,N}(V_N) < \eta r_N c_N\}) < \delta_N$.  To see this, note that all $K_N$ dormant individuals at time $0$ have type $0$, so one does not need to show that at least $(1 - 2 \eps_N) r_N c_N/4$ individuals of type~$0$ enter the seed bank before that time, as in the first paragraph of the proof of Lemma \ref{ktoseed}.  One can then follow the argument in the rest of the proof of Lemma \ref{ktoseed} with $K_N$ in place of $h_N$ and $V_N$ in place of $U_N - V_N$
to show that with probability tending to one, at least $K_N \rho_N/2$ of these individuals stay in the seed bank until generation $V_N$.  By \eqref{Reg1Asm1}, \eqref{Reg1Asm2}, \eqref{Reg2Asm}, and \eqref{rNdef}, we have $r_N c_N \ll U_N c_N \lesssim K_N$, so there exists $\eta > 0$ such that $$\lim_{N \rightarrow \infty} P(D_{0,N}(V_N) \geq \eta r_N c_N) = 1.$$  It then follows from Lemma \ref{kfromseed}, applied when $m = V_N$ and with the roles of summer and winter reversed, that $\Lambda_{0, V_N + r_N, N}$ occurs with probability tending to one as $N \rightarrow \infty$.  It follows that $$\lim_{N \rightarrow \infty} P(G_{2,N}^o) = 1,$$
which completes the proof of the corollary.
\end{proof}

\subsection{Proof of Theorems \ref{Poissontheo}, \ref{activethm}, and \ref{mztheorem}}\label{thmsec}

Let $t_0 > 0$, and let $$G_N = G_{0,N} \cap G_{1,N} \cap G_{2,N} \cap G_{3,N}.$$
It follows from Corollaries \ref{mutcor}, \ref{class3cor}, \ref{class1cor}, and \ref{class2cor} that
\begin{equation}\label{PGN}
\lim_{N \rightarrow \infty} P(G_N) = 1.
\end{equation}
Recall that the times in ${\cal T}_N$ are divided into intervals.  On $G_{0,N}$, we begin with a class 2 interval of length $r_N$, and then alternate between class 3 intervals having length at least $r_N$ and intervals of class 1 or 2 having length exactly $r_N$. 

The following lemma shows that on $G_N$, until some type other than type 0 dominates the population, all individuals have type 0 at the end of every class 3 interval.

\begin{Lemma}\label{allzeros}
Let $$T_N^* = \min\{n: \Lambda_{k,n,N} \mbox{ occurs for some }k \neq 0\}.$$  On $G_N$, if $m \in J_{3,N}$, $m+1 \notin J_{3,N}$, and $m < T_N^*$, then $A_{0,N}(m) = N$ and $D_{0,N}(m) = K_N$.
\end{Lemma}

\begin{proof}
On $G_{0,N}$, because there are no mutations within $r_N$ generations of a seasonal change, if all individuals have type $0$ at the beginning of a class 2 interval, then all individuals have type $0$ in the generation after the end of the class 2 interval.

On $G_{3,N}$, if all individuals have type $0$ at the beginning of a class 3 interval, then all individuals have type 0 at the end of the class 3 interval, and in the following generation if the next generation begins a class 2 interval.

On $G_{1,N}$, if all individuals have type $0$ in the generation before the start of a class 1 interval, then some type dominates the population in the generation after the end of the class 1 interval.  If this is type 0, then the more favorable type has disappeared from the population, and because no further mutations can occur during the class 1 interval, all individuals must have type 0.

Thus, by considering the intervals one at a time, we see that before time $T_N^*$, all individuals have type 0 in the first and last generation of every class 3 interval, which implies the lemma.
\end{proof}

\begin{Lemma}\label{goodevents}
On the event $G_N$, the following hold:
\begin{enumerate}
\item For every class 3 interval, one type dominates the population during the entire interval.

\item For every class 1 interval consisting of the generations $m+1, \dots, m+r_N$, some type will dominate the population in generations $m$ and $m+r_N+1$.  The same type will dominate the population in both of these generations if $\Psi^*_{m+1, N}$ does not occur.  If $\Psi^*_{m+1,N}$ occurs, then if type $k$ dominates the population in generation $m$, the type that dominates in generation $m+r_N+1$ will be type $k+1$ if $R_N(m) = 1$ and type $k-1$ if $R_N(m) = -1$.

\item Type 0 dominates the population in generations $r_N$ and $V_N + r_N$.  For every class 2 interval consisting of the generations $m, m+1, \dots, m+r_N - 1$ with $m > V_N$, the type that dominated the population in the last generation of the previous summer will dominate the population in generation $m + r_N$ if $R_N(m) = 1$, and the type that dominated the population in the last generation of the previous winter will dominate the population in generation $m + r_N$ if $R_N(m) = -1$.  
\end{enumerate}
\end{Lemma}

\begin{proof}
We consider the intervals in chronological order and prove this claim by strong induction.  We start with the class 2 interval that begins in generation zero.  On $G_{2,N}$, the event $\Lambda_{0, r_N, N}$ occurs, so type 0 dominates the population in generation $r_N$, which is the first generation of the following interval.  Now consider any other interval, and assume that the claims above hold for all previous intervals.

Suppose the interval has class 3.  Then the previous interval had class 1 or 2, and so the induction hypothesis implies that some type dominates the population in the first generation of the class 3 interval.  It then follows from the definition of $G_{3,N}$ that this type dominates the population during the entire interval.

Next, suppose the interval has class 1.  By the induction hypothesis, some type dominates the population in the last generation of the previous class 3 interval.  Then the claim follows directly from the definition of the event $G_{1,N}$.

Finally, suppose the interval has class 2.  Consider a class 2 interval that occurs at the beginning of a winter, beginning in generation $m = jU_N + V_N$ for some $j \geq 0$.  If $j \geq 1$, then by the induction hypothesis, some type, which we call type $k$, dominated the population throughout the class 3 interval that ended the previous winter.  On $G_{0,N} \cap G_{3,N}$, type $k$ dominated the population $r_N$ generations before the end of the previous winter.  That is, the event $\Lambda_{k, m - V_N - r_N}$ occurs. 
If instead $j = 0$, then type $0$ dominated the population in generation $m - V_N$, and we set $k = 0$.  In both cases, we need to show $A_{\ell, N}(m) = D_{\ell,N}(m) = 0$ for all $\ell < k$.  It will then follow from Corollary \ref{class2cor} that on $G_{2,N}$, type $k$ dominates the population in generation $m + r_N$, as claimed.  The same argument will then hold during the summer, completing the proof of the lemma.

For $h \geq 1$, on the event $\Lambda_{k, m - V_N - r_N}$, we have $A_{\ell, N}(m - V_N - r_N) = D_{\ell, N}(m - V_N - r_N) = 0$ for all $\ell < k$.  On $G_{0,N}$, no mutation can occur during the last $r_N$ generations of the previous winter, so we have $A_{\ell, N}(m - V_N) = D_{\ell, N}(m - V_N) = 0$ for all $\ell < k$.  This statement which also holds when $j = 0$, because all individuals in generation zero have type $0$.  Therefore, for $j \geq 0$, we need to show that no individual with type less than $k$ appears during the previous summer, between generations $m - V_N$ and $m$.
Such an individual could only appear when a type $k$ individual mutates to type $k-1$.  On $G_{0,N}$, no mutation can occur during class 1 or class~2 intervals.  Therefore, we only need to consider deleterious mutations during class 3 intervals.  Some type will dominate the population during these class 3 intervals.  On $G_{0,N}$, only individuals whose parent is the dominant type can acquire a mutation, so we only need to consider the case in which type $k$ is dominant during part of the previous summer.  The induction hypothesis implies that if we list in succession the types that dominated during the previous type 3 intervals in the summers, we get a nondecreasing sequence which starts at $0$ and can only increase by one when a new mutant becomes dominant in the population.  In particular, the dominant type during the winter can never be positive and the dominant type during the summer can never be negative, so we only need to consider $k = 0$.  However, if type $0$ is dominant during a class 3 interval in the previous summer as well as at the end of the previous winter (in the case $j \geq 1$), then this class 3 interval must end before the time $T_N^*$ defined in Lemma \ref{allzeros}.  By Lemma \ref{allzeros}, all individuals must have type 0 at the end of this class 3 interval, so no individual with negative type can survive until generation $m$.  Thus, we have $A_{\ell, N}(m) = D_{\ell,N}(m) = 0$ for all $\ell < k$, as claimed.
\end{proof}

We now define times $T_{k,N}^*$ for $k \in \Z$.  Let $T_{0,N}^* = 0$.  Then let
\begin{displaymath}
T_{k,N}^* = \left\{
\begin{array}{ll} \min\{m > T^*_{k-1,m}: m \in \mathcal{M}_N^+, R_N(m) = 1, \mbox{ and }\Psi^*_{m+1,N} \mbox{ occurs}\} & \mbox{ if }k > 0, \\
\min\{m > T^*_{k+1,m}: m \in \mathcal{M}_N^+, R_N(m) = -1, \mbox{ and }\Psi^*_{m+1,N} \mbox{ occurs}\} & \mbox{ if }k < 0. \end{array} \right.
\end{displaymath}

\begin{Lemma}\label{Tstar}
Suppose $a \in (0,1)$, and suppose $k$ is a nonzero integer.  Then for sufficiently large $N$, the following hold on the event $G_N$:
\begin{enumerate}
\item If $T_{k,N}^* > \rho_N^{-1} t_0$, then $S_{k,N}(a) > \rho_N^{-1} t_0$.

\item If $T_{k,N}^* \leq \rho_N^{-1} t_0 - r_N$, then $T_{k,N}^* \leq S_{k,N}(a) \leq T_{k,N}^* + r_N$.
\end{enumerate}
\end{Lemma}

\begin{proof}
We may consider the case in which $k$ is positive because the result for negative $k$ then follows in the same way.
By Lemma \ref{goodevents}, on $G_N$, some type dominates the population in the generation before and in the generation after every class 1 interval.  It also follows from Lemma~\ref{goodevents} that if we consider in order of appearance the class 1 intervals that occur during the summers, then the type that dominates in the generation before the $(j+1)$st interval will be the type that dominated in the generation after the $j$th interval.  Also, the dominant type will increase between the beginning and end of a class~1 interval which begins in generation $m$ if and only if $\Psi^*_{m+1,N}$ occurs.  It follows that type $k-1$ dominates in generation $T_{k,N}^*$, and type $k$ dominates in generation $T_{k,N}^* + r_N$, as long as these times are less than $\rho_N^{-1} t_0$.  Therefore, for sufficiently large $N$, we have $S_{k,N}(a) \leq T_{k,N}^* + r_N$ as long as $T_{k,N}^* \leq \rho_N^{-1} t_0 - r_N$.

It remains to show that, up to time $t_0$, the time $S_{k,N}(a)$ can not occur before time $T_{k,N}^*$.  The argument above shows that only types less than $k$ can dominate the population before time $T_{k,N}^*$.  Therefore, before time $T_{k,N}^*$, there can be no type $k$ individual in the population during the summers if some other type is dominating the population.  Furthermore, on $G_N$, the only mutations that occur are mutations to the dominant type during class 3 intervals.
On the event $G_{1,N}$, when a type $k-1$ individual mutates to type $k$ in generation $m$, the number of type $k$ individuals in the population stays below $L_N$ until type $k$ disappears from the population if $\Psi^*_{m+1,N}$ does not occur.  Because $L_N < a N$ for sufficiently large $N$, this means $S_{k,N}(a)$ will not occur during such an interval if $N$ is sufficiently large.

The only other way that type $k$ individuals could enter the population would be if $k = 1$ and a type $0$ individual acquires a deleterious mutation in the middle of a class 3 interval during the winter.  However, if before time $T_{1,N}^*$, type $0$ dominates the population throughout a class~3 interval during the winter, then this interval must end before the time $T_N^*$ defined in Lemma~\ref{allzeros}.  Therefore, by Lemma \ref{allzeros}, type 1 must disappear before the end of this class 3 interval.  Because type 0 dominates throughout the interval, for sufficiently large $N$ the fraction of type 1 individuals can not reach $a$ before type 1 disappears, which completes the proof.
\end{proof}

\begin{proof}[Proof of Theorem \ref{activethm}]
Let $a \in (0, 1)$.  Recall the definition of ${\dom}_{N,a}(m)$ from the statement of Theorem \ref{activethm}.  We claim that that on $G_N$, we have
\begin{equation}\label{Xkclaim}
\big|X_{k,N}(m) - \1_{\{{\dom}_{N,a}(m) = k\}}\big| \leq \eps_N \qquad \mbox{for all }m \in J_{3,N} \cap {\cal T}_N.
\end{equation}
Indeed, on $G_N$, some type dominates throughout every class 3 interval.  If we consider the class~3 intervals during the summer in chronological order, the dominant types form a nondecreasing sequence which, by Lemma \ref{Tstar}, increases by one when, in between two class 3 intervals, there is a class 1 interval during which $S_{k,N}(a)$ occurs for some $k$.  Therefore, the dominant type must be ${\dom}_{N,a}(m)$ whenever $m \in J_{3,N} \cap \mathcal{T}_N$.

Letting $|J|$ denote the cardinality of a set $J$, it follows that
\begin{equation}\label{XD1}
\rho_N \sum_{m=0}^{\lfloor \rho_N^{-1} t_0 \rfloor} \big|X_{k,N}(m) - \1_{\{{\dom}_{N,a}(m) = k\}}\big| \leq t_0 \eps_N + \rho_N |J_{1,N} \cap {\cal T}_N| + \rho_N |J_{2,N} \cap {\cal T}_N| \quad\mbox{on }G_N.
\end{equation}
We know that there is one class 2 interval at the start of every season, and each class 2 interval has length $r_N$.  Using that $r_N \ll U_N$ by (\ref{rNdef}) and that $\rho_N U_N \lesssim 1$ by \eqref{Reg1Asm2} and \eqref{Reg2Asm}, we have
\begin{equation}\label{XD2}
\rho_N |J_{2,N} \cap {\cal T}_N| \leq \rho_N \bigg( 2 + \frac{\rho_N^{-1} t_0}{U_N} \bigg) r_N \rightarrow 0
\end{equation}
as $N \rightarrow \infty$.
Also, because the probability of a mutation in each generation is bounded above by $2 \mu_N$, it follows from (\ref{UNmu}) and (\ref{rNdef}) that
$E[\rho_N|J_{1,N} \cap {\cal T}_N|] \leq \rho_N \cdot \rho_N^{-1} t_0 \cdot 2 N \mu_N \cdot r_N \rightarrow 0$
as $N \rightarrow \infty$, and therefore by Markov's Inequality, for all $\delta > 0$,
\begin{equation}\label{XD3}
\lim_{N \rightarrow \infty} P( \rho_N|J_{1,N} \cap {\cal T}_N| > \delta) = 0.
\end{equation}
Because $t_0 \eps_N \rightarrow 0$ as $N \rightarrow \infty$, the result follows from (\ref{PGN}), (\ref{XD1}), (\ref{XD2}), and (\ref{XD3}).
\end{proof}

In the proof of Theorem \ref{Poissontheo}, we will use the following simple coupling lemma.

\begin{Lemma}\label{BinPois}
Let $p \in (0, 1)$ and $\lambda > 0$.  Then, on some probability space, one can construct a random variable $X$ having a Bernoulli$(p)$ distribution and a random variable $Y$ having the Poisson$(\lambda)$ distribution such that $P(X \neq Y) \leq 2 |p - \lambda| + 2 \lambda^2$
\end{Lemma}

\begin{proof}
Standard coupling arguments imply that such a coupling can be achieved so that
\begin{align*}
P(X \neq Y) &= |P(X = 0) - P(Y = 0)| + |P(X = 1) - P(Y = 1)| + P(Y \geq 2) \nonumber \\
&= |(1 - p) - e^{-\lambda}| + |p - \lambda e^{-\lambda}| + (1 - e^{-\lambda} - \lambda e^{-\lambda}) \nonumber \\
&\leq 2|p - \lambda| + |(1 - \lambda) - e^{-\lambda}| + |\lambda - \lambda e^{-\lambda}| + (1 - e^{-\lambda} - \lambda e^{-\lambda}) \nonumber \\
&= 2 |p - \lambda| + 2 \lambda(1 - e^{-\lambda}) \nonumber \\
&\leq 2 |p - \lambda| + 2 \lambda^2,
\end{align*}
as claimed.
\end{proof}

\begin{proof}[Proof of Theorem \ref{Poissontheo}]
Fix $a \in (0, 1)$.  Lemma \ref{Tstar} and (\ref{PGN}) imply that for each $k \in \Z$,
\begin{equation*}
\lim_{N \rightarrow \infty} P\big(\{|S_{k,N}(a) - T_{k,N}^*| > r_N\} \cap \{T_{k,N}^* \leq \rho_N^{-1} t_0 - r_N\}\big) = 0.
\end{equation*}
Note that $\rho_N r_N \rightarrow 0$ as $N \rightarrow \infty$ because $r_N \ll U_N$ by (\ref{rNdef}) and $\rho_N U_N \lesssim 1$ by \eqref{Reg1Asm2} and \eqref{Reg2Asm}.
Therefore, for all $\eps > 0$, we have
\begin{equation}\label{STstar}
\lim_{N \rightarrow \infty} P\big(\{|\rho_N S_{k,N}(a) - \rho_N T_{k,N}^*| > \eps\} \cap \{\rho_N T_{k,N}^* \leq (1 - \eps)t_0 \}\big) = 0.
\end{equation}
We still need to show that the times $\rho_N T_{k,N}^*$ can be approximated by $T_{k,N}$, where $(I^+(T_{k,N}))_{k=1}^{\infty}$ and $(I^-(T_{-k,N}))_{k=1}^{\infty}$ are the event times of rate one Poisson processes.  We may focus on the $T_{k,N}$ in which $k$ is positive because the same argument applies for negative $k$. 

Let $p_N = P(\Psi^*_{m,N})$, which does not depend on $m$.  Let $(\zeta_{m,N})_{m=0}^{\infty}$ be a sequence of independent Bernoulli random variables with success probability $p_N$, which is independent of the population process.  For $m$ such that $R_N(m) = 1$, we define a Bernoulli random variable $\psi_{m,N}$.  If we have $\{m - r_N, \dots, m - 1\} \cap {\cal M}_N^+ = \emptyset$, then we set $\psi_{m,N} = 1$ if $m \in {\cal M}_N^+$ and $\Psi_{m,N}^*$ occurs.  Otherwise, we set $\psi_{m,N} = 1$ if $m \in {\cal M}_N^+$ and $\zeta_{m,N} = 1$.  The random variables $(\psi_{m,N})_{m=0}^{\infty}$ are i.i.d. Bernoulli random variables with success probability $(1 - (1 - \mu_N)^N) p_N$ because $\Psi_{m+1,N}^*$ is independent of the events $\Psi_{n,N}^*$ for $n \leq m - r_N$.  If $R_N(m) = 1$, then unless $m \in {\cal M}_N^+$ and $n \in {\cal M}_N^+$ for some $n \in \{m-1, \dots, m - r_N\}$, we will have $\psi_{m,N} = \1_{\{m \in {\cal M}_N^+\} \cap \Psi_{m+1,N}^*}$.  Therefore, by Lemma \ref{mutlem2},
\begin{equation}\label{psiPsi}
\lim_{N \rightarrow \infty} P\big(\psi_{m,N} = \1_{\{m \in {\cal M}_N^+\} \cap \Psi_{m+1,N}^*} \mbox{ for all }m \in {\cal T}_N^+ \big) = 1.
\end{equation}

We will couple the random variables $\psi_{m,N}$ with independent Poisson distributed random variables $\phi_{m,N}$.  In Regime 1, these random variables will have mean $\beta \theta/V_N$, and in Regime~2 they will have mean $\beta \rho_N U_N/V_N$.
By Lemma \ref{compareGW} and equations \eqref{slarge}, \eqref{rhoNdef}, and \eqref{Nmubound}, 
$$(1 - (1 - \mu_N)^N) p_N \sim N \mu_N p_N \sim N \mu_N \omega(s) \sim \rho_N.$$
In Regime 1, it follows from (\ref{betadef}) and (\ref{Reg1Asm2}) that
$$\frac{\beta \theta}{V_N} \sim \frac{1}{V_N} \cdot \frac{V_N}{U_N} \cdot U_N \rho_N = \rho_N.$$
In Regime 2, it follows from (\ref{betadef}) that
$$\frac{\rho_N U_N \beta}{V_N} \sim \rho_N.$$
Therefore, by Lemma \ref{BinPois}, in both regimes we can couple the random variables $\psi_{m,N}$ and $\phi_{m,N}$ in such a way that for $m$ such that $R_N(m) = 1$, we have $\rho_NP(\psi_{m,N} \neq \phi_{m,N}) \rightarrow 0.$  Therefore,
\begin{equation}\label{phipsi}
\lim_{N \rightarrow \infty} P \big( \psi_{m,N} = \phi_{m,N} \mbox{ for all }m \in {\cal T}_N^+) = 1.
\end{equation}

We now use the random variables $\phi_{m,N}$ to construct a Poisson point process.  In Regime~1, the generations $jU_N, \dots, jU_N + (V_N - 1)$ will be mapped to the time interval $[j \theta, j \theta + \beta \theta]$.  If $\phi_{jU_N + m,N} = i$ for some $m \in \{0, 1, \dots, V_N - 1\}$, we place $i$ points uniformly at random between $j \theta + m \beta \theta/V_N$ and $j \theta + (m+1) \beta \theta/V_N$.  In Regime 2, the generations $jU_N, \dots, jU_N + (V_N - 1)$ will be mapped to the time interval $[j \rho_N U_N, (j+1) \rho_N U_N]$.  
If $\phi_{jU_N + m,N} = i$, then we place $i$ points uniformly at random between $j \rho_N U_N + m \rho_N (U_N/V_N)$ and $j \rho_N U_N + (m + 1) \rho_N (U_N/V_N)$.  In both regimes, let $T_{1,N} < T_{2,N} < \dots$ be the resulting points, ranked in increasing order.  In Regime 1, these points form an inhomogeneous Poisson process whose rate during the intervals $[j \theta, j \theta + \beta \theta]$ is 1 and whose rate outside these intervals is zero.  In Regime 2, these points form a homogeneous Poisson process of rate $\beta$.  It follows from the definition of $I^+$ that in both regimes, the random times $(I^+(T_{k,N}))_{k=1}^{\infty}$ are distributed as the event times of a homogeneous rate one Poisson process.

It remains to compare the times $\rho_N T_{k,N}^*$ and $T_{k,N}$.  Equations \eqref{psiPsi} and \eqref{phipsi} imply that with probability tending to one as $N \rightarrow \infty$, we have
\begin{equation}\label{Bercouple}
\{m \in \mathcal{T}_N^+: m = T_{k,N}^* \mbox{ for some }k\} = \{m \in \mathcal{T}_N^+: \phi_{m,N} = 1\} = \{m \in \mathcal{T}_N^+: \phi_{m,N} \geq 1\}.
\end{equation}
Suppose $\phi_{jU_N + m,N} = 1$ for some $j \leq \lfloor \rho_N^{-1} t_0 \rfloor/U_N$ and some $m \in \{0, 1, \dots, V_N - 1\}$.  In Regime 1, there is a corresponding point $T_{k,N}$ in the interval $$\bigg[j \theta + \frac{m \beta \theta}{V_N}, j \theta + \frac{(m+1)\beta \theta}{V_N} \bigg].$$  In this case, we have
\begin{align*}
\big|\rho_N (jU_N + m) - T_{k,N} \big| &\leq j |\rho_N U_N - \theta| + m \bigg|\rho_N - \frac{\beta \theta}{V_N} \bigg| + \frac{\beta \theta}{V_N} \\
&\leq \frac{t_0}{U_N \rho_N} \big| \rho_N U_N - \theta| + V_N \bigg| \rho_N - \frac{\beta \theta}{V_N} \bigg| + \frac{\beta \theta}{V_N} \\
&= t_0 \bigg|1 - \frac{\theta}{\rho_N U_N} \bigg| + \big| \rho_N V_N - \beta \theta \big| + \frac{\beta \theta}{V_N},
\end{align*}
which tends to zero as $N \rightarrow \infty$ by (\ref{betadef}), (\ref{largeU}), and (\ref{Reg1Asm2}).  If we are in Regime 2, then the corresponding point $T_{k,N}$ is in the interval
$$\bigg[ j \rho_N U_N + \frac{m \rho_N U_N}{V_N}, j \rho_N U_N + \frac{(m+1) \rho_N U_N}{V_N} \bigg] \subseteq [j \rho_N U_N, (j+1) \rho_N U_N].$$  We then have
$|\rho_N (jU_N + m) - T_{k,N}| \leq \rho_N U_N$, which tends to zero as $N \rightarrow \infty$ by (\ref{Reg2Asm}).  It follows that in both regimes, when (\ref{Bercouple}) holds, the quantity $|\rho_N T_{k,N}^* - T_{k,N}|$ is bounded above by a quantity that tends to zero as $N \rightarrow \infty$ for all $k$ such that $T_{k,N}^* \leq \lfloor \rho_N^{-1} t_0 \rfloor$.
Thus, given $\eps > 0$, when (\ref{Bercouple}) holds, for sufficiently large $N$ we have $|\rho_N T_{k,N}^* - T_{k,N}|$ for all $k$ such that $T_{k,N} \leq  (1 - \eps) t_0$.
Therefore,
\begin{equation}\label{TTstar}
\lim_{N \rightarrow \infty} P\big(\big\{|\rho_N T_{k,N}^* - T_{k,N}| > \eps \big\} \cap \big\{ T_{k,N} \leq (1 - \eps) t_0 \big\}\big) = 0.
\end{equation}
Because $t_0$ is an arbitrary positive constant and $P(T_{k,N} \leq (1 - \eps) t_0) \rightarrow 1$ as $t_0 \rightarrow \infty$ for any fixed $k$, equations (\ref{STstar}) and (\ref{TTstar}) imply that $|\rho_N S_{k,N}(a) - T_{k,N}| \rightarrow_p 0$ as $N \rightarrow \infty$ for any fixed $k$, which implies Theorem \ref{Poissontheo}.
\end{proof}

It remains to prove Theorem \ref{mztheorem}, which involves working with the Meyer-Zheng topology.  Let $(S, d)$ be a metric space.  For c\`adl\`ag functions $f, g: [0, \infty) \rightarrow S$, define the metric $$r(f,g) = \int_0^{\infty} (1 \wedge d(f(t), g(t))) \, e^{-t} \: \dd t = E[1 \wedge d(f(T), g(T))],$$ where $T$ has the exponential distribution with mean 1.
As noted at the beginning of section 4 of \cite{k91}, we have $f_n \rightarrow f$ in the Meyer-Zheng topology if and only if $r(f_n,f) \rightarrow 0$.  It is a straightforward exercise to show that if $(X_n(t), t \geq 0)$ and $(Y_n(t), t \geq 0)$ are $S$-valued stochastic processes and if the Lebesgue measure of $\{t \in [0, t_0]: d(X_n(t), Y_n(t)) > \eps\}$ converges to zero in probability for all $t_0 > 0$ and all $\eps > 0$, then $r(X_n, Y_n)$ converges to zero in probability.  Therefore, if in addition we have $Y_n \Rightarrow X$ in the sense of weak convergence of stochastic processes with respect to the Meyer-Zheng topology, it follows that $X_n \Rightarrow X$ in the same sense.  We will use this fact in the proof of Theorem \ref{mztheorem} below.

\begin{proof}[Proof of Theorem \ref{mztheorem}]
Let $0 < a < 1$.  The process $(\dom_{N,a}(\lfloor \rho_N^{-1} t \rfloor), t \geq 0)$ makes jumps at the times $\rho_N S_{k,N}(a)$ and at the times $\rho_N j U_N$ and $\rho_N(jU_N + V_N)$ for nonnegative integers $j$.  The process $(\domi(t), t \geq 0)$ makes jumps at the times $T_k$ and at the times $j \theta$ and $j \theta + \beta \theta$ for nonnegative integers $j$.  Therefore, it follows from Theorem \ref{Poissontheo} and equations \eqref{betadef} and \eqref{Reg1Asm2} that the processes $(\dom_{N,a}(\lfloor \rho_N^{-1} t \rfloor), t \geq 0)$ converge as $N \rightarrow \infty$ to $(\domi(t), t \geq 0)$, in the sense of weak convergence of stochastic processes with respect to Skorohod's $J_1$-topology.  Therefore, if we define ${\idom}_{N,a}(m)$ to be the $\Delta$-valued random variable whose $k$th coordinate is $\1_{\{\dom_{N,a}(m) = k\}}$, then
\begin{equation}\label{DNDconv}
(\idom_{N,a}(\lfloor \rho_N^{-1} t \rfloor), t \geq 0) \Rightarrow ({\idom}(t), t \geq 0),
\end{equation}
with respect to Skorohod's $J_1$-topology.  It follows that the convergence in \eqref{DNDconv} also holds with respect to the weaker Meyer-Zheng topology.

Fix $t_0 > 0$.  We can rewrite the conclusion of Theorem \ref{activethm} in integral form to get that for all $k \in \Z$,
$$\int_0^{t_0} \big| X_{k,N}(\lfloor \rho_N^{-1}t \rfloor) - \1_{\{{\dom}_{N, a}(\lfloor \rho_N^{-1}t \rfloor) = k\}} \big| \: \dd t \rightarrow_p 0.$$
It follows from Theorem \ref{Poissontheo} and Lemma \ref{Tstar} that for all $\delta > 0$, there exists a positive integer $K$ such that $P(S_{k,N}(a) > t_0 \mbox{ whenever }|k| \geq K) > 1 - \delta$.  Using $\| \cdot \|$ to denote the $\ell_1$ norm on $\Delta$, it then follows that for all $\eps > 0$, the Lebesgue measure of the set 
$$\{ t \in [0,t_0]: \|{\bf X}_N(\lfloor \rho_N^{-1} t \rfloor) - \idom_{N,a}(\lfloor \rho_N^{-1} t \rfloor)\| > \eps\}$$ converges to zero in probability as $N \rightarrow \infty$.
This result, combined with \eqref{DNDconv}, implies the conclusion of the theorem.
\end{proof}

As a corollary of the arguments above, we will also prove Proposition \ref{activeProp} below, which is slightly stronger than Theorem \ref{activethm}.  This result will be used in the proofs of Theorems \ref{dormantthm1} and \ref{dormantthm2}.  Let $\eps \in (0, 1/2)$.  For $k \geq 1$, define
$$\Theta_{k,N}(\eps) = \{m \in {\cal T}_N: S_{k,N}(\eps) + r_N < m < S_{k+1,N}(\eps) \mbox{ and }R_N(\ell) = 1 \mbox{ for }m - r_N \leq \ell \leq m\}.$$
Likewise, for $k \leq -1$, define
$$\Theta_{k,N}(\eps) = \{m \in {\cal T}_N: S_{k,N}(\eps) + r_N < m < S_{k-1,N}(\eps) \mbox{ and }R_N(\ell) = 1 \mbox{ for }m - r_N \leq \ell \leq m\}.$$
Also, define
\begin{align*}
\Theta_{0,N}(\eps) &= \{m \in {\cal T}_N: m \leq S_{1,N}(\eps) \mbox{ and }R_N(\ell) = 1 \mbox{ for }m - r_N \leq \ell \leq m\} \\
&\qquad \cup \{m \in {\cal T}_N: m \leq S_{-1,N}(\eps) \mbox{ and }R_N(\ell) = -1 \mbox{ for }m - r_N \leq \ell \leq m\}.
\end{align*}

\begin{Prop}\label{activeProp}
For all $\eps \in (0, 1/2)$ and all $t_0 > 0$, we have $$\lim_{N \rightarrow \infty} P\big(X_{k,N}(m) > 1 - 2 \eps \mbox{ for all }k\mbox{ and }m\mbox{ such that }k \in \Z\mbox{ and }m \in \Theta_{k,N}(\eps)\big) = 1.$$
\end{Prop}

\begin{proof}
As noted in the proof of Theorem \ref{activethm}, for sufficiently large $N$, whenever $m$ is in a class 3 interval, the type $\dom_{N,\eps}(m)$ must dominate the population.  It follows that for sufficiently large $N$, we have $X_{k,N}(m) > 1 - 2 \eps$ for all $k$ and $m$ such that $m$ is in a class 3 interval and $m \in \Theta_{k,N}(\eps)$.  By the definition of $\Theta_{k,N}(\eps)$, we can never have $m \in \Theta_{k,N}(\eps)$ if $m$ is in a class 2 interval.  Therefore, it remains only to consider class 1 intervals.

Without loss of generality, consider a class 1 interval that occurs during the summer.  By Lemma \ref{goodevents}, on $G_N$, some type, which we call type $k$, will dominate in the previous class 3 interval.  Before the time $S_{k+1,N}(\eps)$, there can be at most $\eps N$ individuals of type $k+1$.  On $G_N$, no further mutation to type $k+2$ can happen during the class 1 interval, so there can be at most $\eps N$ individuals of type greater than $k$.  By Lemma \ref{kstayhigh}, once the number of individuals of type $k$ or higher reaches $(1 - \eps)N$, with probability tending to one as $N \rightarrow \infty$ it will never drop below $(1 - \eps)N$ until the end of the season.  In particular, with probability tending to one as $N \rightarrow \infty$, the number of individuals of type exactly $k$ must stay above $(1 - 2 \eps)N$ until either the season ends or the time $S_{k+1,N}(\eps)$ occurs.  This observation implies the result.
\end{proof}

\subsection{Proof of Propositions \ref{mutdistance} and \ref{gendistance}}\label{proofsdistance}

In this subsection, we prove Propositions \ref{mutdistance} and \ref{gendistance}, which pertain to the genetic distance and the genealogical distance between individuals sampled at random from the population. 

\begin{proof}[Proof of Proposition \ref{mutdistance}]
Fix $t > 0$.  Choose $t_0$ large enough that $\rho_N^{-1} t + U_N < \rho_N^{-1} t_0$ for sufficiently large $N$.  In Regime 1, it is enough to take $t_0 > t + \theta$, and in Regime 2, it is enough to take $t_0 > t$.  We now consider the process up to generation $\lfloor \rho_N t \rfloor$, as before.  On the event $G_N$, the last generation in every season is in a class 3 interval because the definition of $G_{0,N}$ stipulates that no mutation occurs within $2r_N$ generations of the end of a season (see Lemma \ref{mutlem3}).  In particular, on $G_N$, the generations $\chi_{N,t}^+$ and $\chi_{N,t}^-$ are both in class 3 intervals.  Therefore, by part 1 of Lemma~\ref{goodevents}, on $G_N$ one type dominates the population in these two generations.  It then follows from \eqref{PGN} that if two individuals are chosen at random from the population in generations $\chi_{N,t}^+$ or $\chi_{N,t}^-$, then these individuals will have the same type
with probability tending to one as $N \rightarrow \infty$.  This proves \eqref{mutsame}.

It follows from part 2 of Lemma \ref{goodevents} and the definition of $T_{k,N}^*$ (see also the first paragraph in the proof of Lemma \ref{Tstar}) that on $G_N$, if a class 3 interval begins in generation $m$ where $R_N(m) = 1$, then the dominant type throughout this class 3 interval is $\max\{k: T_{k,N}^* < m\}$.  Likewise, on $G_N$, if a class 3 interval begins in generation $m$ where $R_N(m) = -1$, then the dominant type throughout this class 3 interval is $\min\{k: T_{k,N}^* < m\}$.  It follows from \eqref{TTstar} in the proof of Theorem \ref{Poissontheo} that the times $(I^+(\rho_N T_{k,N}^*))_{k=1}^{\infty}$ converge in distribution as $N \rightarrow \infty$ to the event times of a homogeneous rate one Poisson process.  The same argument implies that the times $(I^-(\rho_N T_{-k,N}^*))_{k=1}^{\infty}$ also converge in distribution to the event times of a homogeneous rate one Poisson process.  It follows that if $\eps > 0$ and $n$ is a positive integer, then there is a $t^*$ such that for sufficiently large $N$, we have $\rho_N T_{n,N}^* < t^*$ and $\rho_N T_{-n,N} < t^*$ with probability greater than $1 - \eps/2$.  On the intersection of this event with $G_N$, the dominant types in the population in the two generations $\chi_{N,t}^+$ and $\chi_{N,t}^-$ will differ by at least $2n$, which implies \eqref{mutopp}.
\end{proof}

The following lemma will be useful in the proof of Proposition \ref{gendistance}.

\begin{Lemma}\label{genealogy}
Let $k \geq 1$.  For all $m$ such that $T_{k,N}^* \leq m \leq \lfloor \rho_N^{-1} t_0 \rfloor$, on the event $G_N$, every individual of type $k$ or higher in the active or dormant population in generation $m$ is descended from the individual that acquired a beneficial mutation in generation $T_{k,N}^*$, and every individual of type $-k$ or lower in the active or dormant population in generation $m$ is descended from the individual that acquired a beneficial mutation in generation $T_{-k,N}^*$.
\end{Lemma}

\begin{proof}
A type $k-1$ individual mutates to type $k$ at time $T_{k,N}^*$, which on $G_N$ is the first generation of a class 1 interval.  On $G_N$, there are no other mutations during this class 1 interval, and by part 2 of Lemma \ref{goodevents}, type $k$ dominates the population at the end of this class 1 interval.  This means that all type $k$ individuals in the population at the end of the interval are descended from the individual that got the beneficial mutation at the beginning of the interval.  From this point forward, type $k$ (or higher) dominates every generation during class 3 intervals in the summer.  As can be seen from Lemma \ref{mutlem3} and the definition of $G_{0,N}^*$ in Lemma \ref{mutlem4}, on $G_N$ the only mutations that occur are mutations to the dominant type during class 3 intervals.  Therefore, there can be no more events in which a type $k-1$ individual mutates to type $k$.  It follows that after time $T_{k,N}^*$, all individuals in the population of type $k$ and higher are descended from the individual that acquired the beneficial mutation at time $T_{k,N}^*$.  The same argument applies to individuals of type $-k$.
\end{proof}

\begin{proof}[Proof of Proposition \ref{gendistance}]
As observed in the proof of Proposition \ref{mutdistance}, on the event $G_N$, one type will dominate the population in generation $\chi_{N,t}^+$, and the dominant type
will be the largest value of $k$ for which $T_{k,N}^* < \chi_{N,t}^+$.  Therefore, by Lemma \ref{genealogy}, on $G_N$ all type $k$ individuals at time $\chi_{N,t}^+$ are descended from the individual that acquired the beneficial mutation at time $T_{k,N}^*$.
As noted in the proof of Proposition \ref{mutdistance}, the times $(I^+(\rho_N T_{k,N}^*))_{k=1}^{\infty}$ converge in distribution as $N \rightarrow \infty$ to the event times of a homogeneous rate one Poisson process.  Therefore, given $\eps > 0$, there exists $t' > 0$ such that with probability at least $1 - \eps/2$, one of the times $T_{k,N}^*$ is in the interval $[\rho_N^{-1} (t - t'), \rho_N^{-1} t]$.
The result \eqref{gensummer} follows easily from this observation, if we take $t^* > t' + \theta$ in Regime 1 and $t^* > t'$ in Regime 2.  The same argument implies \eqref{genwinter}.

Furthermore, suppose the individual $x_1(\chi_{N,t}^+)$ has type $k \geq 2$, and that this is the dominant type in the population in generation $\chi_{N,t}^+$.  Suppose also that the individual $x_1(\chi_{N,t}^-)$ has type $j \leq -2$, and that this is the dominant type in the population in generation $\chi_{N,t}^-$.
Then, on $G_N$, these individuals must be descended from the individuals that acquired beneficial mutations at the times $T_{2,N}^*$ and $T_{-2,N}^*$ respectively.  After time $T_{2,N}$, type 1 can never be dominant again in the population (although type 0 could be dominant during the winter).  Because, on $G_N$, only the dominant type can acquire a mutation by the definition of $G_{0,N}^*$ in Lemma \ref{mutlem4}, after time $T_{2,N}^*$, no individual can mutate from type 2 to type 1 and then to type 0.  Likewise, after time $T_{-2,N}^*$, no individual can mutate from type $-2$ to type $-1$ and then to type 0.  Therefore, the individuals that acquire beneficial mutations at times $T_{2,N}^*$ and $T_{-2,N}^*$ can not be descended from one another.  
It follows that on $G_N$, the most recent common ancestor of the individuals $x_1(\chi_{N,t}^+)$ and $x_1(\chi_{N,t}^-)$ must live before time $\max\{T_{2,N}^*, T_{-2,N}^*\}$.  The result \eqref{genopp} follows easily from this observation, and the fact that the times $(I^+(\rho_N T_{k,N}^*))_{k=1}^{\infty}$ and $(I^-(\rho_N T_{-k,N}^*))_{k=1}^{\infty}$ converge in distribution as $N \rightarrow \infty$ to the event times of a homogeneous rate one Poisson process.
\end{proof}

\section{Proofs of results for the dormant population}\label{proofsdormant}

In this section we prove Theorems \ref{dormantthm1} and \ref{dormantthm2}, which describe the evolution of the composition of the dormant population from generation $0$ until generation $\lfloor \rho_N^{-1} t_0\rfloor$ in Regimes 1 and 2, respectively. The analysis builds on our knowledge of the evolution of the active population in the same set of generations, and more precisely, on Proposition \ref{activeProp}.

To approximate the number of individuals of a given type $k$ at generation $n$ in the dormant population, we will count how many individuals joined the dormant population for the last time at a generation when $k$ was the dominant type in the active population; Lemma \ref{geomlem} will play a key role in this endeavor. There are three potential sources of error in this approximation:

\begin{enumerate}
\item At the times where $k$ is the dominant type in the active population, there may be a small fraction of individuals of a different type. It is therefore possible that we are considering individuals of a different type in our counting.
\item At the times where a type different from $k$ is the dominant type in the active population, there may be a small proportion of type $k$ individuals. Consequently, we may be missing some type $k$ individuals in our counting.
\item There are $2r_N$ generations around seasonal changes and selective sweeps, during which we lack information about the composition in the active population. We will not consider individuals entering the dormant population during these generations, which could lead to an undercount of type $k$ individuals.
\end{enumerate} 

In Section \ref{simplemodel} we introduce a simplified model with two types that we use to bound the first two types of errors. In this simplified model, one type will play the role of the dominant type, and the other type will model the rest of the population. In Section \ref{tragen} we will show that the third error vanishes as $N\to\infty$. In Section \ref{goodgen} we use Lemma \ref{geomlem} to determine the asymptotic behavior of our approximate counting process. In Section \ref{finaldorm} we put the pieces together and prove Theorems \ref{dormantthm1} and \ref{dormantthm2}.

To achieve this plan, we introduce some notation that will help us to partition the set of generations according to the role (dominant/dominated) of type $k$ in the active population. First, we partition the set of generations into seasons. For $j\in\Nb_0$ we denote by $I_N^+(j)$ and $I_N^-(j)$ the $j$-th summer and winter seasons, respectively, i.e. 
\begin{align*}
I_N^+(j)&\coloneqq [jU_N, jU_N+V_N)\cap\Nb_0\quad\textrm{and}\quad I_N^-(j)\coloneqq [jU_N+V_N, (j+1)U_N)\cap\Nb.
\end{align*}
We also consider the $j$-th \emph{well-established} summer and winter
\begin{align*}
I_N^{+,\circ}(j)&\coloneqq [jU_N+r_N, jU_N+V_N-r_N)\cap\Nb_0\subseteq I_N^+(j),\\  I_N^{-,\circ}(j)&\coloneqq [jU_N+V_N+r_N, (j+1)U_N-r_N)\cap\Nb\subseteq I_N^-(j),
\end{align*}
as well as the corresponding $r_N$\emph{-transition intervals}
\begin{align*}
\partial I_N^{+}(j)&\coloneqq [jU_N-r_N, jU_N+r_N)\cap\Nb_0,\\  \partial I_N^{-}(j)&\coloneqq [jU_N+V_N-r_N, jU_N+V_N+r_N)\cap\Nb.
\end{align*}
Note that the sets $I_N^{+,\circ}(j)$, $I_N^{-,\circ}(j)$, $\partial I_N^{+}(j)$, $\partial I_N^{-}(j)$, $j\in\Nb_0$, form a partition of $\Nb_0$. Next we consider the intervals between selective sweeps. For reasons that will became clear later, it is convenient to work with the times $(T_{\ell,N}^*)_{\ell\in\Zb}$ instead of $(S_{\ell,N}(a))_{\ell\in\Zb}$; their use is justified by Lemma \ref{Tstar}. Define for integers $\ell\neq 0$ the intervals
\begin{align*}
J_{\ell,N}^{\circ}&=[T^*_{\ell,N}+2r_N, T^*_{\ell+\sgn(\ell),N}-r_N]\cap \Nb_0,\\
\partial J_{\ell,N}&=(T_{\ell,N}^*-r_N, T_{\ell,N}^*+2r_N)\cap \Nb_0.
\end{align*}
According to Lemma \ref{Tstar}, on the event $G_N$, for $N$ sufficiently large, if $T_{\ell,N}^*\leq \rho_N^{-1}t_0-r_N$, then
\begin{equation}\label{JsubS}
J_{\ell,N}^{\circ}\subseteq [S_{\ell,N}(\varepsilon)+r_N, S_{\ell+\sgn(\ell),N}(\varepsilon)-r_N].
\end{equation}
Finally, we consider the set of types that have dominated the active population before generation $\lfloor\rho_N^{-1}t_0\rfloor-r_N$.  We define
$$H_{N}(t_0)\coloneqq\{p\in \Zb: \, T_{p,N}^*\leq \rho_N^{-1}t_0-r_N\}.$$

\subsection{A two-types model}\label{simplemodel}

In this section, we propose a simplified version of our model, comprising only two types of individuals: \emph{red} and \emph{blue}.  The dormant population will be updated in each generation, but the composition of the active population remains unchanged.

As in the original model, the active population consists of $N$ individuals, while the dormant population consists of $K_N$ individuals. At each generation, $c_N$ individuals are randomly selected from the active population in the previous generation to undergo duplication. The resulting duplicates replace $c_N$ individuals randomly selected from the dormant population. Consequently, the active and dormant populations maintain a constant size, and the proportion of red and blue individuals in the active population remains unchanged.  Let $\delta \in (0,1)$.  We assume that $\lfloor N \delta \rfloor$ individuals in the active population are red and the rest are blue, so that the fraction of red individuals is approximately $\delta$.
It is noteworthy that, according to Proposition~\ref{activeProp}, in the original model the active population is most of the time dominated by a single type of individuals. Consequently, if a given type is dominant (resp. dominated) in the active population, we can approximate the behavior of the dormant population by our simplified model with a value of $\delta$ close to $1$ (resp. close to $0$). Therefore, we are interested in describing the evolution, in the simplified model, of the process $\dZ$ that keeps track of the number of red individuals in the dormant population.

Note that, by the same reasoning used in Remark \ref{geomrem}, the probability that an individual in the dormant population in generation $n$ is a copy of a member of the active population in generation $n - \ell$ is
\begin{equation}\label{newswitching}
\frac{c_N}{K_N} \bigg(1 - \frac{c_N}{K_N} \bigg)^{\ell - 1}.
\end{equation}
In this case, we say that the individual switched to the dormant population in generation $n - \ell$ (even though the individual first appears in the dormant population in generation $n - \ell + 1$).

The main result in this subsection is Lemma \ref{dormantred} below, which bounds the proportion of red individuals in the dormant population when $\delta$ is small.  We first need the following simple lemma about choosing balls from urns.

\begin{Lemma}\label{urnlem}
Consider $m$ urns, each containing $a$ red balls and $b$ blue balls.  For $i \in \{1, \dots, m\}$, choose $d_i$ balls from urn $i$, and for $i \in \{1, \dots, m\}$ and $j \in \{1, \dots, d_i\}$, let $A_{i,j}$ be the event that the $j$th ball chosen from urn $i$ is red.  Then the events $A_{i,j}$ are negatively correlated.
\end{Lemma}

\begin{proof}
Considering any subcollection of the events $A_{i,j}$ still corresponds to choosing $d_i'$ balls from urn $i$ for some $d_1', \dots, d_m'$.  Therefore, it suffices to verify \eqref{negcordef} when $I$ includes all events in the collection.  Because the colors of balls chosen from different urns are independent,$$P \bigg( \bigcap_{i=1}^m \bigcap_{j=1}^{d_1} A_{i,j} \bigg) = \prod_{i=1}^m P \bigg( \bigcap_{j=1}^{d_i} A_{i,j} \bigg).$$
Furthermore, $$P \bigg( \bigcap_{j=1}^{d_i} A_{i,j} \bigg) = \frac{a}{b} \cdot \frac{a-1}{b-1} \cdot \ldots \cdot \frac{a - (d_i - 1)}{b - (d_i - 1)} \leq \bigg( \frac{a}{b} \bigg)^{d_i} = \prod_{i=1}^{d_i} P(A_{i,j}).$$
It follows that the events $A_{i,j}$ are negatively correlated.
\end{proof}

\begin{Lemma}\label{dormantred}
Suppose $\dZ(0) = 0$.  For any $\delta,\eta\in(0,1)$, we have in both regimes 
\begin{equation}\label{reddelta}
\lim_{N\to\infty} P\left(\sup_{m \leq \rho_N^{-1}t_0}\left\lvert \frac{Z_N^{(\delta)}(m)}{K_N}\right\rvert \geq (1+\eta)\,\delta\right)=0.
\end{equation}
\end{Lemma}

\begin{proof}
Label the dormant individuals $1, \dots, K_N$.  Let $S_m$ be the $c_N$-element subset of $\{1, \dots, K_N\}$ corresponding to the individuals in the dormant population that are replaced in generation $m$.  For $1 \leq k \leq K_N$, let $M_k = \max\{j \leq m: k \in S_j\}$, and let $M_k = 0$ if $k \notin S_1 \cup \dots \cup S_m$.  Then the $k$th dormant individual had its color selected at random from the active population in generation $M_k - 1$, provided that $M_k \neq 0$.  For $i \in \{1, \dots, m\}$, let $d_i$ be the cardinality of $\{k: M_k = i\}$.  Then, conditional on $S_1, \dots, S_m$, the distribution of $\dZ(m)$ is the same as the distribution of the number of red balls chosen if we have $m$ urns with $\lfloor N \delta \rfloor$ red balls and $N - \lfloor N \delta \rfloor$ blue balls and choose $d_i$ balls from the $i$th urn.  Let $Z_N$ be obtained by adding to $\dZ(m)$ an independent binomial random variable with parameters $K_N - \sum_{i=1}^m d_i$ and $\lfloor N \delta \rfloor/N$, which we can think of as arising by picking one ball from each of $K_N - \sum_{i=1}^m d_i$ additional urns, so that we select $K_N$ balls in total.  By part 1 of Lemma \ref{negcorlem} and Lemma \ref{urnlem},
$$P \big( \dZ(m) > (1 + \eta) \delta K_N \big) \leq P(Z_N > (1 + \eta) \delta K_N) \leq P(Z_N - E[Z_N] > \eta E[Z_N]) \leq e^{-\eta^2 E[Z_N]/3}.$$
Using that $E[Z_N] = K_N \lfloor N \delta \rfloor /N$, we get
$$P\left(\sup_{m \leq \rho_N^{-1}t_0}\left\lvert \frac{Z_N^{(\delta)}(m)}{K_N}\right\rvert \geq (1+\eta)\,\delta\right) \leq \rho_N^{-1} t_0 \exp \bigg(- \frac{\eta^2 \lfloor N \delta \rfloor}{3N} K_N \bigg).$$
Because $\rho_N^{-1} t_0 \ll N^{b-1}$ by \eqref{totalgen} and $K_N \gtrsim U_N c_N \gg \log N$ by \eqref{largeU}, \eqref{Reg1Asm1}, \eqref{Reg1Asm2}, and \eqref{Reg2Asm}, the result \eqref{reddelta} follows.
\end{proof}

\subsection{The contribution of the seasonal and selective sweep transitions}\label{tragen}

In this section, we will show that the number of individuals joining the dormant population at generations in the vicinity of a season change or a selective sweep vanishes as $N\to\infty$. To make this precise we define the sets
\begin{align*}
\partial I_N^+&\coloneqq\cup_{j\in\Nb_0}\partial I_N^+(j),\quad \partial I_N^-\coloneqq\cup_{j\in\Nb_0}\partial I_N^-(j),\\
\quad\partial J_N&\coloneqq\cup_{\ell\in H_N(t_0)}\partial J_{\ell,N}, \quad \partial_N\coloneqq (\rho_N^{-1}t_0-r_N,\rho_N^{-1}t_0]\cap\Nb_0, 
\end{align*}
and set
$$\partial \Gamma_{N}\coloneqq \partial I_N^+\cup \partial I_N^- \cup \partial J_{N}\cup \partial_N.$$
Following the notation in Lemma \ref{geomlem}, if $\Gamma\subset \Nb_0$, we denote by $Y_{\Gamma}(n)$ the number of individuals at generation $n$ in the dormant population that switched to the dormant population for the last time before $n$ at a generation in $\Gamma.$ 

\begin{Lemma}\label{lemtrans}
In both regimes, 
$$\sup_{n\leq \rho_N^{-1}t_0}\frac{Y_{\partial \Gamma_{N}}(n)}{K_N}\to_p 0.$$
\end{Lemma}

\begin{proof}
Note that
$$Y_{\partial \Gamma_N}(n)\leq Y_{\partial_N}(n)+ Y_{\partial I_N^+}(n)+ Y_{\partial I_N^-}(n)+ Y_{\partial J_N}(n).$$
We start by analyzing the first term. By its definition, $Y_{\partial_N}(n)$ equals $0$ for $n\leq\rho_N^{-1} t_0-r_N+1$. For
$n= \lfloor \rho_N^{-1}t_0 \rfloor-r_N+\ell$ with $2\leq\ell \leq r_N$, we have

$$\frac{E[Y_{\partial_N}(n)]}{K_N}=\sum_{j=0}^{\ell-2}\frac{c_N}{K_N}\bigg(1-\frac{c_N}{K_N}\bigg)^{j}\leq 1-\bigg(1-\frac{c_N}{K_N}\bigg)^{r_N}.$$
Since in both regimes $r_N\ll K_N/c_N$, we conclude that, as $N\to\infty$, 
\begin{equation}\label{EpN}
\sup_{n\leq \rho_N^{-1}t_0}\frac{E[Y_{\partial_{N}}(n)]}{K_N}\to 0.
\end{equation}
Moreover, using Lemma \ref{geomlem}, we have for $\varepsilon>0$ and $N$ sufficiently large that
$$P(\lvert Y_{\partial_N}(n)-E[Y_{\partial_N}(n)]\rvert > \varepsilon K_N)\leq 2e^{-\varepsilon K_N/3}.$$
Since in both regimes $r_N\ll K_N$, we conclude that
\begin{align}
P\bigg(\sup_{n\leq \rho_N^{-1}t_0} \left\lvert \frac{Y_{\partial_N}(n)}{K_N}-\frac{E[Y_{\partial_N}(n)]}{K_N}\right\rvert>\varepsilon\bigg)&\leq \sum_{n= \lfloor \rho_N^{-1}t_0 \rfloor-r_N+2}^{\lfloor \rho_N^{-1}t_0 \rfloor}P(\lvert Y_{\partial_N}(n)-E[Y_{\partial_N}(n)]\rvert > \varepsilon K_N)\nonumber\\
&\leq 2r_N e^{-\varepsilon K_N/3}\to 0,\label{pN}
\end{align}
as $N\to\infty$. Combining \eqref{EpN} and \eqref{pN}, we infer that $\sup_{n\leq \rho_N^{-1}t_0}{Y_{\partial_N}(n)}/{K_N}$ converges to $0$ in probability as $N\to\infty.$

Let us now consider the second term. Observe that the map $n\mapsto  Y_{\partial I_N^+}(n)$ is non-increasing on sets of the form $[iU_N+r_N+1, (i+1)U_N-r_N+1)$. Therefore, it suffices to consider generations of the form $n=iU_N-r_N+\ell$ for some $i\leq \rho_N^{-1}t_0/U_N$ and $1 \leq \ell \leq 2r_N$. For such an $n$, we count  the individuals entering the seedbank in $\partial I_N^+(i)$ and those that joined the seedbank in $\partial I_N^+(i-j)$ for some $1\leq j\leq i$, and obtain  
\begin{align*}
\frac{E\big[Y_{\partial I_N^+}(iU_N-r_N+\ell)\big]}{K_N}&\leq 1-\bigg(1-\frac{c_N}{K_N}\bigg)^{\ell}+\sum_{j=1}^i \bigg(1-\frac{c_N}{K_N}\bigg)^{\ell+jU_N-2r_N}\Bigg(1-\bigg(1-\frac{c_N}{K_N}\bigg)^{2r_N}\Bigg)\\
&\leq \frac{2 c_N r_N}{K_N} \cdot \frac{\big(1-\frac{c_N}{K_N}\big)^{-2r_N}}{1-\big(1-\frac{c_N}{K_N}\big)^{U_N}}.
\end{align*}
In Regime 1, the right-hand side is $O(c_Nr_N/K_N)$, and in Regime 2, the right-hand side is $O(r_N/U_N)$. In both cases, we obtain 
\begin{equation}\label{EpIN+}
\sup_{n\leq \rho_N^{-1}t_0}\frac{E[Y_{\partial I_N^+}(n)]}{K_N}\to 0.
\end{equation}
Another application of Lemma \ref{geomlem} yields for $\varepsilon>0$ and $N$ sufficiently large
$$P(\lvert Y_{\partial I_N^+}(n)-E[Y_{\partial I_N^+}(n)]\rvert > \varepsilon K_N)\leq 2e^{-\varepsilon K_N/3}.$$
Hence, 
\begin{align}
P\bigg(\sup_{n\leq \rho_N^{-1}t_0} \left\lvert \frac{Y_{\partial I_N^+}(n)}{K_N}-\frac{E[Y_{\partial I_N^+}(n)]}{K_N}\right\rvert>\varepsilon\bigg)&\leq \sum_{n=0}^{\rho_N^{-1}t_0}P(\lvert Y_{\partial I_N^+}(n)-E[Y_{\partial I_N^+}(n)]\rvert > \varepsilon K_N)\nonumber\\
&\leq t_0 \rho_N^{-1} e^{-\varepsilon K_N/3}\to 0,\label{p+N}
\end{align}
where in the last step, we used that $\rho_N^{-1}$ is $O(K_N)$ in Regime 1, and that in Regime 2, we have $K_N\gg U_N\gg\log N$ and $\rho_N^{-1}\ll N^{b-1}$ for some $b>1$.  
The last two terms can be treated analogously, but for the last term, we need to work conditionally on the times $T_{k,N}^*$.
\end{proof}

\subsection{The contribution of the good generations}\label{goodgen}

In this section, we are interested in the number of individuals of a given type $k$ that switched to the dormant population for the last time during a well-established season where $k$ was the dominant type in the active population. To make this precise, define for $k\in\Zb$ and $j\in\Nb_0$
\[
  \Ib_{k,N}(j) =
  \begin{cases}
    I_N^{+,\circ}(j)\cap J_{k,N}^\circ & \text{if $k>0$}, \\
    I_N^{-,\circ}(j)\cap J_{k,N}^\circ & \text{if $k<0$},\\ 
    [0,T_{1,N}^*-r_N]\cap I_N^{+,\circ}(j)\cup [0,T_{-1,N}^*-r_N]\cap I_N^{-,\circ}(j)& \text{if $k=0$},\\
  \end{cases}
\]
and let 
$$\Gamma_{k,N}\coloneqq \cup_{j \in \Nb_0} \Ib_{k,N}(j).$$
be the set of generations belonging to a well-established season, where type $k$ was dominating the active population.

We consider first Regime 1. As a preliminary step towards establishing the convergence result for $Y_{\Gamma_{k,N}}$, we introduce suitable analogues to the functions $d_{k,N}$ and $y_{k,N}$. Specifically, for $k\in\Zb$ and $s\ge 0$, we define $d_{k,N}^*$ via
\begin{displaymath}
d_{k,N}^*(s) = \left\{
\begin{array}{ll}  \1_{\{T_{k,N}^* \leq s < T_{k+1,N}^*\}}\1_{\{R_N(\lfloor s \rfloor)=1\}} & \mbox{ if }k \geq 1  \\
\1_{\{T_{k,N}^* \leq s < T_{k-1,N}^*\}} \1_{\{R_N(\lfloor s \rfloor)=-1\}} & \mbox{ if }k \leq -1 \\
\1_{\{t < T_{1,N}^*\}} \1_{\{R_N(\lfloor s \rfloor)=1\}} + \1_{\{t <T_{-1,N}^*\}}\1_{\{R_N(\lfloor s \rfloor)=-1\}}  & \mbox{ if }k = 0,
\end{array} \right.
\end{displaymath}
and $y_{k,N}^*$ as the solution of the differential equation
$$\frac{\dd}{\dd s}y_{k,N}^*(s)=\frac{c_N}{K_N}(d_{k,N}^*(s)-y_{k,N}^*(s)),\qquad y_{k,N}^*(0)=\1_{\{k=0\}}.$$
We start with an intermediate result dealing with  
$$\gamma_{k,N}(m)\coloneqq \frac{E_*[Y_{\Gamma_{k,N}}(m)]}{K_N},\quad m\in\Nb_0,$$
where $E_*$ denotes the conditional expectation with respect to $T_{k,N}^*$ and $T_{k+\sgn(k),N}^*$ if $k\neq 0$ and with respect to $T_{1,N}^*$ and $T_{-1,N}^*$ for $k=0$.

\begin{Lemma}\label{expectode}
For any $k\in\Zb$, almost surely
$$\sup_{t\in[0,t_0]}|y_{k,N}^*(\lfloor \rho_N^{-1}t\rfloor)-\gamma_{k,N}(\lfloor \rho_N^{-1}t\rfloor)|\xrightarrow[N\to\infty]{}0.$$
\end{Lemma}

\begin{proof}
Assume that $k>0$ (the other cases are analogous). In the proof we work on the basis of a given realization of $T_{k,N}^*,\,T_{k+1,N}^*$. For simplicity, we assume that $[T_{k,N}^*,T_{k+1,N}^*]$ do not occur in the same year.
For $n\in\Nb$, let $y_n^\downarrow$  and $y_n^\uparrow$ be the values of $y_{k,N}^*$ at the beginning of the $n$-th summer and winter, respectively, i.e.
\begin{align*} 
y_n^\downarrow&\coloneqq y_{k,N}^*(nU_N),\quad y_n^\uparrow\coloneqq y_{k,n}^*(nU_N+V_N),\quad \gamma_n^{\downarrow}\coloneqq \gamma_{k,N}(nU_N), \,\gamma_{n}^{\uparrow}\coloneqq \gamma_{k,N}(nU_N+V_N).
\end{align*}
Set also $s_0\coloneqq \lfloor T_{k,N}^*/U_N\rfloor U_N$, $w_0\coloneqq s_0+V_N$, and $s_1\coloneqq s_0+U_N$. Similarly, for $s_1\leq m\leq T_{k+1,N}^*$, we set $s_0(m)\coloneqq  \lfloor m/U_N\rfloor U_N$, $w_0(m)\coloneqq s_0(m)+V_N$, and $s_1(m)\coloneqq s_0(m)+U_N$. In particular, $s_0(m)\leq m<s_1(m)$ and $s_1(m)$ is the beginning of the first summer after $m$. 

Consider now a generation $m$ with $s_1\leq m\leq T_{k+1,N}^*$. Clearly, we have $s_1\leq s_0(m)$. Solving the differential equation defining $y_{k,N}^*$ between $s_0(m)$ and $m$ we obtain
\begin{equation}\label{ykn}
y_{k,N}^*(m) = \left\{
\begin{array}{ll}  1-e^{-\frac{c_N}{K_N}(m-s_0(m))}(1-y_{\lfloor m/U_N\rfloor}^\downarrow)& \mbox{ if } m\leq w_0(m)\\
e^{-\frac{c_N}{K_N}(m-w_0(m))}y_{\lfloor m/U_N\rfloor}^\uparrow & \mbox{ if } m\geq w_0(m).
\end{array} \right.
\end{equation}
Note that if $m$ belongs to a winter, i.e. $m\geq w_0(m)$, the $i$-th individual in the dormant population is included in the count for $Y_{\Gamma_{k,N}}(m)$ if and only if it did not switch to the dormant population at a generation in $[w_0(m),m)$ and the $i$-th individual counts for $Y_{\Gamma_{k,N}}(w_0(m))$.
Similarly, if $m$ belongs to a summer, i.e. $m< w_0(m)$, the $i$-th individual in the dormant population is included in the count for $Y_{\Gamma_{k,N}}(m)$ if and only if either it switched to the dormant population at a generation in $[s_0(m)+r_N, m\wedge (w_0(m)-r_N)]$ or it did not switch to the dormant state in $[s_0(m),m)$ and it is included in the count for $Y_{\Gamma_{k,N}}(s_0(m))$. Using this,  \eqref{newswitching}, and rearranging terms, we obtain
\begin{equation}\label{gkn}
\gamma_{k,N}(m) = \left\{
\begin{array}{ll}  1-\left(1-\frac{c_N}{K_N}\right)^{m-s_0(m)}(1-\gamma_{\lfloor m/U_N\rfloor}^\downarrow)+\epsilon_N(m)& \mbox{ if } m\leq w_0(m)\\
\left(1 - \frac{c_N}{K_N} \right)^{m-w_0(m)} \gamma_{\lfloor m/U_N\rfloor}^\uparrow & \mbox{ if } m\geq w_0(m),
\end{array} \right.
\end{equation}
with $|\epsilon_N(m)|\leq Cr_N\frac{c_N}{K_N}$ for some constant $C>0$. Proceeding in a similar way, we obtain 
\begin{align}
y_n^\downarrow&=\widehat{a}_N y_{n-1}^\uparrow,\qquad\qquad\qquad \gamma_n^{\downarrow}=a_N \gamma_{n-1}^{\uparrow}\label{1s1}\\
y_{n-1}^\uparrow&=1-\widehat{b}_N(1- y_{n-1}^\downarrow),\quad\quad\gamma_{n-1}^{\uparrow}=1-b_N(1- \gamma_{n-1}^{\downarrow})-\epsilon_N.\label{1s2}
\end{align}
where $\epsilon_N\coloneqq 1-\big(1-\frac{c_N}{K_N}\big)^{r_N} \leq c_N\,r_N/K_N$ and
\begin{align*}
 \widehat{a}_N&\coloneqq e^{-\frac{c_N(U_N-V_N)}{K_N}}, \quad \widehat{b}_N\coloneqq e^{-\frac{c_NV_N}{K_N}},\quad a_N\coloneqq\left(1-\frac{c_N}{K_N}\right)^{U_N-V_N},\, b_N\coloneqq\left(1-\frac{c_N}{K_N}\right)^{V_N},
\end{align*}
and \eqref{1s2} requires that $T_{k,N}^*\leq (n-1) U_N<(n-1) U_N+V_N\leq T_{k+1,N}^*$. For $m\geq T_{k+1,N}^*$, a similar argument shows that
\begin{equation}\label{atkp}
y_{k,N}^*(m)=e^{-\frac{c_N}{K_N}(m-T_{k+1,N}^*)}y_{k,N}^*(T_{k+1,N}^*), \quad \gamma_{k,N}(m)=\left(1-\frac{c_N}{K_N}\right)^{m-T_{k+1,N}^*}\gamma_{k,N}(T_{k+1,N}^*).
\end{equation}
Note also that
$$\sup_{m\leq \rho_{N}^{-1}t_0} \left\lvert \left(1-\frac{c_N}{K_N}\right)^{m}-e^{-\frac{c_N}{K_N}m}\right\rvert\leq \frac{c_N}{K_N}\rho_{N}^{-1}t_0 \left\lvert  1+\frac{\log(1-\frac{c_N}{K_N})}{\frac{c_N}{K_N}}\right\rvert\eqqcolon\tilde{\epsilon}_N\xrightarrow[N\to\infty]{} 0.$$
Since in addition $N\mapsto \rho_N^{-1}t_0/U_N$ is bounded and $y_{k,N}^*(m)=0=\gamma_{k,N}(m)$ for $m\leq T_{k,N}^*$, we can iterate \eqref {1s1} and \eqref{1s2}, and combine with \eqref{ykn}, \eqref{gkn} and \eqref{atkp} to show that for some positive constants $C_1$ and $C_2$,
\begin{align*}
\sup_{t\in[0,t_0]}|y_{k,N}^*(\lfloor \rho_N^{-1}t\rfloor)-\gamma_{k,N}(\lfloor \rho_N^{-1}t\rfloor)|&\leq \sup_{T_{k,N}^*\leq m\leq s_1\wedge \rho_N^{-1}t_0}|y_{k,N}^*(m)-\gamma_{k,N}(m)|
\\ &\qquad \qquad \qquad\qquad \qquad +C_1 \frac{c_N}{K_N}r_N+C_2\tilde{\epsilon}_N.
\end{align*}
Finally, for $T_{k,N}^*\leq m\leq s_1$, note that $T_{k,N}^* < w_0$ because $T_{k,N}^*$ must occur during the summer, and we have
\begin{displaymath}
y_{k,N}^*(m) = \left\{
\begin{array}{ll} 1-e^{-\frac{c_N}{K_N}(m-T_{k,N}^*)} & \mbox{if } m \leq w_0 \\
e^{-\frac{c_N}{K_N}(m-w_0)} y_{k,N}^*(w_0)& \mbox{otherwise},
\end{array} \right.
\end{displaymath}
and 
\begin{displaymath}
\gamma_{k,N}(m) = \left\{
\begin{array}{ll} 1-\left(1-\frac{c_N}{K_N}\right)^{(m-T_{k,N}^*-2r_N)_+} & \mbox{if } m \leq w_0 - r_N \\
\left(1 - \frac{c_N}{K_N} \right)^{m-(w_0 - r_N)} \left(1 -  \left(1-\frac{c_N}{K_N}\right)^{(w_0 - r_N -T_{k,N}^*-2r_N)_+} \right) &\mbox{if }w_0 - r_N < m \leq w_0 \\
\left(1-\frac{c_N}{K_N}\right)^{m-w_0} \gamma_{k,N}(w_0) & \mbox{otherwise}.
\end{array} \right.
\end{displaymath}
Therefore, there are positive constants $\tilde{C}_1$ and $\tilde{C}_2$ such that
\begin{align*}
\sup_{t\in[0,t_0]}|y_{k,N}^*(\lfloor \rho_N^{-1}t\rfloor)-\gamma_{k,N}(\lfloor \rho_N^{-1}t\rfloor)|&\leq \tilde{C}_1 \frac{c_N}{K_N}r_N+\tilde{C_2}\tilde{\epsilon}_N.
\end{align*}
The result follows letting $N\to\infty$ and using that $c_Nr_N/K_N$ and $\tilde{\epsilon}_N$ converge to $0$ as $N\to\infty$.
\end{proof} 

Now we are ready to state the convergence result in Regime 1.

\begin{Lemma}\label{dorstarreg1}
In Regime 1 we have for any $k\in\Zb$
$$\sup_{t\in[0,t_0]}\left|\frac{Y_{\Gamma_{k,N}^{}}(\lfloor \rho_N^{-1} t\rfloor)}{K_N}-y_{k,N}^*(\lfloor \rho_{N}^{-1}t \rfloor)\right|\to_p 0.$$
\end{Lemma}

\begin{proof}
Using Lemma \ref{geomlem} conditionally on $T_{k,N}^*$ and $T_{k+1,N}^*$, and then taking the expectation yields for $\varepsilon>0$ and $N$ sufficiently large
$$P(\lvert Y_{\Gamma_{k,N}}(n)-E_*[Y_{\Gamma_{k,N}}(n)]\rvert > \varepsilon K_N)\leq 2e^{-g_\varepsilon K_N/3},$$
where $g_\varepsilon>0$ is a constant only depending on $\varepsilon.$
Hence, 
\begin{align*}
P\bigg(\sup_{n\leq \rho_N^{-1}t_0} \left\lvert \frac{Y_{\Gamma_{k,N}}(n)}{K_N}-\frac{E_*[Y_{\Gamma_{k,N}}(n)]}{K_N}\right\rvert>\varepsilon\bigg)&\leq \sum_{n=0}^{\rho_N^{-1}t_0}P(\lvert Y_{\Gamma_{k,N}}(n)-E_*[Y_{\Gamma_{k,N}}(n)]\rvert > \varepsilon K_N)\nonumber\\
&\leq t_0 \rho_N^{-1} e^{-g_\varepsilon K_N/3}\to 0,
\end{align*}
where in the last step, we used that in Regime 1, $K_N\geq C\rho_N^{-1}$ for some $C>0$. The result follows combining this with Lemma \ref{expectode}.
\end{proof}

The next result compares the function $y_{k,N}^*$ to the function $y_{k,N}$ defined in section \ref{earlycouple}.

\begin{Lemma}\label{ww-star}
In Regime 1 we have for any $k\in\Zb$
$$\sup_{t\in[0,t_0]}\left|y_{k,N}(t)-y_{k,N}^*(\lfloor \rho_N^{-1} t \rfloor)\right|\to_p 0.$$
\end{Lemma}

\begin{proof}
Using the triangle inequality, we obtain  
$$\left\lvert y_{k,N}^*(\lfloor \rho_N^{-1}t\rfloor)-y_{k,N}(t)\right\rvert\leq \left\lvert y_{k,N}^*(\lfloor \rho_N^{-1}t\rfloor)-y_{k,N}^*(\rho_N^{-1}t)\right\rvert+\left\lvert y_{k,N}^*(\rho_N^{-1}t)-y_{k,N}(t)\right\rvert.$$
Using the differential equation defining $y_{k,N}^*$, one can easily see that
$$\left\lvert y_{k,N}^*(\lfloor \rho_N^{-1}t\rfloor)-y_{k,N}^*(\rho_N^{-1}t)\right\rvert\leq \frac{c_N}{K_N}.$$
Therefore, it only remains to show the uniform convergence to $0$ of $\lvert y_{k,N}^*(\rho_N^{-1}t)-y_{k,N}(t)\rvert$. Define $z_{k,N}^{}(t)\coloneqq \yk(t)-y_{k,N}^*(\rho_N^{-1}t)$. From the definitions of $\yk$ and $y_{k,N}^*$, we deduce that $z_{k,N}$ solves the non-homogeneous linear differential equation
\[\frac{\dd}{\dd t}z_{k,N}^{}(t)=-\alpha z_{k,N}^{}(t)+\alpha(d_{k,N}(t)-d_{k,N}^*(\rho_N^{-1}t))-\Big(\frac{c_N}{K_N\rho_N}-\alpha\Big)(d_{k,N}^*(\rho_N^{-1}t)-y_{k,N}^*(\rho_{N}^{-1}t)),\]
with initial condition $z_{k,N}^{}(0)=0$. Therefore, we can express $z_{k,N}^{}$ as
\begin{equation}\label{eq:zkn}
z_{k,N}(t)=e^{-\alpha t}\int_0^t e^{\alpha s}\bigg(\alpha(d_{k,N}(s)-d^*_{k,N}(\rho_N^{-1}s))-\Big(\frac{c_N}{K_N\rho_N}-\alpha\Big)(d_{k,N}^*(\rho_N^{-1}s)-y_{k,N}^*(\rho_N^{-1}s))\bigg) \dd s.
\end{equation}
Note first that for $t\leq t_0$ 
\begin{align}\label{eq:l1a}
\int_0^t \lvert d_{k,N}(s)-d_{k,N}^*(\rho_N^{-1}s)\rvert \dd s &\leq \lvert\rho_N T_{k,N}^*-T_{k,N}^{}\rvert + \lvert\rho_N T_{k+1,N}^*-T_{k+1,N}^{}\rvert\nonumber\\
&\qquad \qquad+ \int_{0}^{t_0}\lvert \ind_{\{R_N(\lfloor \rho_N^{-1} s \rfloor)=1\}}-\ind_{\{R(s)=1\}}\rvert \, \dd s.
\end{align}
Since $U_N\rho_N\to\theta$ and $V_N/U_N\to\beta$ as $N\to\infty$, the number of jumps of $R_N(\lfloor \cdot/\rho_N \rfloor)$ and $R$ in $[0,t_0]$ differs by at most $1$ and can be bounded by $2t_0/\theta+1$ for $N$ sufficiently large. Hence, we have
\[\int_{0}^{t_0}\lvert \ind_{\{R_N(s/\rho_N)=1\}}-\ind_{\{R(s)=1\}}\rvert\dd s\leq  \left(\frac{2t_0}{\theta}+1\right)^2(|U_N\rho_N-\theta|+|V_N\rho_N-\beta\theta|).\]
Plugging this bound in \eqref{eq:l1a} and combining the resulting inequality with the expression for $z_{k,N}^{}$ \eqref{eq:zkn}, we conclude that, for any $t\in[0,t_0]$,
\begin{align*}
\lvert z_{k,N}^{}(t)\rvert\leq &\, \alpha \left( \lvert\rho_N T_{k,N}^*-T_{k,N}^{}\rvert + \lvert\rho_N T_{k+1,N}^*-T_{k+1,N}^{}\rvert \right)\\
&+\alpha \left(\frac{2t_0}{\theta}+1\right)^2(|U_N\rho_N-\theta|+|V_N\rho_N-\beta\theta|)+\left\lvert\frac{c_N}{K_N\rho_N}-\alpha \right\rvert t_0,
\end{align*} 
and the result follows from the assumptions on the parameters made in Regime 1 and the convergence in probability of $\rho_N T_{\ell,N}^*-T_{\ell,N}$ to $0$ for $\ell\in\{k,k+1\}$, which is shown in \eqref{TTstar}.
\end{proof}

Let us now consider Regime 2. We start by defining for any integer $k$ and $n\in\Nb$,
\begin{displaymath}
y_{k,N}^*(n) = \left\{
\begin{array}{ll} \beta \1_{\{T_{k,N}^* \leq n < T_{k+1,N}^*\}} & \mbox{ if }k \geq 1  \\
(1 - \beta) \1_{\{T_{k,N}^* \leq n < T_{k-1,N}^*\}} & \mbox{ if }k \leq -1 \\
\beta \1_{\{n < T_{1,N}^*\}} + (1 - \beta) \1_{\{n<T_{-1,N}^*\}} & \mbox{ if }k = 0.
\end{array} \right.
\end{displaymath}
The next result establishes the convergence for $Y_{\Gamma_{k,N}^{}}$ in Regime 2. 

\begin{Lemma}\label{dorstar}Let $\delta\in(0,1)$, and define $$I_{k,N}^*(\delta,t_0)\coloneqq \{m\leq \rho_{N}^{-1}t_0: |m-T_{k,N}^*|> \delta\rho_N^{-1},\quad  |m-T_{k+1,N}^*|> \delta\rho_N^{-1}\}.$$
In Regime 2 we have
$$\sup_{n\in I_{k,N}^*(\delta,t_0)}\left|\frac{Y_{\Gamma_{k,N}^{}}(n)}{K_N}-y_{k,N}^*(n)\right|\to_p 0.$$
\end{Lemma}

\begin{proof}
We will prove the result for $k>0$; the other cases are analogous. In the rest of the proof we will work conditionally on the times $T_{k,N}^*$ and $T_{k+1,N}^*$; we will use $E_*[\cdot]$ to denote the conditional expectation with respect to these variables. By definition, 
$$Y_{\Gamma_{k,N}}(n)=0\quad \textrm{for $n< T_{k,N}^*-\delta\rho_{N}^{-1}$}.$$ 
Now consider $n=T_{k,N}^*+m$ with $\delta \rho_N^{-1}< m< T_{k+1,N}^*-T_{k,N}^*-\delta\rho_{N}^{-1}$. We take $N$ sufficiently large so that $\delta\rho_N^{-1}>2r_N$.
To compute $E_*[Y_{\Gamma_{k,N}}(n)]$, consider first separately the contribution of each year in
$[T_{k,N}^*+2r_N, n]$. For this, we refer as year $j$ into the past the generations in $[(\lfloor n/U_N\rfloor-j)U_N, (\lfloor n/U_N\rfloor-j+1)U_N)$. In particular, year $0$ is the year containing generation $n$. Set
$$j_{k,N}(n)=\left\lfloor\frac{n}{U_N}\right\rfloor-\left\lfloor\frac{n-m}{U_N}\right\rfloor.$$
Thus, we have to consider the incomplete year $0$, the earlier complete years $1,\ldots,j_{k,N}(n)-1$, and the incomplete year $j_{k,N}(n)$.

The contribution of years $j=0$ and $j=j_{k,N}(n)$ can be easily upper-bounded by
$$ \sum_{\ell=1}^{V_N-2r_N} \frac{c_N}{K_N}\bigg(1-\frac{c_N}{K_N}\bigg)^{\ell-1}=1-\bigg(1-\frac{c_N}{K_N}\bigg)^{V_N-2r_N}\eqqcolon \iota_N\sim 0.$$
The contribution of year $j\in\{1,\ldots,j_{k,N}(n)-1\}$ is
\begin{align*}
&\sum_{\ell=1}^{V_N-2r_N} \frac{c_N}{K_N}\bigg(1-\frac{c_N}{K_N}\bigg)^{\ell-1}\bigg(1-\frac{c_N}{K_N}\bigg)^{U_N-V_N+U_N(j-1)+n-\lfloor n/U_N\rfloor U_N+ r_N} \\
&\hspace{3in}\eqqcolon \iota_N \bigg(1-\frac{c_N}{K_N}\bigg)^{U_N(j-1)}f_N(n),
\end{align*}
with $(1-c_N/K_N)^{2U_N}\leq f_N(n)\leq 1.$
Therefore, the total contribution of those years is
$$\iota_Nf_N(n)\frac{1-\big(1-\frac{c_N}{K_N}\big)^{U_N(j_{k,N}(n)-1)}}{ 1-\big(1-\frac{c_N}{K_N}\big)^{U_N}}.$$
Note that $U_N(j_k(n)-1)\geq\delta\rho_N^{-1} - 2U_N$. Altogether, we obtain
$$\frac{\iota_N\upsilon_N}{ 1-\big(1-\frac{c_N}{K_N}\big)^{U_N}}\leq \frac{E_*[Y_{\Gamma_{k,N}^{}}(n)]}{K_N}\leq 2\iota_N+ \frac{\iota_N}{ 1-\big(1-\frac{c_N}{K_N}\big)^{U_N}},$$
with $\upsilon_N\sim 1$. Note that the left and right bounds in the previous inequality are independent of $n\in[T_{k,N}^*+\delta\rho_N^{-1},T_{k+1,N}^*-\delta\rho_N^{-1}]$ and converge as $N\to\infty$ to $\beta$.

Moreover, for $n\geq T_{k+1,N}^{*}+\delta\rho_N^{-1}$
$$0\leq \frac{E_*[Y_{\Gamma_{k,N}^{}}(n)]}{K_N}\leq \bigg(1-\frac{c_N}{K_N}\bigg)^{\delta\rho_N^{-1}}.$$
We conclude that 
\begin{equation}\label{expdor2}
\sup_{n\in I_{k,N}^*(\delta,t_0) }\left|\frac{E_*[Y_{\Gamma_{k,N}^{}}(n)]}{K_N}-\beta 1_{\{T_{k,N}^*\leq n<T_{k+1,N}^*\}}\right|\leq \epsilon_N,
\end{equation}
where $\epsilon_N$ is a deterministic sequence of positive numbers converging to $0$.

In addition, using Lemma \ref{geomlem} conditionally on $T_{k,N}^*$ and $T_{k+1,N}^*$ and then taking the expectation yields for $\varepsilon>0$ and $N$ sufficiently large,
$$P(\lvert Y_{\Gamma_{k,N}}(n)-E_*[Y_{\Gamma_{k,N}}(n)]\rvert > \varepsilon K_N)\leq 2e^{-g_\varepsilon K_N/3},$$
where $g_\varepsilon>0$ is a constant only depending on $\varepsilon.$
Hence, 
\begin{align}
P\bigg(\sup_{n\in I_{k,N}^*(\delta,t_0)} \left\lvert \frac{Y_{\Gamma_{k,N}}(n)}{K_N}-\frac{E_*[Y_{\Gamma_{k,N}}(n)]}{K_N}\right\rvert>\varepsilon\bigg)&\leq \sum_{n=0}^{\rho_N^{-1}t_0}P(\lvert Y_{\Gamma_{k,N}}(n)-E_*[Y_{\Gamma_{k,N}}(n)]\rvert > \varepsilon K_N)\nonumber\\
&\leq t_0 \rho_N^{-1} e^{-g_\varepsilon K_N/3}\to 0,\label{cher2}
\end{align}
where in the last step, we used that in Regime 2, $K_N\gg U_N\gg\log N$ and $\rho_N^{-1}\ll N^{b-1}$ for some $b>1$.  
Combining \eqref{expdor2} and \eqref{cher2} yields the result.
\end{proof}

\subsection{Proofs of Theorems \ref{dormantthm1} and \ref{dormantthm2}}\label{finaldorm}

\begin{proof}[Proof of Theorem \ref{dormantthm1}]
We prove the result for $k\in\Nb$; the other cases are analogous. Using the triangle inequality yields 
\begin{align*}
|Y_{k,N}(\lfloor \rho_N^{-1}t\rfloor)-y_{k,N}(t)|&\leq \left\lvert Y_{k,N}(\lfloor \rho_N^{-1}t\rfloor)-\frac{Y_{\Gamma_{k,N}}(\lfloor \rho_N^{-1}t\rfloor)}{K_N}\right\rvert+\left\lvert \frac{Y_{\Gamma_{k,N}}(\lfloor \rho_N^{-1}t\rfloor)}{K_N}-y_{k,N}^*(\lfloor \rho_N^{-1}t\rfloor)\right\rvert\\
&\quad+\left\lvert y_{k,N}^*(\lfloor \rho_N^{-1}t\rfloor)-y_{k,N}(t)\right\rvert.
\end{align*}
It is therefore sufficient to show that the supremum on $[0,t_0]$ of each of the three terms on the right-hand side of the previous inequality converges to $0$ in probability. The uniform convergence of the last two terms to $0$ is a direct consequence of Lemma \ref{dorstarreg1} and Lemma \ref{ww-star}. 

Let us now consider the first term. Define for $\varepsilon>0$ the event
\begin{equation*}
G_N(\varepsilon)\coloneqq\{X_{\ell,N}(m)>1-2\varepsilon\, \textrm{for all $\ell\in\Zb$ and $m\in\Nb$ such that $m\in \Theta_{\ell,N}(\varepsilon)$} \}.
\end{equation*}
Proposition \ref{activeProp} asserts that $P(G_N(\varepsilon))\to 1$ as $N\to \infty$. Fix now $\varepsilon>0$ and let us work on the event $G_N(\varepsilon/8)$. Note first that if $m\leq n$ is such that $m\notin\partial \Gamma_N$, then $m\in\Gamma_{\ell,N}$ for some $\ell\in\Zb$ with $T_{\ell,N}^*\leq \rho_N^{-1} t_0-r_N$. Thus, thanks to \eqref{JsubS}, we infer that $X_{\ell,N}(m)\geq 1-\varepsilon/4$. We conclude that on $G_N(\varepsilon/8)$, for any $m\leq n$ with $m\notin \partial \Gamma_N$ and any $k \in \Z$, we have $X_{k,N}(m)\geq 1- \varepsilon/4$ or $X_{k,N}(m)\leq \varepsilon/4$.

Following the discussion at the beginning of Section \ref{proofsdormant}, we infer that
\begin{equation}\label{uplowbounds}
\frac{Y_{\Gamma_{k,N}}(n)}{K_N} - \frac{Z^\circ_{k,N}(n)}{K_N}\leq Y_{k,N}(n)\leq \frac{Y_{\Gamma_{k,N}}(n)}{K_N}+\frac{Y_{\partial \Gamma_N}(n)}{K_N}+ \frac{Z^*_{k,N}(n)}{K_N},
\end{equation}
where $Z_{k,N}^*(n)$ denotes the number of type $k$ individuals in the dormant population at time $n$ that entered the dormant population for the last time at a generation outside $\partial \Gamma_N \cup \Gamma_{k,N}$, and $Z_{k,N}^\circ(n)$ denotes the number of dormant individuals at time $n$ with a type different from $k$ and that entered the dormant population for the last time at a generation in $\Gamma_{k,N}$.

Artificially coloring red or blue the individuals in the population in a way that the number of red active individuals in each generation is $\lfloor N \eps/4 \rfloor$, except that
at generations outside $\partial \Gamma_N \cup \Gamma_{k,N}$ every type $k$ active individual is red, allows one to couple on $G_N(\varepsilon/8)$ the process $Z_{k,N}^*$ to a process $\tilde{Z}_{N}^{(\varepsilon/4)}$ distributed as  ${Z}_{N}^{(\varepsilon/4)}$ (from Section \ref{simplemodel}), such that 
$$Z^*_{k,N}(n)\leq \tilde{Z}_{N}^{(\varepsilon/4)}(n),\quad n\leq \rho_N^{-1}t_0.$$
An additional red/blue coloring with a constant proportion $\varepsilon/4$ of red active individuals, and such that at generations in $\Gamma_{k,N}$ every individual with a type different from $k$ in the active population is red, allows one to couple on $G_N(\varepsilon/8)$ the process $Z_{k,N}^\circ$ to a process $\hat{Z}_{N}^{(\varepsilon/4)}$ distributed as the process ${Z}_{N}^{(\varepsilon/4)}$ such that 
$$Z_{k,N}^\circ(n) \leq \hat{Z}_{N}^{(\varepsilon/4)}(n),\quad n\leq \rho_N^{-1}t_0.$$  
Altogether, we have
\begin{align*}
P\bigg(\sup_{t\in [0,t_0]}\bigg|Y_{k,N}(\lfloor\rho_N^{-1} t\rfloor)-\frac{Y_{\Gamma_{k,N}}(\lfloor\rho_N^{-1} t\rfloor))}{K_N}\bigg| >\varepsilon\bigg)
&\leq P\bigg(\sup_{n\leq\rho_N^{-1}t_0}\bigg|\frac{Y_{\partial \Gamma_{N}}(n)}{K_N}\bigg| >\frac{\varepsilon}{3}\bigg)\\
&+ 2P\bigg(\sup_{n\leq\rho_N^{-1}t_0}\bigg|\frac{Z_{N}^{(\varepsilon/4)}(n)}{K_N}\bigg| >\frac{\varepsilon}{3}\bigg)\nonumber\\
&+P(G_N(\varepsilon/8)^c).
\end{align*}
Letting $N\to\infty$ and using Lemma \ref{dormantred} and Lemma \ref{lemtrans} yields the result.
\end{proof}

\begin{proof}[Proof of Theorem \ref{dormantthm2}]

We prove the result for $k\in\Nb$; the other cases are analogous. Using the triangle inequality yields 
\begin{align*}
|Y_{k,N}(\lfloor \rho_N^{-1}t\rfloor)-y_{k,N}(t)|&\leq \left\lvert Y_{k,N}(\lfloor \rho_N^{-1}t\rfloor)-\frac{Y_{\Gamma_{k,N}}(\lfloor \rho_N^{-1}t\rfloor)}{K_N}\right\rvert+\left\lvert \frac{Y_{\Gamma_{k,N}}(\lfloor \rho_N^{-1}t\rfloor)}{K_N}-y_{k,N}^*(\lfloor \rho_N^{-1}t\rfloor)\right\rvert\\
&+\left\lvert y_{k,N}^*(\lfloor \rho_N^{-1}t\rfloor)-y_{k,N}(t)\right\rvert.
\end{align*}
Let $I_{k,N}(\delta, t_0) = \{t \in [0, t_0]: |t - T_{k,N}| > \delta \mbox{ and } |t - T_{k+1,N}| > \delta\}$.
It suffices to show that the supremum on $I_{k,N}(\delta,t_0)$ of each of the three terms on the right-hand side of the previous inequality converges to $0$ in probability. 

Let us consider first the second term. For this note that on the event $$B_{k,N}\coloneqq \{|\rho_N T_{\ell,N}^*-T_{\ell,N}|\leq \delta/3,\, \ell\in\{k,k+1\}\},$$ if $t\in I_{k,N}(\delta,t_0)$, then for $N$ sufficiently large (independent of $t$), we have $\lfloor\rho_N^{-1} t\rfloor \in I_{k,N}^*(\delta/2, t_0)$. Hence,
\begin{align}\label{term2to0}
P\bigg(\sup_{t\in I_{k,N}(\delta,t_0)}\bigg|\frac{Y_{\Gamma_{k,N}^{}}(\lfloor\rho_N^{-1} t\rfloor)}{K_N}-y_{k,N}^*(\lfloor\rho_N^{-1} t\rfloor)\bigg| >\varepsilon\bigg) &\leq P\bigg(\sup_{n\in I_{k,N}^*(\frac{\delta}{2},t_0)}\bigg|\frac{Y_{\Gamma_{k,N}^{}}(n)}{K_N}-y_{k,N}^*(n)\bigg| >\varepsilon\bigg)\nonumber\\
&\qquad+P\big(B_{k,N}^c\big).
\end{align}
Recall from \eqref{TTstar} that for every $\ell\in\Zb$, $\rho_N T_{\ell,N}^*-T_{\ell,N}$ converges in probability to $0$. This implies that $P(B_{k,N}^c)\to 0$ as $N\to\infty$. Therefore, the desired result for the second term follows by letting $N\to\infty$ in \eqref{term2to0} and using Lemma \ref{dorstar}.

Let us now show, by considering three cases, that the third term is small on the event $B_{k,N}$. Recall that if $t\in I_{k,N}(\delta,t_0)$, then for $N$ sufficiently large $\lfloor\rho_N^{-1} t\rfloor \in I_{k,N}^*(\delta/2, t_0)$. Now, if $t\leq T_{k,N}-\delta$, then $\lfloor \rho_N^{-1} t\rfloor\leq \rho_N^{-1}(T_{k,N}-\delta)\leq T_{k,N}^*-\delta\rho_{N}^{-1}/2$. If $t\geq T_{k+1,N}+\delta$, then for $N$ sufficiently large, $\lfloor \rho_N^{-1} t\rfloor\geq \rho_N^{-1}(T_{k+1,N}+\delta)-1\geq T_{k+1,N}^*+\delta\rho_{N}^{-1}/2$. Similarly, if $T_{k,N}+\delta\leq t\leq T_{k+1,N}-\delta$, then for $N$ sufficiently large, we have $T_{k,N}^*+\delta\rho_N^{-1}/2\leq \lfloor \rho_N^{-1} t\rfloor\leq T_{k+1,N}^*-\delta\rho_N^{-1}/2$. Altogether, for $N$ sufficiently large and for any $t\in I_{k,N}(\delta,t_0)$, we have $y_{k,N}(t)=y_{k,N}^*(\lfloor \rho_N^{-1} t\rfloor)$. Thus, the uniform convergence in probability to $0$ of the third term follows using that $P(B_{k,N}^c)\to 0$ as $N\to\infty$.

For the first term, we proceed as in the previous cases to obtain
\begin{align*}
P\bigg(\sup_{t\in I_{k,N}(\delta,t_0)}\bigg|{Y_{k,N}}(\lfloor\rho_N^{-1} t\rfloor)-\frac{Y_{\Gamma_{k,N}}(\lfloor\rho_N^{-1} t\rfloor)}{K_N}\bigg| >\varepsilon\bigg) &\leq P\bigg(\sup_{n\in I_{k,N}^*(\frac{\delta}{2},t_0)}\bigg|Y_{k,N}(n)-\frac{Y_{\Gamma_{k,N}}(n)}{K_N}\bigg| >\varepsilon\bigg)\nonumber\\
&\qquad+P(B_{N,k}^c).
\end{align*}
Recall that $P(B_{N,k}^c)\to 0$.
The rest of the proof proceeds just like
the proof of Theorem \ref{dormantthm1}.
\end{proof}

\section*{List of Symbols}
{\small  

\begin{tabularx}{\textwidth}{>{\raggedright\arraybackslash}p{0.19\textwidth} X}
\rowcolor{LightGray}

\hline
\textbf{Symbol} & \textbf{Description} \\
\hline
$N, K_N$ & sizes of active and dormant populations \\
$U_N, V_N$ & year's length and summer's length \\
$R_N(m)$ & environment at generation $m$ \\
$c_N$ & number of individuals that switch between active and dormant populations \\
$s_N, \mu_N$ & selection and mutation parameters \\
$\omega(u)$ & survival probability of Galton-Watson process with Pois$(u)$ offspring distribution \\
$\rho_N$ & probability that a mutation causes a selective sweep \\
$\alpha, \beta, \theta$ & limiting parameters \\
$A_{k,N}(m), X_{k,N}(m)$ & number and fraction of type-$k$ active individuals at generation $m$ \\
$D_{k,N}(m), Y_{k,N}(m)$ & number and fraction of type-$k$ dormant individuals at generation $m$ \\
$S_{k,N}(a)$ & first time $X_{k,N}$ surpasses level $a$ \\
$R$ & limiting environment (Regime 1) \\
$I^+(t), I^-(t)$ & time spent up to time $t$ in positive/negative environment \\
$\mathcal{N}^+, \mathcal{N}^-$ & independent inhomogeneous Poisson processes \\
$T_{\pm k}$ & first time $\mathcal{N}^\pm$ reaches $\pm k$ \\
$T_{k,N}$ & copy of $T_k$ approximating $\rho_N S_{k,N}(a)$ \\
$\dom_{N,a}(m)$ & dominant type at generation $m$ \\
$\Lambda_{k,m,N}$ & event that type $k$ dominates population in generation $m$ \\
$\idom, \mathbf{Y}$ & limits of active and dormant populations \\
$d_{k,N}$ & indicator that $k$ mutations have spread in current season \\
$y_{k,N}$ & copy of $k$-th coordinate of limit dormant population \\
$U_{h,i,m,N}$ & r.v.s used to identify parents of active individuals in generation $m+1$ \\
$\xi_{i,m,N}$ & r.v. used to determine if individual $i$ mutates in generation $m+1$ \\
$\mathcal{S}_{i,m,N}^a, \mathcal{S}_{i,m,N}^d$ & r.v.s used to determine individuals switching between active and dormant populations \\
$Z_{i,m,N}^a, Z_{i,m,N}^d$ & type of individual with label $i$ in active/dormant population \\
$F_{h,N}(m)$ & fitness of individual with mark $h$ in generation $m$ \\
$\Theta_{h,N}(m)$ & relative fitness of individual with mark $h$ in generation $m$ \\
$Q_{i,m,N}$ & mark of parent of individual with label $i$ in generation $m+1$ \\
\end{tabularx}

\begin{tabularx}{\textwidth}{>{\raggedright\arraybackslash}p{0.15\textwidth} X}
\rowcolor{LightGray}
\hline
\textbf{Symbol} & \textbf{Description} \\
\hline
$\mathcal{G}_{m,N}$ & $\sigma$-field induced by previous r.v.s \\
$\mathcal{F}_{m,N}$ & $\sigma$-field generated by $\mathcal{G}_{0,N},\dots,\mathcal{G}_{m-1,N}$ \\
$\mathcal{T}_N$ & set of generations between $0$ and $\lfloor \rho_N^{-1} t_0\rfloor$ \\
$\mathcal{T}_N^+, \mathcal{T}_N^-$ & summer and winter generations in $\mathcal{T}_N$ \\
$r_N$ & auxiliary sequence of positive integers \\
$\mathcal{M}_N^+$ & generations when an individual acquires a beneficial mutation \\
$\mathcal{U}_N$ & generations when the season changes \\
$J_{1,N}$ & small intervals after mutations that increase fitness \\
$J_{2,N}$ & small intervals after season changes \\
$J_{3,N}$ & remaining generations \\
$L_N$ & sequence of positive numbers tending very slowly to $\infty$ \\
$\varepsilon_N$ & sequence of positive numbers tending very slowly to $0$ \\
$\mathcal{M}_N$ & generations when some active individual gets a mutation \\
$\mathcal{M}_N^*$ & generations when more than one individual gets a mutation \\
$G_{0,N}^*, G_{0,N}$ & events related to the times of mutations \\
$\Phi_{m,N}$ & event that type 0 recovers quickly after deleterious mutation \\
$G_{3,N}$ & event that active population behaves as expected in class 3 intervals \\
$\Phi_{k,m,N}$ & event related to population after beneficial mutation \\
$\tau_{k,m,N}$ & first time after $m$ that $A_{k,N}$ dies out or goes above $L_N$ \\
$G_{1,N}$ & event that active population behaves as expected in class 1 intervals \\
$F_{m,N}$ & no mutations in generations $m+1,\dots,m+r_N$ \\
$G_{2,N}$ & event that active population behaves as expected in class 2 intervals \\
$G_N$ & intersection of $G_{0,N}, \dots, G_{3,N}$ \\
$W_{m,N}^*, W_{m,N}'$ & branching processes coupled to population process \\
$\Psi_{m,N}^*$ & $W_{m,N}^*$ reaches $L_N$ within $r_N/2$ generations \\
$\Psi_{m,N}'$ & $W_{m,N}'$ reaches $L_N$ within $r_N/2$ generations \\
$T_N^*$ & first time that some type other than type 0 is dominant \\
$T_{k,N}^*$ & non-scaled version of $T_{k,N}$ \\
$\Theta_{k,N}(\varepsilon)$ & good generations for type $k$ \\
$Z_N^{(\delta)}$ & type composition in dormant population in two-types model \\
$I_N^+(j), I_N^-(j)$ & $j$-th summer and winter intervals \\
$I_N^{+,\circ}(j), I_N^{-,\circ}(j)$ & $j$-th well-established summer and winter intervals \\
$\partial I_N^+(j), \partial I_N^-(j)$ & generations around beginning of $j$-th summer and winter \\
$J_{\ell,N}^\circ$ & times between establishment of types $\ell$ and $\ell+\sgn(\ell)$ \\
$\partial J_{\ell,N}$ & transition generations around selective sweep of type $\ell$ \\
$H_N(t_0)$ & set of relevant types \\
$\partial I_N^+, \partial I_N^-$ & transition generations around summers and winters \\
$\partial J_N$ & transition generations around selective sweeps \\
$\partial_N$ & last $r_N$ generations \\
$\partial \Gamma_N$ & transition generations \\
$\mathbb{I}_{k,N}(j)$ & $j$-th year generations in $J_{k,N}^\circ$ favorable to type $k$ \\
$\Gamma_{k,N}$ & generations where type $k$ is dominant in active population \\
$Y_{\Gamma}(n)$ & number of individuals in generation $n$ switching to dormant population during $\Gamma$ \\
$d_{k,N}^*, y_{k,N}^*$ & non-scaled versions of $d_{k,N}$ and $y_{k,N}$ \\
$\gamma_{k,N}(m)$ & conditional expectation of $Y_{\Gamma_{k,N}}(m)/K_N$ \\
\end{tabularx}

}

\bigskip
\noindent {\bf {\Large Acknowledgments}}

\bigskip
\noindent This collaboration began when all three authors were visiting the Hausdorff Research Institute for Mathematics in Bonn, which was funded by the Deutsche Forschungsgemeinschaft under Germany's Excellence Strategy -- EXC-2047/1 -- 390685813.  The authors thank the Hausdorff Research Institute, and the organizers of the Junior Trimester Program on Stochastic Modelling in the Life Science, for their support. FC was funded by the Deutsche Forschungsgemeinschaft (DFG, German Research Foundation) — Project-ID 317210226 — SFB1283.

\end{document}